\documentclass{fac}

\usepackage{graphicx}
\usepackage{amssymb}
\usepackage{amsfonts}
\usepackage{amsmath}
\usepackage{dsfont}
\usepackage{MnSymbol}
\usepackage{mathtools}
\usepackage{epsfig}
\usepackage{epstopdf}
\usepackage{tikz}
\usepackage{amsthm}
\ifprodtf \else \usepackage{latexsym}\fi

\usepackage{tikz}
\usetikzlibrary{calc,decorations.pathmorphing,shapes,graphs}

\newcounter{sarrow}

\newcounter{sarrow1}
\newcommand\xnrsquigarrow[1]{%
\stepcounter{sarrow1}%
\mathrel{\begin{tikzpicture}[baseline= {( $ (current bounding box.south) + (0,-0.5ex) $ )}]
\node[inner sep=.5ex] (\thesarrow) {$\scriptstyle #1$};
\path[draw,<-,decorate,
  decoration={zigzag,amplitude=0.7pt,segment length=1.2mm,pre=lineto,pre length=4pt}]
    (\thesarrow1.south east) -- (\thesarrow1.south west);
    $\slashedarrowfill@\relbar\relbar/$
    \end{tikzpicture}}%
}

\makeatletter
\def\slashedarrowfill@#1#2#3#4#5{%
  $\m@th\thickmuskip0mu\medmuskip\thickmuskip\thinmuskip\thickmuskip
   \relax#5#1\mkern-7mu%
   \cleaders\hbox{$#5\mkern-2mu#2\mkern-2mu$}\hfill
   \mathclap{#3}\mathclap{#2}%
   \cleaders\hbox{$#5\mkern-2mu#2\mkern-2mu$}\hfill
   \mkern-7mu#4$%
}
\def\rightslashedarrowfillb@{%
  \slashedarrowfill@\relbar\relbar/\rightarrow}
\newcommand\xnrightarrow[2][]{%
  \ext@arrow 0055{\rightslashedarrowfillb@}{#1}{#2}}

\def\rightslashedarrowfille@{%
  \slashedarrowfill@\relbar\relbar/\twoheadrightarrow}
\newcommand\xntworightarrow[2][]{%
  \ext@arrow 0055{\rightslashedarrowfille@}{#1}{#2}}

\def\rightslashedarrowfillg@{%
  \slashedarrowfill@\relbar\relbar{\raisebox{.12em}{}}\twoheadrightarrow}
\newcommand\xtworightarrow[2][]{%
  \ext@arrow 0055{\rightslashedarrowfillg@}{#1}{#2}}

\def\rightslashedarrowfillx@{%
  \slashedarrowfill@\Relbar\Relbar/\rightrightarrows}
\newcommand\xnTworightarrow[2][]{%
  \ext@arrow 0055{\rightslashedarrowfillx@}{#1}{#2}}

\def\rightslashedarrowfilly@{%
  \slashedarrowfill@\Relbar\Relbar{\raisebox{.12em}{}}\rightrightarrows}
\newcommand\xTworightarrow[2][]{%
  \ext@arrow 0055{\rightslashedarrowfilly@}{#1}{#2}}

\pgfdeclareshape{slash underlined}
{
  \inheritsavedanchors[from=rectangle] 
  \inheritanchorborder[from=rectangle]
  \inheritanchor[from=rectangle]{north}
  \inheritanchor[from=rectangle]{north west}
  \inheritanchor[from=rectangle]{north east}
  \inheritanchor[from=rectangle]{center}
  \inheritanchor[from=rectangle]{west}
  \inheritanchor[from=rectangle]{east}
  \inheritanchor[from=rectangle]{mid}
  \inheritanchor[from=rectangle]{mid west}
  \inheritanchor[from=rectangle]{mid east}
  \inheritanchor[from=rectangle]{base}
  \inheritanchor[from=rectangle]{base west}
  \inheritanchor[from=rectangle]{base east}
  \inheritanchor[from=rectangle]{south}
  \inheritanchor[from=rectangle]{south west}
  \inheritanchor[from=rectangle]{south east}
  \inheritanchorborder[from=rectangle]
  \foregroundpath{
    \southwest \pgf@xa=\pgf@x \pgf@ya=\pgf@y
    \northeast \pgf@xb=\pgf@x \pgf@yb=\pgf@y
    \pgf@xc=\pgf@xa
    \advance\pgf@xc by .5\pgf@xb
    \pgf@yc=\pgf@ya
    \advance\pgf@xc by -1.3pt
    \advance\pgf@yc by -1.8pt
    \pgfpathmoveto{\pgfqpoint{\pgf@xc}{\pgf@yc}}
    \advance\pgf@xc by  2.6pt
    \advance\pgf@yc by  3.6pt
    \pgfpathlineto{\pgfqpoint{\pgf@xc}{\pgf@yc}}
    \pgfpathmoveto{\pgfqpoint{\pgf@xa}{\pgf@ya}}
    \pgfpathlineto{\pgfqpoint{\pgf@xb}{\pgf@ya}}
 }
}
\tikzset{nomorepostaction/.code=\let\tikz@postactions\pgfutil@empty}

\makeatother

\newcommand\black{\ensuremath{\blacktriangleright}}
\newcommand\white{\ensuremath{\vartriangleright}}

  \newcommand\whbl{\white\kern-.1em--\kern-.1em\black}
  \newcommand\blwh{\black\kern-.1em--\kern-.1em\white}
  \newcommand\blbl{\black\kern-.1em--\kern-.1em\black}
  \newcommand\whwh{\white\kern-.1em--\kern-.1em\white}

\newtheorem{theorem}{Theorem}[section]
\newtheorem{definition}[theorem]{Definition}

\title[Draft of Truly Concurrent Process Algebra Is Reversible]
      {Truly Concurrent Process Algebra Is Reversible}

\author[Yong Wang]
    {Yong Wang\\
     College of Computer Science and Technology,\\
     Faculty of Information Technology,\\
     Beijing University of Technology, Beijing, China\\
     }

\correspond{Yong Wang, Pingleyuan 100, Chaoyang District, Beijing, China.
            e-mail: wangy@bjut.edu.cn}

\pubyear{2017}
\pagerange{\pageref{firstpage}--\pageref{lastpage}}

\begin{document}
\label{firstpage}

\makecorrespond

\maketitle

\begin{abstract}
Based on our previous process algebra for concurrency APTC, we prove that it is reversible with a little modifications. The reversible algebra has four parts: Basic Algebra for Reversible True Concurrency (BARTC), Algebra for Parallelism in Reversible True Concurrency (APRTC), recursion and abstraction.
\end{abstract}

\begin{keywords}
Reversible Computation; True Concurrency; Behaviorial Equivalence; Bisimilarity
\end{keywords}

\section{Introduction}{\label{int}}

Process algebra is a formal tool to capture computation, especially concurrency, such as CCS \cite{ALNC} \cite{CC} \cite{CCS} and ACP \cite{ACP}. Several years ago, we do some work on process algebra for true concurrency, such as APTC \cite{APTC} and CTC \cite{CTC}, while traditional process algebra focuses on interleaving.

Reversible calculi \cite{CR} \cite{RCCS2} \cite{TCSR} tries to describe reversible computation in the framework of process algebra. Based on CTC and APTC, we also did some work on reversible algebra called RCTC \cite{RCTC} and RAPTC \cite{RAPTC}. But the axiomatization of RAPTC is imperfect, it is sound, but not complete. The main reason is that the existence of multi choice operator makes a sound and complete axiomatization can not be established.

In this paper, we try to use alternative operator to replace multi choice operator and we get a sound and complete axiomatization for reversible computation. The main reason of using alternative operator is that when an alternative branch is forward executing, the reverse branch is also determined and other branches have no necessaries to remain. But, when a process is reversed, the other branches disappear. We call the reversible algebra using alternative operator partially reversible algebra.

This paper is organized as follows. In section \ref{bg}, we introduce some preliminaries on APTC, reversible semantics, and proof techniques. We introduce the whole sound and complete axiomatization in section \ref{bartc}, \ref{aprtc}, \ref{rec}, \ref{abs}. Finally, we conclude this paper in section \ref{con}.

\section{Backgrounds}\label{bg}

\subsection{APTC}\label{aptc}

In this subsection, we introduce the preliminaries on truly concurrent process algebra APTC \cite{APTC}, which is based on the truly concurrent bisimulation semantics. APTC has an almost perfect axiomatization to capture laws on truly concurrent bisimulation equivalence, including equational logic and truly concurrent bisimulation semantics, and also the soundness and completeness bridged between them.

APTC captures several computational properties in the form of algebraic laws, and proves the soundness and completeness modulo truly concurrent bisimulation/rooted branching truly concurrent bisimulation equivalence. These computational properties are organized in a modular way by use of the concept of conservational extension, which include the following modules, note that, every algebra are composed of constants and operators, the constants are the computational objects, while operators capture the computational properties.

\begin{enumerate}
  \item \textbf{BATC (Basic Algebras for True Concurrency)}. BATC has sequential composition $\cdot$ and alternative composition $+$ to capture causality computation and conflict. The constants are ranged over $\mathbb{E}$, the set of atomic events. The algebraic laws on $\cdot$ and $+$ are sound and complete modulo truly concurrent bisimulation equivalences, such as pomset bisimulation $\sim_p$, step bisimulation $\sim_s$, history-preserving (hp-) bisimulation $\sim_{hp}$ and hereditary history-preserving (hhp-) bisimulation $\sim_{hhp}$.
  \item \textbf{APTC (Algebra for Parallelism for True Concurrency)}. APTC uses the whole parallel operator $\between$, the parallel operator $\parallel$ to model parallelism, and the communication merge $\mid$ to model causality (communication) among different parallel branches. Since a communication may be blocked, a new constant called deadlock $\delta$ is extended to $\mathbb{E}$, and also a new unary encapsulation operator $\partial_H$ is introduced to eliminate $\delta$, which may exist in the processes. And also a conflict elimination operator $\Theta$ to eliminate conflicts existing in different parallel branches. The algebraic laws on these operators are also sound and complete modulo truly concurrent bisimulation equivalences, such as pomset bisimulation $\sim_p$, step bisimulation $\sim_s$, history-preserving (hp-) bisimulation $\sim_{hp}$. Note that, these operators in a process except the parallel operator $\parallel$ can be eliminated by deductions on the process using axioms of APTC, and eventually be steadied by $\cdot$, $+$ and $\parallel$, this is also why bisimulations are called an \emph{truly concurrent} semantics.
  \item \textbf{Recursion}. To model infinite computation, recursion is introduced into APTC. In order to obtain a sound and complete theory, guarded recursion and linear recursion are needed. The corresponding axioms are RSP (Recursive Specification Principle) and RDP (Recursive Definition Principle), RDP says the solutions of a recursive specification can represent the behaviors of the specification, while RSP says that a guarded recursive specification has only one solution, they are sound with respect to APTC with guarded recursion modulo truly concurrent bisimulation equivalences, such as pomset bisimulation $\sim_p$, step bisimulation $\sim_s$, history-preserving (hp-) bisimulation $\sim_{hp}$, and they are complete with respect to APTC with linear recursion modulo truly concurrent bisimulation equivalence, such as pomset bisimulation $\sim_p$, step bisimulation $\sim_s$, history-preserving (hp-) bisimulation $\sim_{hp}$.
  \item \textbf{Abstraction}. To abstract away internal implementations from the external behaviors, a new constant $\tau$ called silent step is added to $\mathbb{E}$, and also a new unary abstraction operator $\tau_I$ is used to rename actions in $I$ into $\tau$ (the resulted APTC with silent step and abstraction operator is called $APTC_{\tau}$). The recursive specification is adapted to guarded linear recursion to prevent infinite $\tau$-loops specifically. The axioms for $\tau$ and $\tau_I$ are sound modulo rooted branching truly concurrent bisimulation equivalences (a kind of weak truly concurrent bisimulation equivalence), such as rooted branching pomset bisimulation $\approx_p$, rooted branching step bisimulation $\approx_s$, rooted branching history-preserving (hp-) bisimulation $\approx_{hp}$. To eliminate infinite $\tau$-loops caused by $\tau_I$ and obtain the completeness, CFAR (Cluster Fair Abstraction Rule) is used to prevent infinite $\tau$-loops in a constructible way.
\end{enumerate}

APTC can be used to verify the correctness of system behaviors, by deduction on the description of the system using the axioms of APTC. Base on the modularity of APTC, it can be extended easily and elegantly. For more details, please refer to the manuscript of APTC \cite{APTC}.

\subsection{Truly Concurrent Behavioral Semantics}\label{tcbs}

The semantics of APTC is based on truly concurrent bisimulation/rooted branching truly concurrent bisimulation equivalences, and the modularity of APTC relies on the concept of conservative extension, for the conveniences, we introduce some concepts and conclusions on them.

\begin{definition}[Prime event structure with silent event]\label{PES}
Let $\Lambda$ be a fixed set of labels, ranged over $a,b,c,\cdots$ and $\tau$. A ($\Lambda$-labelled) prime event structure with silent event $\tau$ is a tuple $\mathcal{E}=\langle \mathbb{E}, \leq, \sharp, \lambda\rangle$, where $\mathbb{E}$ is a denumerable set of events, including the silent event $\tau$. Let $\hat{\mathbb{E}}=\mathbb{E}\backslash\{\tau\}$, exactly excluding $\tau$, it is obvious that $\hat{\tau^*}=\epsilon$, where $\epsilon$ is the empty event. Let $\lambda:\mathbb{E}\rightarrow\Lambda$ be a labelling function and let $\lambda(\tau)=\tau$. And $\leq$, $\sharp$ are binary relations on $\mathbb{E}$, called causality and conflict respectively, such that:

\begin{enumerate}
  \item $\leq$ is a partial order and $\lceil e \rceil = \{e'\in \mathbb{E}|e'\leq e\}$ is finite for all $e\in \mathbb{E}$. It is easy to see that $e\leq\tau^*\leq e'=e\leq\tau\leq\cdots\leq\tau\leq e'$, then $e\leq e'$.
  \item $\sharp$ is irreflexive, symmetric and hereditary with respect to $\leq$, that is, for all $e,e',e''\in \mathbb{E}$, if $e\sharp e'\leq e''$, then $e\sharp e''$.
\end{enumerate}

Then, the concepts of consistency and concurrency can be drawn from the above definition:

\begin{enumerate}
  \item $e,e'\in \mathbb{E}$ are consistent, denoted as $e\frown e'$, if $\neg(e\sharp e')$. A subset $X\subseteq \mathbb{E}$ is called consistent, if $e\frown e'$ for all $e,e'\in X$.
  \item $e,e'\in \mathbb{E}$ are concurrent, denoted as $e\parallel e'$, if $\neg(e\leq e')$, $\neg(e'\leq e)$, and $\neg(e\sharp e')$.
\end{enumerate}
\end{definition}

\begin{definition}[Configuration]
Let $\mathcal{E}$ be a PES. A (finite) configuration in $\mathcal{E}$ is a (finite) consistent subset of events $C\subseteq \mathcal{E}$, closed with respect to causality (i.e. $\lceil C\rceil=C$). The set of finite configurations of $\mathcal{E}$ is denoted by $\mathcal{C}(\mathcal{E})$. We let $\hat{C}=C\backslash\{\tau\}$.
\end{definition}

A consistent subset of $X\subseteq \mathbb{E}$ of events can be seen as a pomset. Given $X, Y\subseteq \mathbb{E}$, $\hat{X}\sim \hat{Y}$ if $\hat{X}$ and $\hat{Y}$ are isomorphic as pomsets. In the following of the paper, we say $C_1\sim C_2$, we mean $\hat{C_1}\sim\hat{C_2}$.

\begin{definition}[Pomset transitions and step]
Let $\mathcal{E}$ be a PES and let $C\in\mathcal{C}(\mathcal{E})$, and $\emptyset\neq X\subseteq \mathbb{E}$, if $C\cap X=\emptyset$ and $C'=C\cup X\in\mathcal{C}(\mathcal{E})$, then $C\xrightarrow{X} C'$ is called a pomset transition from $C$ to $C'$. When the events in $X$ are pairwise concurrent, we say that $C\xrightarrow{X}C'$ is a step.
\end{definition}

\begin{definition}[Weak pomset transitions and weak step]
Let $\mathcal{E}$ be a PES and let $C\in\mathcal{C}(\mathcal{E})$, and $\emptyset\neq X\subseteq \hat{\mathbb{E}}$, if $C\cap X=\emptyset$ and $\hat{C'}=\hat{C}\cup X\in\mathcal{C}(\mathcal{E})$, then $C\xRightarrow{X} C'$ is called a weak pomset transition from $C$ to $C'$, where we define $\xRightarrow{e}\triangleq\xrightarrow{\tau^*}\xrightarrow{e}\xrightarrow{\tau^*}$. And $\xRightarrow{X}\triangleq\xrightarrow{\tau^*}\xrightarrow{e}\xrightarrow{\tau^*}$, for every $e\in X$. When the events in $X$ are pairwise concurrent, we say that $C\xRightarrow{X}C'$ is a weak step.
\end{definition}

We will also suppose that all the PESs in this paper are image finite, that is, for any PES $\mathcal{E}$ and $C\in \mathcal{C}(\mathcal{E})$ and $a\in \Lambda$, $\{e\in \mathbb{E}|C\xrightarrow{e} C'\wedge \lambda(e)=a\}$ and $\{e\in\hat{\mathbb{E}}|C\xRightarrow{e} C'\wedge \lambda(e)=a\}$ is finite.

\begin{definition}[Pomset, step bisimulation]\label{PSB}
Let $\mathcal{E}_1$, $\mathcal{E}_2$ be PESs. A pomset bisimulation is a relation $R\subseteq\mathcal{C}(\mathcal{E}_1)\times\mathcal{C}(\mathcal{E}_2)$, such that if $(C_1,C_2)\in R$, and $C_1\xrightarrow{X_1}C_1'$ then $C_2\xrightarrow{X_2}C_2'$, with $X_1\subseteq \mathbb{E}_1$, $X_2\subseteq \mathbb{E}_2$, $X_1\sim X_2$ and $(C_1',C_2')\in R$, and vice-versa. We say that $\mathcal{E}_1$, $\mathcal{E}_2$ are pomset bisimilar, written $\mathcal{E}_1\sim_p\mathcal{E}_2$, if there exists a pomset bisimulation $R$, such that $(\emptyset,\emptyset)\in R$. By replacing pomset transitions with steps, we can get the definition of step bisimulation. When PESs $\mathcal{E}_1$ and $\mathcal{E}_2$ are step bisimilar, we write $\mathcal{E}_1\sim_s\mathcal{E}_2$.
\end{definition}

\begin{definition}[Weak pomset, step bisimulation]\label{WPSB}
Let $\mathcal{E}_1$, $\mathcal{E}_2$ be PESs. A weak pomset bisimulation is a relation $R\subseteq\mathcal{C}(\mathcal{E}_1)\times\mathcal{C}(\mathcal{E}_2)$, such that if $(C_1,C_2)\in R$, and $C_1\xRightarrow{X_1}C_1'$ then $C_2\xRightarrow{X_2}C_2'$, with $X_1\subseteq \hat{\mathbb{E}_1}$, $X_2\subseteq \hat{\mathbb{E}_2}$, $X_1\sim X_2$ and $(C_1',C_2')\in R$, and vice-versa. We say that $\mathcal{E}_1$, $\mathcal{E}_2$ are weak pomset bisimilar, written $\mathcal{E}_1\approx_p\mathcal{E}_2$, if there exists a weak pomset bisimulation $R$, such that $(\emptyset,\emptyset)\in R$. By replacing weak pomset transitions with weak steps, we can get the definition of weak step bisimulation. When PESs $\mathcal{E}_1$ and $\mathcal{E}_2$ are weak step bisimilar, we write $\mathcal{E}_1\approx_s\mathcal{E}_2$.
\end{definition}

\begin{definition}[Posetal product]
Given two PESs $\mathcal{E}_1$, $\mathcal{E}_2$, the posetal product of their configurations, denoted $\mathcal{C}(\mathcal{E}_1)\overline{\times}\mathcal{C}(\mathcal{E}_2)$, is defined as

$$\{(C_1,f,C_2)|C_1\in\mathcal{C}(\mathcal{E}_1),C_2\in\mathcal{C}(\mathcal{E}_2),f:C_1\rightarrow C_2 \textrm{ isomorphism}\}.$$

A subset $R\subseteq\mathcal{C}(\mathcal{E}_1)\overline{\times}\mathcal{C}(\mathcal{E}_2)$ is called a posetal relation. We say that $R$ is downward closed when for any $(C_1,f,C_2),(C_1',f',C_2')\in \mathcal{C}(\mathcal{E}_1)\overline{\times}\mathcal{C}(\mathcal{E}_2)$, if $(C_1,f,C_2)\subseteq (C_1',f',C_2')$ pointwise and $(C_1',f',C_2')\in R$, then $(C_1,f,C_2)\in R$.

For $f:X_1\rightarrow X_2$, we define $f[x_1\mapsto x_2]:X_1\cup\{x_1\}\rightarrow X_2\cup\{x_2\}$, $z\in X_1\cup\{x_1\}$,(1)$f[x_1\mapsto x_2](z)=
x_2$,if $z=x_1$;(2)$f[x_1\mapsto x_2](z)=f(z)$, otherwise. Where $X_1\subseteq \mathbb{E}_1$, $X_2\subseteq \mathbb{E}_2$, $x_1\in \mathbb{E}_1$, $x_2\in \mathbb{E}_2$.
\end{definition}

\begin{definition}[Weakly posetal product]
Given two PESs $\mathcal{E}_1$, $\mathcal{E}_2$, the weakly posetal product of their configurations, denoted $\mathcal{C}(\mathcal{E}_1)\overline{\times}\mathcal{C}(\mathcal{E}_2)$, is defined as

$$\{(C_1,f,C_2)|C_1\in\mathcal{C}(\mathcal{E}_1),C_2\in\mathcal{C}(\mathcal{E}_2),f:\hat{C_1}\rightarrow \hat{C_2} \textrm{ isomorphism}\}.$$

A subset $R\subseteq\mathcal{C}(\mathcal{E}_1)\overline{\times}\mathcal{C}(\mathcal{E}_2)$ is called a weakly posetal relation. We say that $R$ is downward closed when for any $(C_1,f,C_2),(C_1',f,C_2')\in \mathcal{C}(\mathcal{E}_1)\overline{\times}\mathcal{C}(\mathcal{E}_2)$, if $(C_1,f,C_2)\subseteq (C_1',f',C_2')$ pointwise and $(C_1',f',C_2')\in R$, then $(C_1,f,C_2)\in R$.

For $f:X_1\rightarrow X_2$, we define $f[x_1\mapsto x_2]:X_1\cup\{x_1\}\rightarrow X_2\cup\{x_2\}$, $z\in X_1\cup\{x_1\}$,(1)$f[x_1\mapsto x_2](z)=
x_2$,if $z=x_1$;(2)$f[x_1\mapsto x_2](z)=f(z)$, otherwise. Where $X_1\subseteq \hat{\mathbb{E}_1}$, $X_2\subseteq \hat{\mathbb{E}_2}$, $x_1\in \hat{\mathbb{E}}_1$, $x_2\in \hat{\mathbb{E}}_2$. Also, we define $f(\tau^*)=f(\tau^*)$.
\end{definition}

\begin{definition}[(Hereditary) history-preserving bisimulation]\label{HHPB}
A history-preserving (hp-) bisimulation is a posetal relation $R\subseteq\mathcal{C}(\mathcal{E}_1)\overline{\times}\mathcal{C}(\mathcal{E}_2)$ such that if $(C_1,f,C_2)\in R$, and $C_1\xrightarrow{e_1} C_1'$, then $C_2\xrightarrow{e_2} C_2'$, with $(C_1',f[e_1\mapsto e_2],C_2')\in R$, and vice-versa. $\mathcal{E}_1,\mathcal{E}_2$ are history-preserving (hp-)bisimilar and are written $\mathcal{E}_1\sim_{hp}\mathcal{E}_2$ if there exists a hp-bisimulation $R$ such that $(\emptyset,\emptyset,\emptyset)\in R$.

A hereditary history-preserving (hhp-)bisimulation is a downward closed hp-bisimulation. $\mathcal{E}_1,\mathcal{E}_2$ are hereditary history-preserving (hhp-)bisimilar and are written $\mathcal{E}_1\sim_{hhp}\mathcal{E}_2$.
\end{definition}

\begin{definition}[Weak (hereditary) history-preserving bisimulation]\label{WHHPB}
A weak history-preserving (hp-) bisimulation is a weakly posetal relation $R\subseteq\mathcal{C}(\mathcal{E}_1)\overline{\times}\mathcal{C}(\mathcal{E}_2)$ such that if $(C_1,f,C_2)\in R$, and $C_1\xRightarrow{e_1} C_1'$, then $C_2\xRightarrow{e_2} C_2'$, with $(C_1',f[e_1\mapsto e_2],C_2')\in R$, and vice-versa. $\mathcal{E}_1,\mathcal{E}_2$ are weak history-preserving (hp-)bisimilar and are written $\mathcal{E}_1\approx_{hp}\mathcal{E}_2$ if there exists a hp-bisimulation $R$ such that $(\emptyset,\emptyset,\emptyset)\in R$.

A weakly hereditary history-preserving (hhp-)bisimulation is a downward closed weak hp-bisimulation. $\mathcal{E}_1,\mathcal{E}_2$ are weakly hereditary history-preserving (hhp-)bisimilar and are written $\mathcal{E}_1\approx_{hhp}\mathcal{E}_2$.
\end{definition}

\begin{definition}[Branching pomset, step bisimulation]\label{BPSB}
Assume a special termination predicate $\downarrow$, and let $\surd$ represent a state with $\surd\downarrow$. Let $\mathcal{E}_1$, $\mathcal{E}_2$ be PESs. A branching pomset bisimulation is a relation $R\subseteq\mathcal{C}(\mathcal{E}_1)\times\mathcal{C}(\mathcal{E}_2)$, such that:
 \begin{enumerate}
   \item if $(C_1,C_2)\in R$, and $C_1\xrightarrow{X}C_1'$ then
   \begin{itemize}
     \item either $X\equiv \tau^*$, and $(C_1',C_2)\in R$;
     \item or there is a sequence of (zero or more) $\tau$-transitions $C_2\xrightarrow{\tau^*} C_2^0$, such that $(C_1,C_2^0)\in R$ and $C_2^0\xRightarrow{X}C_2'$ with $(C_1',C_2')\in R$;
   \end{itemize}
   \item if $(C_1,C_2)\in R$, and $C_2\xrightarrow{X}C_2'$ then
   \begin{itemize}
     \item either $X\equiv \tau^*$, and $(C_1,C_2')\in R$;
     \item or there is a sequence of (zero or more) $\tau$-transitions $C_1\xrightarrow{\tau^*} C_1^0$, such that $(C_1^0,C_2)\in R$ and $C_1^0\xRightarrow{X}C_1'$ with $(C_1',C_2')\in R$;
   \end{itemize}
   \item if $(C_1,C_2)\in R$ and $C_1\downarrow$, then there is a sequence of (zero or more) $\tau$-transitions $C_2\xrightarrow{\tau^*}C_2^0$ such that $(C_1,C_2^0)\in R$ and $C_2^0\downarrow$;
   \item if $(C_1,C_2)\in R$ and $C_2\downarrow$, then there is a sequence of (zero or more) $\tau$-transitions $C_1\xrightarrow{\tau^*}C_1^0$ such that $(C_1^0,C_2)\in R$ and $C_1^0\downarrow$.
 \end{enumerate}

We say that $\mathcal{E}_1$, $\mathcal{E}_2$ are branching pomset bisimilar, written $\mathcal{E}_1\approx_{bp}\mathcal{E}_2$, if there exists a branching pomset bisimulation $R$, such that $(\emptyset,\emptyset)\in R$.

By replacing pomset transitions with steps, we can get the definition of branching step bisimulation. When PESs $\mathcal{E}_1$ and $\mathcal{E}_2$ are branching step bisimilar, we write $\mathcal{E}_1\approx_{bs}\mathcal{E}_2$.
\end{definition}

\begin{definition}[Rooted branching pomset, step bisimulation]\label{RBPSB}
Assume a special termination predicate $\downarrow$, and let $\surd$ represent a state with $\surd\downarrow$. Let $\mathcal{E}_1$, $\mathcal{E}_2$ be PESs. A rooted branching pomset bisimulation is a relation $R\subseteq\mathcal{C}(\mathcal{E}_1)\times\mathcal{C}(\mathcal{E}_2)$, such that:
 \begin{enumerate}
   \item if $(C_1,C_2)\in R$, and $C_1\xrightarrow{X}C_1'$ then $C_2\xrightarrow{X}C_2'$ with $C_1'\approx_{bp}C_2'$;
   \item if $(C_1,C_2)\in R$, and $C_2\xrightarrow{X}C_2'$ then $C_1\xrightarrow{X}C_1'$ with $C_1'\approx_{bp}C_2'$;
   \item if $(C_1,C_2)\in R$ and $C_1\downarrow$, then $C_2\downarrow$;
   \item if $(C_1,C_2)\in R$ and $C_2\downarrow$, then $C_1\downarrow$.
 \end{enumerate}

We say that $\mathcal{E}_1$, $\mathcal{E}_2$ are rooted branching pomset bisimilar, written $\mathcal{E}_1\approx_{rbp}\mathcal{E}_2$, if there exists a rooted branching pomset bisimulation $R$, such that $(\emptyset,\emptyset)\in R$.

By replacing pomset transitions with steps, we can get the definition of rooted branching step bisimulation. When PESs $\mathcal{E}_1$ and $\mathcal{E}_2$ are rooted branching step bisimilar, we write $\mathcal{E}_1\approx_{rbs}\mathcal{E}_2$.
\end{definition}

\begin{definition}[Branching (hereditary) history-preserving bisimulation]\label{BHHPB}
Assume a special termination predicate $\downarrow$, and let $\surd$ represent a state with $\surd\downarrow$. A branching history-preserving (hp-) bisimulation is a weakly posetal relation $R\subseteq\mathcal{C}(\mathcal{E}_1)\overline{\times}\mathcal{C}(\mathcal{E}_2)$ such that:

 \begin{enumerate}
   \item if $(C_1,f,C_2)\in R$, and $C_1\xrightarrow{e_1}C_1'$ then
   \begin{itemize}
     \item either $e_1\equiv \tau$, and $(C_1',f[e_1\mapsto \tau],C_2)\in R$;
     \item or there is a sequence of (zero or more) $\tau$-transitions $C_2\xrightarrow{\tau^*} C_2^0$, such that $(C_1,f,C_2^0)\in R$ and $C_2^0\xrightarrow{e_2}C_2'$ with $(C_1',f[e_1\mapsto e_2],C_2')\in R$;
   \end{itemize}
   \item if $(C_1,f,C_2)\in R$, and $C_2\xrightarrow{e_2}C_2'$ then
   \begin{itemize}
     \item either $e_2\equiv \tau$, and $(C_1,f[e_2\mapsto \tau],C_2')\in R$;
     \item or there is a sequence of (zero or more) $\tau$-transitions $C_1\xrightarrow{\tau^*} C_1^0$, such that $(C_1^0,f,C_2)\in R$ and $C_1^0\xrightarrow{e_1}C_1'$ with $(C_1',f[e_2\mapsto e_1],C_2')\in R$;
   \end{itemize}
   \item if $(C_1,f,C_2)\in R$ and $C_1\downarrow$, then there is a sequence of (zero or more) $\tau$-transitions $C_2\xrightarrow{\tau^*}C_2^0$ such that $(C_1,f,C_2^0)\in R$ and $C_2^0\downarrow$;
   \item if $(C_1,f,C_2)\in R$ and $C_2\downarrow$, then there is a sequence of (zero or more) $\tau$-transitions $C_1\xrightarrow{\tau^*}C_1^0$ such that $(C_1^0,f,C_2)\in R$ and $C_1^0\downarrow$.
 \end{enumerate}

$\mathcal{E}_1,\mathcal{E}_2$ are branching history-preserving (hp-)bisimilar and are written $\mathcal{E}_1\approx_{bhp}\mathcal{E}_2$ if there exists a branching hp-bisimulation $R$ such that $(\emptyset,\emptyset,\emptyset)\in R$.

A branching hereditary history-preserving (hhp-)bisimulation is a downward closed branching hp-bisimulation. $\mathcal{E}_1,\mathcal{E}_2$ are branching hereditary history-preserving (hhp-)bisimilar and are written $\mathcal{E}_1\approx_{bhhp}\mathcal{E}_2$.
\end{definition}

\begin{definition}[Rooted branching (hereditary) history-preserving bisimulation]\label{RBHHPB}
Assume a special termination predicate $\downarrow$, and let $\surd$ represent a state with $\surd\downarrow$. A rooted branching history-preserving (hp-) bisimulation is a weakly posetal relation $R\subseteq\mathcal{C}(\mathcal{E}_1)\overline{\times}\mathcal{C}(\mathcal{E}_2)$ such that:

 \begin{enumerate}
   \item if $(C_1,f,C_2)\in R$, and $C_1\xrightarrow{e_1}C_1'$, then $C_2\xrightarrow{e_2}C_2'$ with $C_1'\approx_{bhp}C_2'$;
   \item if $(C_1,f,C_2)\in R$, and $C_2\xrightarrow{e_2}C_2'$, then $C_1\xrightarrow{e_1}C_1'$ with $C_1'\approx_{bhp}C_2'$;
   \item if $(C_1,f,C_2)\in R$ and $C_1\downarrow$, then $C_2\downarrow$;
   \item if $(C_1,f,C_2)\in R$ and $C_2\downarrow$, then $C_1\downarrow$.
 \end{enumerate}

$\mathcal{E}_1,\mathcal{E}_2$ are rooted branching history-preserving (hp-)bisimilar and are written $\mathcal{E}_1\approx_{rbhp}\mathcal{E}_2$ if there exists a rooted branching hp-bisimulation $R$ such that $(\emptyset,\emptyset,\emptyset)\in R$.

A rooted branching hereditary history-preserving (hhp-)bisimulation is a downward closed rooted branching hp-bisimulation. $\mathcal{E}_1,\mathcal{E}_2$ are rooted branching hereditary history-preserving (hhp-)bisimilar and are written $\mathcal{E}_1\approx_{rbhhp}\mathcal{E}_2$.
\end{definition}

\begin{definition}[Congruence]
Let $\Sigma$ be a signature. An equivalence relation $R$ on $\mathcal{T}(\Sigma)$ is a congruence if for each $f\in\Sigma$, if $s_i R t_i$ for $i\in\{1,\cdots,ar(f)\}$, then $f(s_1,\cdots,s_{ar(f)}) R f(t_1,\cdots,t_{ar(f)})$.
\end{definition}

\begin{definition}[Conservative extension]
Let $T_0$ and $T_1$ be TSSs (transition system specifications) over signatures $\Sigma_0$ and $\Sigma_1$, respectively. The TSS $T_0\oplus T_1$ is a conservative extension of $T_0$ if the LTSs (labeled transition systems) generated by $T_0$ and $T_0\oplus T_1$ contain exactly the same transitions $t\xrightarrow{a}t'$ and $tP$ with $t\in \mathcal{T}(\Sigma_0)$.
\end{definition}

\begin{definition}[Source-dependency]
The source-dependent variables in a transition rule of $\rho$ are defined inductively as follows: (1) all variables in the source of $\rho$ are source-dependent; (2) if $t\xrightarrow{a}t'$ is a premise of $\rho$ and all variables in $t$ are source-dependent, then all variables in $t'$ are source-dependent. A transition rule is source-dependent if all its variables are. A TSS is source-dependent if all its rules are.
\end{definition}

\begin{definition}[Freshness]
Let $T_0$ and $T_1$ be TSSs over signatures $\Sigma_0$ and $\Sigma_1$, respectively. A term in $\mathbb{T}(T_0\oplus T_1)$ is said to be fresh if it contains a function symbol from $\Sigma_1\setminus\Sigma_0$. Similarly, a transition label or predicate symbol in $T_1$ is fresh if it does not occur in $T_0$.
\end{definition}

\begin{theorem}[Conservative extension]\label{TCE}
Let $T_0$ and $T_1$ be TSSs over signatures $\Sigma_0$ and $\Sigma_1$, respectively, where $T_0$ and $T_0\oplus T_1$ are positive after reduction. Under the following conditions, $T_0\oplus T_1$ is a conservative extension of $T_0$. (1) $T_0$ is source-dependent. (2) For each $\rho\in T_1$, either the source of $\rho$ is fresh, or $\rho$ has a premise of the form $t\xrightarrow{a}t'$ or $tP$, where $t\in \mathbb{T}(\Sigma_0)$, all variables in $t$ occur in the source of $\rho$ and $t'$, $a$ or $P$ is fresh.
\end{theorem}

\subsection{Forward-reverse Truly Concurrent Bisimulations}{\label{ftcb}}

Reversible computation is based on reverse semantics \cite{RCCS2} \cite{TCSR} \cite{CR}. In this subsection, we introduce the reverse semantics for true concurrency, which are firstly introduced in our previous work on reversible process algebra \cite{RCTC} \cite{RAPTC}.

\begin{definition}[Forward-reverse (FR) pomset transitions and forward-reverse (FR) step]
Let $\mathcal{E}$ be a PES and let $C\in\mathcal{C}(\mathcal{E})$, $\emptyset\neq X\subseteq \mathbb{E}$, $\mathcal{K}\subseteq \mathbb{N}$, and $X[\mathcal{K}]$ denotes that for each $e\in X$, there is $e[m]\in X[\mathcal{K}]$ where $(m\in\mathcal{K})$, which is called the past of $e$, and we extend $\mathbb{E}$ to $\mathbb{E}\cup\tau\cup\mathbb{E}[\mathcal{K}]$. If $C\cap X[\mathcal{K}]=\emptyset$ and $C'=C\cup X[\mathcal{K}], X\in\mathcal{C}(\mathcal{E})$, then $C\xrightarrow{X} C'$ is called a forward pomset transition from $C$ to $C'$, and $C'\xtworightarrow{X[\mathcal{K}]} C$ is called a reverse pomset transition from $C'$ to $C$. When the events in $X$ are pairwise concurrent, we say that $C\xrightarrow{X}C'$ is a forward step and $C'\xtworightarrow{X[\mathcal{K}]} C$ is a reverse step.
\end{definition}

\begin{definition}[Weak forward-reverse (FR) pomset transitions and weak forward-reverse (FR) step]
Let $\mathcal{E}$ be a PES and let $C\in\mathcal{C}(\mathcal{E})$, and $\emptyset\neq X\subseteq \hat{\mathbb{E}}$, $\mathcal{K}\subseteq \mathbb{N}$, and $X[\mathcal{K}]$ denotes that for each $e\in X$, there is $e[m]\in X[\mathcal{K}]$ where $(m\in\mathcal{K})$, which is called the past of $e$. If $C\cap X[\mathcal{K}]=\emptyset$ and $\hat{C'}=\hat{C}\cup X[\mathcal{K}], X\in\mathcal{C}(\mathcal{E})$, then $C\xRightarrow{X} C'$ is called a weak forward pomset transition from $C$ to $C'$, where we define $\xRightarrow{e}\triangleq\xrightarrow{\tau^*}\xrightarrow{e}\xrightarrow{\tau^*}$ and $\xRightarrow{X}\triangleq\xrightarrow{\tau^*}\xrightarrow{e}\xrightarrow{\tau^*}$, for every $e\in X$. And $C'\xTworightarrow{X[\mathcal{K}]} C$ is called a weak reverse pomset transition from $C'$ to $C$, where we define $\xTworightarrow{e[m]}\triangleq\xtworightarrow{\tau^*}\xtworightarrow{e[m]}\xtworightarrow{\tau^*}$, $\xTworightarrow{X[\mathcal{K}]}\triangleq\xtworightarrow{\tau^*}\xtworightarrow{e[m]} \xtworightarrow{\tau^*}$, for every $e\in X$ and $m\in\mathcal{K}$. When the events in $X$ are pairwise concurrent, we say that $C\xRightarrow{X}C'$ is a weak forward step and $C'\xTworightarrow{X[\mathcal{K}]} C$ is a weak reverse step.
\end{definition}

We will also suppose that all the PESs in this paper are image finite, that is, for any PES $\mathcal{E}$ and $C\in \mathcal{C}(\mathcal{E})$, and $a\in \Lambda$, $\{e\in \mathbb{E}|C\xrightarrow{e} C'\wedge \lambda(e)=a\}$ and $\{e\in\hat{\mathbb{E}}|C\xRightarrow{e} C'\wedge \lambda(e)=a\}$, and $a\in \Lambda$, $\{e\in \mathbb{E}|C'\xtworightarrow{e[m]} C\wedge \lambda(e)=a\}$ and $\{e\in\hat{\mathbb{E}}|C'\xTworightarrow{e[m]} C\wedge \lambda(e)=a\}$ are finite.

\begin{definition}[Forward-reverse (FR) pomset, step bisimulation]\label{FRPSB}
Let $\mathcal{E}_1$, $\mathcal{E}_2$ be PESs. An FR pomset bisimulation is a relation $R\subseteq\mathcal{C}(\mathcal{E}_1)\times\mathcal{C}(\mathcal{E}_2)$, such that (1) if $(C_1,C_2)\in R$, and $C_1\xrightarrow{X_1}C_1'$ then $C_2\xrightarrow{X_2}C_2'$, with $X_1\subseteq \mathbb{E}_1$, $X_2\subseteq \mathbb{E}_2$, $X_1\sim X_2$ and $(C_1',C_2')\in R$, and vice-versa; (2) if $(C_1',C_2')\in R$, and $C_1'\xtworightarrow{X_1[\mathcal{K}_1]}C_1$ then $C_2'\xtworightarrow{X_2[\mathcal{K}_2]}C_2$, with $X_1\subseteq \mathbb{E}_1$, $X_2\subseteq \mathbb{E}_2$, $\mathcal{K}_1,\mathcal{K}_2\subseteq\mathbb{N}$, $X_1\sim X_2$ and $(C_1,C_2)\in R$, and vice-versa. We say that $\mathcal{E}_1$, $\mathcal{E}_2$ are FR pomset bisimilar, written $\mathcal{E}_1\sim_p^{fr}\mathcal{E}_2$, if there exists an FR pomset bisimulation $R$, such that $(\emptyset,\emptyset)\in R$. By replacing FR pomset transitions with FR steps, we can get the definition of FR step bisimulation. When PESs $\mathcal{E}_1$ and $\mathcal{E}_2$ are FR step bisimilar, we write $\mathcal{E}_1\sim_s^{fr}\mathcal{E}_2$.
\end{definition}

\begin{definition}[Weak forward-reverse (FR) pomset, step bisimulation]\label{FRWPSB}
Let $\mathcal{E}_1$, $\mathcal{E}_2$ be PESs. A weak FR pomset bisimulation is a relation $R\subseteq\mathcal{C}(\mathcal{E}_1)\times\mathcal{C}(\mathcal{E}_2)$, such that (1) if $(C_1,C_2)\in R$, and $C_1\xRightarrow{X_1}C_1'$ then $C_2\xRightarrow{X_2}C_2'$, with $X_1\subseteq \hat{\mathbb{E}_1}$, $X_2\subseteq \hat{\mathbb{E}_2}$, $X_1\sim X_2$ and $(C_1',C_2')\in R$, and vice-versa; (2) if $(C_1',C_2')\in R$, and $C_1'\xTworightarrow{X_1[\mathcal{K}_1]}C_1$ then $C_2'\xTworightarrow{X_2[\mathcal{K}_2]}C_2$, with $X_1\subseteq \hat{\mathbb{E}_1}$, $X_2\subseteq \hat{\mathbb{E}_2}$, $\mathcal{K}_1,\mathcal{K}_2\subseteq\mathbb{N}$, $X_1\sim X_2$ and $(C_1,C_2)\in R$, and vice-versa. We say that $\mathcal{E}_1$, $\mathcal{E}_2$ are weak FR pomset bisimilar, written $\mathcal{E}_1\approx_p^{fr}\mathcal{E}_2$, if there exists a weak FR pomset bisimulation $R$, such that $(\emptyset,\emptyset)\in R$. By replacing weak FR pomset transitions with weak FR steps, we can get the definition of weak FR step bisimulation. When PESs $\mathcal{E}_1$ and $\mathcal{E}_2$ are weak FR step bisimilar, we write $\mathcal{E}_1\approx_s^{fr}\mathcal{E}_2$.
\end{definition}

\begin{definition}[Forward-reverse (FR) (hereditary) history-preserving bisimulation]\label{FRHHPB}
An FR history-preserving (hp-) bisimulation is a posetal relation $R\subseteq\mathcal{C}(\mathcal{E}_1)\overline{\times}\mathcal{C}(\mathcal{E}_2)$ such that (1) if $(C_1,f,C_2)\in R$, and $C_1\xrightarrow{e_1} C_1'$, then $C_2\xrightarrow{e_2} C_2'$, with $(C_1',f[e_1\mapsto e_2],C_2')\in R$, and vice-versa, (2) if $(C_1',f',C_2')\in R$, and $C_1'\xtworightarrow{e_1[m]} C_1$, then $C_2'\xtworightarrow{e_2[n]} C_2$, with $(C_1,f'[e_1[m]\mapsto e_2[n]],C_2)\in R$, and vice-versa. $\mathcal{E}_1,\mathcal{E}_2$ are FR history-preserving (hp-) bisimilar and are written $\mathcal{E}_1\sim_{hp}^{fr}\mathcal{E}_2$ if there exists an FR hp-bisimulation $R$ such that $(\emptyset,\emptyset,\emptyset)\in R$.

An FR hereditary history-preserving (hhp-)bisimulation is a downward closed FR hp-bisimulation. $\mathcal{E}_1,\mathcal{E}_2$ are FR hereditary history-preserving (hhp-)bisimilar and are written $\mathcal{E}_1\sim_{hhp}^{fr}\mathcal{E}_2$.
\end{definition}

\begin{definition}[Weak forward-reverse (FR) (hereditary) history-preserving bisimulation]\label{FRWHHPB}
A weak FR history-preserving (hp-) bisimulation is a weakly posetal relation $R\subseteq\mathcal{C}(\mathcal{E}_1)\overline{\times}\mathcal{C}(\mathcal{E}_2)$ such that (1) if $(C_1,f,C_2)\in R$, and $C_1\xRightarrow{e_1} C_1'$, then $C_2\xRightarrow{e_2} C_2'$, with $(C_1',f[e_1\mapsto e_2],C_2')\in R$, and vice-versa, (2) if $(C_1',f',C_2')\in R$, and $C_1'\xTworightarrow{e_1[m]} C_1$, then $C_2'\xTworightarrow{e_2[n]} C_2$, with $(C_1,f'[e_1[m]\mapsto e_2[n]],C_2)\in R$, and vice-versa. $\mathcal{E}_1,\mathcal{E}_2$ are weak FR history-preserving (hp-) bisimilar and are written $\mathcal{E}_1\approx_{hp}^{fr}\mathcal{E}_2$ if there exists a weak FR hp-bisimulation $R$ such that $(\emptyset,\emptyset,\emptyset)\in R$.

A weak FR hereditary history-preserving (hhp-) bisimulation is a downward closed weak FR hp-bisimulation. $\mathcal{E}_1,\mathcal{E}_2$ are weak FR hereditary history-preserving (hhp-) bisimilar and are written $\mathcal{E}_1\approx_{hhp}^{fr}\mathcal{E}_2$.
\end{definition}

\begin{definition}[Branching forward-reverse pomset, step bisimulation]\label{FRBPSB}
Assume a special termination predicate $\downarrow$, and let $\surd$ represent a state with $\surd\downarrow$. Let $\mathcal{E}_1$, $\mathcal{E}_2$ be PESs. A branching FR pomset bisimulation is a relation $R\subseteq\mathcal{C}(\mathcal{E}_1)\times\mathcal{C}(\mathcal{E}_2)$, such that:
 \begin{enumerate}
   \item if $(C_1,C_2)\in R$, and $C_1\xrightarrow{X}C_1'$ then
   \begin{itemize}
     \item either $X\equiv \tau^*$, and $(C_1',C_2)\in R$;
     \item or there is a sequence of (zero or more) $\tau$-transitions $C_2\xrightarrow{\tau^*} C_2^0$, such that $(C_1,C_2^0)\in R$ and $C_2^0\xRightarrow{X}C_2'$ with $(C_1',C_2')\in R$;
   \end{itemize}
   \item if $(C_1,C_2)\in R$, and $C_2\xrightarrow{X}C_2'$ then
   \begin{itemize}
     \item either $X\equiv \tau^*$, and $(C_1,C_2')\in R$;
     \item or there is a sequence of (zero or more) $\tau$-transitions $C_1\xrightarrow{\tau^*} C_1^0$, such that $(C_1^0,C_2)\in R$ and $C_1^0\xRightarrow{X}C_1'$ with $(C_1',C_2')\in R$;
   \end{itemize}
   \item if $(C_1,C_2)\in R$ and $C_1\downarrow$, then there is a sequence of (zero or more) $\tau$-transitions $C_2\xrightarrow{\tau^*}C_2^0$ such that $(C_1,C_2^0)\in R$ and $C_2^0\downarrow$;
   \item if $(C_1,C_2)\in R$ and $C_2\downarrow$, then there is a sequence of (zero or more) $\tau$-transitions $C_1\xrightarrow{\tau^*}C_1^0$ such that $(C_1^0,C_2)\in R$ and $C_1^0\downarrow$;
   \item if $(C_1',C_2')\in R$, and $C_1'\xtworightarrow{X[\mathcal{K}]}C_1$ then
   \begin{itemize}
     \item either $X[\mathcal{K}]\equiv \tau^*$, and $(C_1,C_2')\in R$;
     \item or there is a sequence of (zero or more) $\tau$-transitions $C_2'\xtworightarrow{\tau^*} C_2'^0$, such that $(C_1',C_2'^0)\in R$ and $C_2'^0\xTworightarrow{X[\mathcal{K}]}C_2$ with $(C_1,C_2)\in R$;
   \end{itemize}
   \item if $(C_1',C_2')\in R$, and $C_2'\xtworightarrow{X}C_2$ then
   \begin{itemize}
     \item either $X[\mathcal{K}]\equiv \tau^*$, and $(C_1',C_2)\in R$;
     \item or there is a sequence of (zero or more) $\tau$-transitions $C_1'\xtworightarrow{\tau^*} C_1'^0$, such that $(C_1'^0,C_2')\in R$ and $C_1'^0\xTworightarrow{X[\mathcal{K}]}C_1$ with $(C_1,C_2)\in R$;
   \end{itemize}
   \item if $(C_1',C_2')\in R$ and $C_1'\downarrow$, then there is a sequence of (zero or more) $\tau$-transitions $C_2'\xtworightarrow{\tau^*}C_2'^0$ such that $(C_1',C_2'^0)\in R$ and $C_2'^0\downarrow$;
   \item if $(C_1',C_2')\in R$ and $C_2'\downarrow$, then there is a sequence of (zero or more) $\tau$-transitions $C_1'\xtworightarrow{\tau^*}C_1'^0$ such that $(C_1'^0,C_2')\in R$ and $C_1'^0\downarrow$.
 \end{enumerate}

We say that $\mathcal{E}_1$, $\mathcal{E}_2$ are branching FR pomset bisimilar, written $\mathcal{E}_1\approx_{bp}^{fr}\mathcal{E}_2$, if there exists a branching FR pomset bisimulation $R$, such that $(\emptyset,\emptyset)\in R$.

By replacing FR pomset transitions with FR steps, we can get the definition of branching FR step bisimulation. When PESs $\mathcal{E}_1$ and $\mathcal{E}_2$ are branching FR step bisimilar, we write $\mathcal{E}_1\approx_{bs}^{fr}\mathcal{E}_2$.
\end{definition}

\begin{definition}[Rooted branching forward-reverse (FR) pomset, step bisimulation]\label{FRRBPSB}
Assume a special termination predicate $\downarrow$, and let $\surd$ represent a state with $\surd\downarrow$. Let $\mathcal{E}_1$, $\mathcal{E}_2$ be PESs. A rooted branching FR pomset bisimulation is a relation $R\subseteq\mathcal{C}(\mathcal{E}_1)\times\mathcal{C}(\mathcal{E}_2)$, such that:
 \begin{enumerate}
   \item if $(C_1,C_2)\in R$, and $C_1\xrightarrow{X}C_1'$ then $C_2\xrightarrow{X}C_2'$ with $C_1'\approx_{bp}C_2'$;
   \item if $(C_1,C_2)\in R$, and $C_2\xrightarrow{X}C_2'$ then $C_1\xrightarrow{X}C_1'$ with $C_1'\approx_{bp}C_2'$;
   \item if $(C_1',C_2')\in R$, and $C_1'\xtworightarrow{X[\mathcal{K}]}C_1$ then $C_2'\xtworightarrow{X[\mathcal{K}]}C_2$ with $C_1\approx_{bp}^{fr}C_2$;
   \item if $(C_1',C_2')\in R$, and $C_2'\xtworightarrow{X[\mathcal{K}]}C_2$ then $C_1'\xtworightarrow{X[\mathcal{K}]}C_1$ with $C_1\approx_{bp}^{fr}C_2$;
   \item if $(C_1,C_2)\in R$ and $C_1\downarrow$, then $C_2\downarrow$;
   \item if $(C_1,C_2)\in R$ and $C_2\downarrow$, then $C_1\downarrow$.
 \end{enumerate}

We say that $\mathcal{E}_1$, $\mathcal{E}_2$ are rooted branching FR pomset bisimilar, written $\mathcal{E}_1\approx_{rbp}^{fr}\mathcal{E}_2$, if there exists a rooted branching FR pomset bisimulation $R$, such that $(\emptyset,\emptyset)\in R$.

By replacing FR pomset transitions with FR steps, we can get the definition of rooted branching FR step bisimulation. When PESs $\mathcal{E}_1$ and $\mathcal{E}_2$ are rooted branching FR step bisimilar, we write $\mathcal{E}_1\approx_{rbs}^{fr}\mathcal{E}_2$.
\end{definition}

\begin{definition}[Branching forward-reverse (FR) (hereditary) history-preserving bisimulation]\label{FRBHHPB}
Assume a special termination predicate $\downarrow$, and let $\surd$ represent a state with $\surd\downarrow$. A branching FR history-preserving (hp-) bisimulation is a weakly posetal relation $R\subseteq\mathcal{C}(\mathcal{E}_1)\overline{\times}\mathcal{C}(\mathcal{E}_2)$ such that:

 \begin{enumerate}
   \item if $(C_1,f,C_2)\in R$, and $C_1\xrightarrow{e_1}C_1'$ then
   \begin{itemize}
     \item either $e_1\equiv \tau$, and $(C_1',f[e_1\mapsto \tau],C_2)\in R$;
     \item or there is a sequence of (zero or more) $\tau$-transitions $C_2\xrightarrow{\tau^*} C_2^0$, such that $(C_1,f,C_2^0)\in R$ and $C_2^0\xrightarrow{e_2}C_2'$ with $(C_1',f[e_1\mapsto e_2],C_2')\in R$;
   \end{itemize}
   \item if $(C_1,f,C_2)\in R$, and $C_2\xrightarrow{e_2}C_2'$ then
   \begin{itemize}
     \item either $e_2\equiv \tau$, and $(C_1,f[e_2\mapsto \tau],C_2')\in R$;
     \item or there is a sequence of (zero or more) $\tau$-transitions $C_1\xrightarrow{\tau^*} C_1^0$, such that $(C_1^0,f,C_2)\in R$ and $C_1^0\xrightarrow{e_1}C_1'$ with $(C_1',f[e_2\mapsto e_1],C_2')\in R$;
   \end{itemize}
   \item if $(C_1,f,C_2)\in R$ and $C_1\downarrow$, then there is a sequence of (zero or more) $\tau$-transitions $C_2\xrightarrow{\tau^*}C_2^0$ such that $(C_1,f,C_2^0)\in R$ and $C_2^0\downarrow$;
   \item if $(C_1,f,C_2)\in R$ and $C_2\downarrow$, then there is a sequence of (zero or more) $\tau$-transitions $C_1\xrightarrow{\tau^*}C_1^0$ such that $(C_1^0,f,C_2)\in R$ and $C_1^0\downarrow$;
   \item if $(C_1',f',C_2')\in R$, and $C_1'\xtworightarrow{e_1[m]}C_1$ then
   \begin{itemize}
     \item either $e_1[m]\equiv \tau$, and $(C_1,f'[e_1[m]\mapsto \tau],C_2')\in R$;
     \item or there is a sequence of (zero or more) $\tau$-transitions $C_2'\xtworightarrow{\tau^*} C_2'^0$, such that $(C_1',f',C_2'^0)\in R$ and $C_2'^0\xtworightarrow{e_2[n]}C_2$ with $(C_1,f'[e_1[m]\mapsto e_2[n]],C_2)\in R$;
   \end{itemize}
   \item if $(C_1',f',C_2')\in R$, and $C_2'\xtworightarrow{e_2[n]}C_2$ then
   \begin{itemize}
     \item either $e_2[n]\equiv \tau$, and $(C_1',f'[e_2[n]\mapsto \tau],C_2)\in R$;
     \item or there is a sequence of (zero or more) $\tau$-transitions $C_1'\xtworightarrow{\tau^*} C_1'^0$, such that $(C_1'^0,f',C_2')\in R$ and $C_1'^0\xtworightarrow{e_1[m]}C_1$ with $(C_1,f[e_2[n]\mapsto e_1[m]],C_2)\in R$;
   \end{itemize}
   \item if $(C_1',f',C_2')\in R$ and $C_1'\downarrow$, then there is a sequence of (zero or more) $\tau$-transitions $C_2'\xtworightarrow{\tau^*}C_2'^0$ such that $(C_1',f',C_2'^0)\in R$ and $C_2'^0\downarrow$;
   \item if $(C_1',f',C_2')\in R$ and $C_2'\downarrow$, then there is a sequence of (zero or more) $\tau$-transitions $C_1'\xtworightarrow{\tau^*}C_1'^0$ such that $(C_1'^0,f',C_2')\in R$ and $C_1'^0\downarrow$.
 \end{enumerate}

$\mathcal{E}_1,\mathcal{E}_2$ are branching FR history-preserving (hp-)bisimilar and are written $\mathcal{E}_1\approx_{bhp}^{fr}\mathcal{E}_2$ if there exists a branching FR hp-bisimulation $R$ such that $(\emptyset,\emptyset,\emptyset)\in R$.

A branching FR hereditary history-preserving (hhp-)bisimulation is a downward closed branching FR hp-bisimulation. $\mathcal{E}_1,\mathcal{E}_2$ are branching FR hereditary history-preserving (hhp-)bisimilar and are written $\mathcal{E}_1\approx_{bhhp}^{fr}\mathcal{E}_2$.
\end{definition}

\begin{definition}[Rooted branching forward-reverse (FR) (hereditary) history-preserving bisimulation]\label{FRRBHHPB}
Assume a special termination predicate $\downarrow$, and let $\surd$ represent a state with $\surd\downarrow$. A rooted branching FR history-preserving (hp-) bisimulation is a weakly posetal relation $R\subseteq\mathcal{C}(\mathcal{E}_1)\overline{\times}\mathcal{C}(\mathcal{E}_2)$ such that:

 \begin{enumerate}
   \item if $(C_1,f,C_2)\in R$, and $C_1\xrightarrow{e_1}C_1'$, then $C_2\xrightarrow{e_2}C_2'$ with $C_1'\approx_{bhp}C_2'$;
   \item if $(C_1,f,C_2)\in R$, and $C_2\xrightarrow{e_2}C_2'$, then $C_1\xrightarrow{e_1}C_1'$ with $C_1'\approx_{bhp}C_2'$;
   \item if $(C_1',f',C_2')\in R$, and $C_1'\xtworightarrow{e_1[m]}C_1$, then $C_2'\xtworightarrow{e_2[n]}C_2$ with $C_1\approx_{bhp}^{fr}C_2$;
   \item if $(C_1',f',C_2')\in R$, and $C_2'\xtworightarrow{e_2[n]}C_2$, then $C_1'\xtworightarrow{e_1[m]}C_1$ with $C_1\approx_{bhp}^{fr}C_2$;
   \item if $(C_1,f,C_2)\in R$ and $C_1\downarrow$, then $C_2\downarrow$;
   \item if $(C_1,f,C_2)\in R$ and $C_2\downarrow$, then $C_1\downarrow$.
 \end{enumerate}

$\mathcal{E}_1,\mathcal{E}_2$ are rooted branching FR history-preserving (hp-)bisimilar and are written $\mathcal{E}_1\approx_{rbhp}^{fr}\mathcal{E}_2$ if there exists a rooted branching FR hp-bisimulation $R$ such that $(\emptyset,\emptyset,\emptyset)\in R$.

A rooted branching FR hereditary history-preserving (hhp-)bisimulation is a downward closed rooted branching FR hp-bisimulation. $\mathcal{E}_1,\mathcal{E}_2$ are rooted branching FR hereditary history-preserving (hhp-)bisimilar and are written $\mathcal{E}_1\approx_{rbhhp}^{fr}\mathcal{E}_2$.
\end{definition}

\subsection{Proof Techniques}\label{pt}

In this subsection, we introduce the concepts and conclusions about elimination, which is very important in the proof of completeness theorem.

\begin{definition}[Elimination property]
Let a process algebra with a defined set of basic terms as a subset of the set of closed terms over the process algebra. Then the process algebra has the elimination to basic terms property if for every closed term $s$ of the algebra, there exists a basic term $t$ of the algebra such that the algebra$\vdash s=t$.
\end{definition}

\begin{definition}[Strongly normalizing]
A term $s_0$ is called strongly normalizing if does not an infinite series of reductions beginning in $s_0$.
\end{definition}

\begin{definition}
We write $s>_{lpo} t$ if $s\rightarrow^+ t$ where $\rightarrow^+$ is the transitive closure of the reduction relation defined by the transition rules of a algebra.
\end{definition}

\begin{theorem}[Strong normalization]\label{SN}
Let a term rewriting (TRS) system with finitely many rewriting rules and let $>$ be a well-founded ordering on the signature of the corresponding algebra. If $s>_{lpo} t$ for each rewriting rule $s\rightarrow t$ in the TRS, then the term rewriting system is strongly normalizing.
\end{theorem}

\section{Basic Algebra for Reversible True Concurrency}{\label{bartc}}

In this section, we will discuss the algebraic laws of the confliction $+$ and causal relation $\cdot$ based on reversible truly concurrent bisimulations. The resulted algebra is called Basic Algebra for Reversible True Concurrency, abbreviated BARTC.

\subsection{Axiom System of BARTC}

In the following, let $e_1, e_2, e_1', e_2'\in \mathbb{E}$, and let variables $x,y,z$ range over the set of terms for true concurrency, $p,q,s$ range over the set of closed terms. The predicate $Std(x)$ denotes that $x$ contains only standard events (no histories of events) and $NStd(x)$ means that $x$ only contains histories of events. The set of axioms of BARTC consists of the laws given in Table \ref{AxiomsForBARTC}.

\begin{center}
    \begin{table}
        \begin{tabular}{@{}ll@{}}
            \hline No. &Axiom\\
            $A1$ & $x+ y = y+ x$\\
            $A2$ & $(x+ y)+ z = x+ (y+ z)$\\
            $A3$ & $x+ x = x$\\
            $A41$ & $(x+ y)\cdot z = x\cdot z + y\cdot z\quad Std(x),Std(y), Std(z)$\\
            $A42$ & $x\cdot (y+z) = x\cdot y + x\cdot z\quad NStd(x),NStd(y), NStd(z)$\\
            $A5$ & $(x\cdot y)\cdot z = x\cdot(y\cdot z)$\\
        \end{tabular}
        \caption{Axioms of BARTC}
        \label{AxiomsForBARTC}
    \end{table}
\end{center}

\subsection{Properties of BARTC}

\begin{definition}[Basic terms of BARTC]\label{BTBARTC}
The set of basic terms of BARTC, $\mathcal{B}(BARTC)$, is inductively defined as follows:
\begin{enumerate}
  \item $\mathbb{E}\subset\mathcal{B}(BARTC)$;
  \item if $e\in \mathbb{E}, t\in\mathcal{B}(BARTC)$ then $e\cdot t\in\mathcal{B}(BARTC)$;
  \item if $e[m]\in \mathbb{E}, t\in\mathcal{B}(BARTC)$ then $t\cdot e[m]\in\mathcal{B}(BARTC)$;
  \item if $t,s\in\mathcal{B}(BARTC)$ then $t+ s\in\mathcal{B}(BARTC)$.
\end{enumerate}
\end{definition}

\begin{theorem}[Elimination theorem of BARTC]\label{ETBARTC}
Let $p$ be a closed BARTC term. Then there is a basic BARTC term $q$ such that $BARTC\vdash p=q$.
\end{theorem}

\begin{proof}
(1) Firstly, suppose that the following ordering on the signature of BARTC is defined: $\cdot > +$ and the symbol $\cdot$ is given the lexicographical status for the first argument, then for each rewrite rule $p\rightarrow q$ in Table \ref{TRSForBARTC} relation $p>_{lpo} q$ can easily be proved. We obtain that the term rewrite system shown in Table \ref{TRSForBARTC} is strongly normalizing, for it has finitely many rewriting rules, and $>$ is a well-founded ordering on the signature of BARTC, and if $s>_{lpo} t$, for each rewriting rule $s\rightarrow t$ is in Table \ref{TRSForBARTC} (see Theorem \ref{SN}).

\begin{center}
    \begin{table}
        \begin{tabular}{@{}ll@{}}
            \hline No. &Rewriting Rule\\
            $RA3$ & $x+ x \rightarrow x$\\
            $RA41$ & $(x+ y)\cdot z \rightarrow x\cdot z + y\cdot z$\\
            $RA42$ & $x\cdot(y+ z) \rightarrow x\cdot y + x\cdot z$\\
            $RA5$ & $(x\cdot y)\cdot z \rightarrow x\cdot(y\cdot z)$\\
        \end{tabular}
        \caption{Term rewrite system of BARTC}
        \label{TRSForBARTC}
    \end{table}
\end{center}

(2) Then we prove that the normal forms of closed BARTC terms are basic BARTC terms.

Suppose that $p$ is a normal form of some closed BARTC term and suppose that $p$ is not a basic term. Let $p'$ denote the smallest sub-term of $p$ which is not a basic term. It implies that each sub-term of $p'$ is a basic term. Then we prove that $p$ is not a term in normal form. It is sufficient to induct on the structure of $p'$:

\begin{itemize}
  \item Case $p'\equiv e, e\in \mathbb{E}$. $p'$ is a basic term, which contradicts the assumption that $p'$ is not a basic term, so this case should not occur.
  \item Case $p'\equiv p_1\cdot p_2$. By induction on the structure of the basic term $p_1$:
      \begin{itemize}
        \item Subcase $p_1\in \mathbb{E}$. $p'$ would be a basic term, which contradicts the assumption that $p'$ is not a basic term;
        \item Subcase $p_1\equiv e\cdot p_1'$. $RA5$ rewriting rule can be applied. So $p$ is not a normal form;
        \item Subcase $p_1\equiv p_1'\cdot e[m]$. $RA5$ rewriting rule can be applied. So $p$ is not a normal form;
        \item Subcase $p_1\equiv p_1'+ p_1''$. $RA41$ and $RA42$ rewriting rule can be applied. So $p$ is not a normal form.
      \end{itemize}
  \item Case $p'\equiv p_1+ p_2$. By induction on the structure of the basic terms both $p_1$ and $p_2$, all subcases will lead to that $p'$ would be a basic term, which contradicts the assumption that $p'$ is not a basic term.
\end{itemize}
\end{proof}

\subsection{Structured Operational Semantics of BARTC}

In this subsection, we will define a term-deduction system which gives the operational semantics of BARTC. We give the forward operational transition rules of operators $\cdot$ and $+$ as Table \ref{SETRForBARTC} shows, and the reverse rules of operators $\cdot$ and $+$ as Table \ref{RSETRForBARTC} shows. And the predicate $\xrightarrow{e}e[m]$ represents successful forward termination after forward execution of the event $e$, the predicate $\xtworightarrow{e[m]}e$ represents successful reverse termination after reverse execution of the event history $e[m]$.

\begin{center}
    \begin{table}
        $$\frac{}{e\xrightarrow{e}e[m]}$$
        $$\frac{x\xrightarrow{e}e[m]}{x+ y\xrightarrow{e}e[m]} \quad\frac{x\xrightarrow{e}x'}{x+ y\xrightarrow{e}x'} \quad\frac{y\xrightarrow{e}e[m]}{x+ y\xrightarrow{e}e[m]} \quad\frac{y\xrightarrow{e}y'}{x+ y\xrightarrow{e}y'}$$
        $$\frac{x\xrightarrow{e}e[m]}{x\cdot y\xrightarrow{e} e[m]\cdot y} \quad\frac{x\xrightarrow{e}x'}{x\cdot y\xrightarrow{e}x'\cdot y}$$
        \caption{Forward single event transition rules of BARTC}
        \label{SETRForBARTC}
    \end{table}
\end{center}

\begin{center}
    \begin{table}
        $$\frac{}{e[m]\xtworightarrow{e[m]}e}$$
        $$\frac{x\xtworightarrow{e[m]}e}{x+ y\xtworightarrow{e[m]}e} \quad\frac{x\xtworightarrow{e[m]}x'}{x+ y\xtworightarrow{e[m]}x'} \quad\frac{y\xtworightarrow{e[m]}e}{x+ y\xtworightarrow{e[m]}e} \quad\frac{y\xtworightarrow{e[m]}y'}{x+ y\xtworightarrow{e[m]}y'}$$
        $$\frac{y\xtworightarrow{e[m]}e}{x\cdot y\xtworightarrow{e[m]} x\cdot e} \quad\frac{y\xtworightarrow{e[m]}y'}{x\cdot y\xtworightarrow{e[m]}x\cdot y'}$$
        \caption{Reverse single event transition rules of BARTC}
        \label{RSETRForBARTC}
    \end{table}
\end{center}

%
%
%
%
%

The forward pomset transition rules are shown in Table \ref{PTRForBARTC}, and reverse pomset transition rules are shown in Table \ref{RPTRForBARTC}, different to single event transition rules, the pomset transition rules are labeled by pomsets, which are defined by causality $\cdot$ and conflict $+$.

\begin{center}
    \begin{table}
        $$\frac{}{X\xrightarrow{X}X[\mathcal{K}]}$$
        $$\frac{x\xrightarrow{X}X[\mathcal{K}]}{x+ y\xrightarrow{X}X[\mathcal{K}]} (X\subseteq x)\quad\frac{x\xrightarrow{X}x'}{x+ y\xrightarrow{X}x'} (X\subseteq x) \quad\frac{y\xrightarrow{Y}Y[\mathcal{K}]}{x+ y\xrightarrow{Y}Y[\mathcal{K}]} (Y\subseteq y)\quad\frac{y\xrightarrow{Y}y'}{x+ y\xrightarrow{Y}y'}(Y\subseteq y)$$
        $$\frac{x\xrightarrow{X}X[\mathcal{K}]}{x\cdot y\xrightarrow{X} X[\mathcal{K}]\cdot y} (X\subseteq x)\quad\frac{x\xrightarrow{X}x'}{x\cdot y\xrightarrow{X}x'\cdot y} (X\subseteq x)$$
        \caption{Forward pomset transition rules of BARTC}
        \label{PTRForBARTC}
    \end{table}
\end{center}

\begin{center}
    \begin{table}
        $$\frac{}{X[\mathcal{K}]\xtworightarrow{X[\mathcal{K}]}X}$$
        $$\frac{x\xtworightarrow{X[\mathcal{K}]}X}{x+ y\xtworightarrow{X[\mathcal{K}]}X} (X\subseteq x)\quad\frac{x\xtworightarrow{X[\mathcal{K}]}x'}{x+ y\xtworightarrow{X[\mathcal{K}]}x'} (X\subseteq x) \quad\frac{y\xtworightarrow{Y[\mathcal{K}]}Y}{x+ y\xtworightarrow{Y[\mathcal{K}]}Y} (Y\subseteq y)\quad\frac{y\xtworightarrow{Y[\mathcal{K}]}y'}{x+ y\xtworightarrow{Y[\mathcal{K}]}y'}(Y\subseteq y)$$
        $$\frac{y\xtworightarrow{Y[\mathcal{K}]}Y}{x\cdot y\xtworightarrow{Y[\mathcal{K}]} x\cdot Y} (Y\subseteq y)\quad\frac{y\xtworightarrow{Y[\mathcal{K}]}y'}{x\cdot y\xtworightarrow{Y[\mathcal{K}]}x\cdot y'} (Y\subseteq y)$$
        \caption{Reverse pomset transition rules of BARTC}
        \label{RPTRForBARTC}
    \end{table}
\end{center}

\begin{theorem}[Congruence of BARTC with respect to FR pomset bisimulation equivalence]
FR pomset bisimulation equivalence $\sim_p^{fr}$ is a congruence with respect to BARTC.
\end{theorem}

\begin{proof}
It is easy to see that FR pomset bisimulation is an equivalent relation on BARTC terms, we only need to prove that $\sim_p^{fr}$ is preserved by the operators $\cdot$ and $+$.

\begin{itemize}
  \item Causality operator $\cdot$. Let $x_1,x_2$ and $y_1,y_2$ be BARTC processes, and $x_1\sim_p^{fr} y_1$, $x_2\sim_p^{fr} y_2$, it is sufficient to prove that $x_1\cdot x_2\sim_p^{fr} y_1\cdot y_2$.

      By the definition of FR pomset bisimulation $\sim_p^{fr}$ (Definition \ref{FRPSB}), $x_1\sim_p^{fr} y_1$ means that

      $$x_1\xrightarrow{X_1} x_1' \quad y_1\xrightarrow{Y_1} y_1'$$

      $$x_1\xtworightarrow{X_1[\mathcal{K}]} x_1' \quad y_1\xtworightarrow{Y_1[\mathcal{L}]} y_1'$$

      with $X_1\subseteq x_1$, $Y_1\subseteq y_1$, $X_1\sim Y_1$ and $x_1'\sim_p^{fr} y_1'$. The meaning of $x_2\sim_p^{fr} y_2$ is similar.

      By the pomset transition rules for causality operator $\cdot$ in Table \ref{PTRForBARTC} and Table \ref{RPTRForBARTC}, we can get

      $$x_1\cdot x_2\xrightarrow{X_1} X_1[\mathcal{K}]\cdot x_2 \quad y_1\cdot y_2\xrightarrow{Y_1} Y_1[\mathcal{L}]\cdot y_2$$

      $$x_1\cdot x_2\xtworightarrow{X_2[\mathcal{K}]} x_1\cdot X_2 \quad y_1\cdot y_2\xtworightarrow{Y_2[\mathcal{L}]} y_1\cdot Y_2$$

      with $X_1\subseteq x_1$, $Y_1\subseteq y_1$, $X_1\sim Y_1$ and $x_2\sim_p^{fr} y_2$; $X_2\subseteq x_2$, $Y_2\subseteq y_2$, $X_2\sim Y_2$ and $x_1\sim_p^{fr} y_1$ so, we get $x_1\cdot x_2\sim_p^{fr} y_1\cdot y_2$, as desired.

      Or, we can get

      $$x_1\cdot x_2\xrightarrow{X_1} x_1'\cdot x_2 \quad y_1\cdot y_2\xrightarrow{Y_1} y_1'\cdot y_2$$

      $$x_1\cdot x_2\xtworightarrow{X_2[\mathcal{K}]} x_1\cdot x_2' \quad y_1\cdot y_2\xtworightarrow{Y_2[\mathcal{L}]} y_1\cdot y_2'$$

      with $X_1\subseteq x_1$, $Y_1\subseteq y_1$, $X_1\sim Y_1$ and $x_1'\sim_p^{fr} y_1'$, $x_2\sim_p^{fr} y_2$; $X_2\subseteq x_2$, $Y_2\subseteq y_2$, $X_2\sim Y_2$ and $x_2'\sim_p^{fr} y_2'$, $x_1\sim_p^{fr} y_1$, so, we get $x_1\cdot x_2\sim_p^{fr} y_1\cdot y_2$, as desired.
  \item Conflict operator $+$. Let $x_1, x_2$ and $y_1, y_2$ be BARTC processes, and $x_1\sim_p^{fr} y_1$, $x_2\sim_p^{fr} y_2$, it is sufficient to prove that $x_1+ x_2 \sim_p^{fr} y_1+ y_2$. The meanings of $x_1\sim_p^{fr} y_1$ and $x_2\sim_p^{fr} y_2$ are the same as the above case, according to the definition of FR pomset bisimulation $\sim_p^{fr}$ in Definition \ref{FRPSB}.

      By the pomset transition rules for conflict operator $+$ in Table \ref{PTRForBARTC} and Table \ref{RPTRForBARTC}, we can get four cases:

      $$x_1+ x_2\xrightarrow{X_1} X_1[\mathcal{K}] \quad y_1+ y_2\xrightarrow{Y_1} Y_1[\mathcal{L}]$$

      $$x_1+ x_2\xtworightarrow{X_1[\mathcal{K}]} X_1 \quad y_1+ y_2\xtworightarrow{Y_1[\mathcal{L}]} Y_1$$

      with $X_1\subseteq x_1$, $Y_1\subseteq y_1$, $X_1\sim Y_1$, so, we get $x_1+ x_2\sim_p^{fr} y_1+ y_2$, as desired.

      Or, we can get

      $$x_1+ x_2\xrightarrow{X_1} x_1' \quad y_1+ y_2\xrightarrow{Y_1} y_1'$$

      $$x_1+ x_2\xtworightarrow{X_1[\mathcal{K}]} x_1' \quad y_1+ y_2\xtworightarrow{Y_1[\mathcal{L}]} y_1'$$

      with $X_1\subseteq x_1$, $Y_1\subseteq y_1$, $X_1\sim Y_1$, and $x_1'\sim_p^{fr} y_1'$, so, we get $x_1+ x_2\sim_p^{fr} y_1+ y_2$, as desired.

      Or, we can get

      $$x_1+ x_2\xrightarrow{X_2} X_2[\mathcal{K}] \quad y_1+ y_2\xrightarrow{Y_2} Y_2[\mathcal{L}]$$

      $$x_1+ x_2\xtworightarrow{X_2[\mathcal{K}]} X_2 \quad y_1+ y_2\xtworightarrow{Y_2[\mathcal{L}]} Y_2$$

      with $X_2\subseteq x_2$, $Y_2\subseteq y_2$, $X_2\sim Y_2$, so, we get $x_1+ x_2\sim_p^{fr} y_1+ y_2$, as desired.

      Or, we can get

      $$x_1+ x_2\xrightarrow{X_2} x_2' \quad y_1+ y_2\xrightarrow{Y_2} y_2'$$

      $$x_1+ x_2\xtworightarrow{X_2[\mathcal{K}]} x_2' \quad y_1+ y_2\xtworightarrow{Y_2[\mathcal{L}]} y_2'$$

      with $X_2\subseteq x_2$, $Y_2\subseteq y_2$, $X_2\sim Y_2$, and $x_2'\sim_p^{fr} y_2'$, so, we get $x_1+ x_2\sim_p^{fr} y_1+ y_2$, as desired.
\end{itemize}
\end{proof}

\begin{theorem}[Soundness of BARTC modulo FR pomset bisimulation equivalence]\label{SBARTCPBE}
Let $x$ and $y$ be BARTC terms. If $BARTC\vdash x=y$, then $x\sim_p^{fr} y$.
\end{theorem}

\begin{proof}
Since FR pomset bisimulation $\sim_p^{fr}$ is both an equivalent and a congruent relation, we only need to check if each axiom in Table \ref{AxiomsForBARTC} is sound modulo FR pomset bisimulation equivalence.

\begin{itemize}
  \item \textbf{Axiom $A1$}. Let $p,q$ be BARTC processes, and $p+ q=q+ p$, it is sufficient to prove that $p+ q\sim_p^{fr} q+ p$. By the pomset transition rules for operator $+$ in Table \ref{PTRForBARTC} and Table \ref{RPTRForBARTC}, we get

      $$\frac{p\xrightarrow{P}P[\mathcal{K}]}{p+ q\xrightarrow{P}P[\mathcal{K}]} (P\subseteq p) \quad \frac{p\xrightarrow{P}P[\mathcal{K}]}{q+ p\xrightarrow{P}P[\mathcal{K}]}(P\subseteq p)$$

      $$\frac{p\xtworightarrow{P[\mathcal{K}]}P}{p+ q\xtworightarrow{P[\mathcal{K}]}P} (P\subseteq p) \quad \frac{p\xtworightarrow{P[\mathcal{K}]}P}{q+ p\xtworightarrow{P[\mathcal{K}]}P}(P\subseteq p)$$

      $$\frac{p\xrightarrow{P}p'}{p+ q\xrightarrow{P}p'}(P\subseteq p) \quad \frac{p\xrightarrow{P}p'}{q+ p\xrightarrow{P}p'}(P\subseteq p)$$

      $$\frac{p\xtworightarrow{P[\mathcal{K}]}p'}{p+ q\xtworightarrow{P[\mathcal{K}]}p'}(P\subseteq p) \quad \frac{p\xtworightarrow{P[\mathcal{K}]}p'}{q+ p\xtworightarrow{P[\mathcal{K}]}p'}(P\subseteq p)$$

      $$\frac{q\xrightarrow{Q}Q[\mathcal{L}]}{p+ q\xrightarrow{Q}Q[\mathcal{L}]}(Q\subseteq q) \quad \frac{q\xrightarrow{Q}Q[\mathcal{L}]}{q+ p\xrightarrow{Q}Q[\mathcal{L}]}(Q\subseteq q)$$

      $$\frac{q\xtworightarrow{Q[\mathcal{L}]}Q}{p+ q\xtworightarrow{Q[\mathcal{L}]}Q}(Q\subseteq q) \quad \frac{q\xtworightarrow{Q[\mathcal{L}]}Q}{q+ p\xtworightarrow{Q[\mathcal{L}]}Q}(Q\subseteq q)$$

      $$\frac{q\xrightarrow{Q}q'}{p+ q\xrightarrow{Q}q'}(Q\subseteq q) \quad \frac{q\xrightarrow{Q}q'}{q+ p\xrightarrow{Q}q'}(Q\subseteq q)$$

      $$\frac{q\xtworightarrow{Q[\mathcal{L}]}q'}{p+ q\xtworightarrow{Q[\mathcal{L}]}q'}(Q\subseteq q) \quad \frac{q\xtworightarrow{Q[\mathcal{L}]}q'}{q+ p\xtworightarrow{Q[\mathcal{L}]}q'}(Q\subseteq q)$$

      So, $p+ q\sim_p^{fr} q+ p$, as desired.
  \item \textbf{Axiom $A2$}. Let $p,q,s$ be BARTC processes, and $(p+ q)+ s=p+ (q+ s)$, it is sufficient to prove that $(p+ q)+ s \sim_p^{fr} p+ (q+ s)$. By the pomset transition rules for operator $+$ in Table \ref{PTRForBARTC} and Table \ref{RPTRForBARTC}, we get

      $$\frac{p\xrightarrow{P}P[\mathcal{K}]}{(p+ q)+ s\xrightarrow{P}P[\mathcal{K}]} (P\subseteq p) \quad \frac{p\xrightarrow{P}P[\mathcal{K}]}{p+ (q+ s)\xrightarrow{P}P[\mathcal{K}]}(P\subseteq p)$$

      $$\frac{p\xtworightarrow{P[\mathcal{K}]}P}{(p+ q)+ s\xtworightarrow{P[\mathcal{K}]}P} (P\subseteq p) \quad \frac{p\xtworightarrow{P[\mathcal{K}]}P}{p+ (q+ s)\xtworightarrow{P[\mathcal{K}]}P}(P\subseteq p)$$

      $$\frac{p\xrightarrow{P}p'}{(p+ q)+ s\xrightarrow{P}p'}(P\subseteq p) \quad \frac{p\xrightarrow{P}p'}{p+ (q+ s)\xrightarrow{P}p'}(P\subseteq p)$$

      $$\frac{p\xtworightarrow{P[\mathcal{K}]}p'}{(p+ q)+ s\xtworightarrow{P[\mathcal{K}]}p'}(P\subseteq p) \quad \frac{p\xtworightarrow{P[\mathcal{K}]}p'}{p+ (q+ s)\xtworightarrow{P[\mathcal{K}]}p'}(P\subseteq p)$$

      $$\frac{q\xrightarrow{Q}Q[\mathcal{L}]}{(p+ q)+ s\xrightarrow{Q}Q[\mathcal{L}]}(Q\subseteq q) \quad \frac{q\xrightarrow{Q}Q[\mathcal{L}]}{p+ (q+ s)\xrightarrow{Q}Q[\mathcal{L}]}(Q\subseteq q)$$

      $$\frac{q\xtworightarrow{Q[\mathcal{L}]}Q}{(p+ q)+ s\xtworightarrow{Q[\mathcal{L}]}Q}(Q\subseteq q) \quad \frac{q\xtworightarrow{Q[\mathcal{L}]}Q}{p+ (q+ s)\xtworightarrow{Q[\mathcal{L}]}Q}(Q\subseteq q)$$

      $$\frac{q\xrightarrow{Q}q'}{(p+ q)+ s\xrightarrow{Q}q'}(Q\subseteq q) \quad \frac{q\xrightarrow{Q}q'}{p+ (q+ s)\xrightarrow{Q}q'}(Q\subseteq q)$$

      $$\frac{q\xtworightarrow{Q[\mathcal{L}]}q'}{(p+ q)+ s\xtworightarrow{Q[\mathcal{L}]}q'}(Q\subseteq q) \quad \frac{q\xtworightarrow{Q[\mathcal{L}]}q'}{p+ (q+ s)\xtworightarrow{Q[\mathcal{L}]}q'}(Q\subseteq q)$$

      $$\frac{s\xrightarrow{S}S[\mathcal{M}]}{(p+ q)+ s\xrightarrow{S}S[\mathcal{M}]}(S\subseteq s) \quad \frac{s\xrightarrow{S}S[\mathcal{M}]}{p+ (q+ s)\xrightarrow{S}S[\mathcal{M}]}(S\subseteq s)$$

      $$\frac{s\xtworightarrow{S[\mathcal{M}]}S}{(p+ q)+ s\xtworightarrow{S[\mathcal{M}]}S}(S\subseteq s) \quad \frac{s\xtworightarrow{S[\mathcal{M}]}S}{p+ (q+ s)\xtworightarrow{S[\mathcal{M}]}S}(S\subseteq s)$$

      $$\frac{s\xrightarrow{S}s'}{(p+ q)+ s\xrightarrow{S}s'}(S\subseteq s) \quad \frac{s\xrightarrow{S}s'}{p+ (q+ s)\xrightarrow{S}s'}(S\subseteq s)$$

      $$\frac{s\xtworightarrow{S[\mathcal{M}]}s'}{(p+ q)+ s\xtworightarrow{S[\mathcal{M}]}s'}(S\subseteq s) \quad \frac{s\xtworightarrow{S[\mathcal{M}]}s'}{p+ (q+ s)\xtworightarrow{S[\mathcal{M}]}s'}(S\subseteq s)$$

      So, $(p+ q)+ s\sim_p^{fr} p+ (q+ s)$, as desired.
  \item \textbf{Axiom $A3$}. Let $p$ be a BARTC process, and $p+ p=p$, it is sufficient to prove that $p+ p\sim_p^{fr} p$. By the pomset transition rules for operator $+$ in Table \ref{PTRForBARTC} and Table \ref{RPTRForBARTC}, we get

      $$\frac{p\xrightarrow{P}P[\mathcal{K}]}{p+ p\xrightarrow{P}P[\mathcal{K}]} (P\subseteq p) \quad \frac{p\xrightarrow{P}P[\mathcal{K}]}{p\xrightarrow{P}P[\mathcal{K}]}(P\subseteq p)$$

      $$\frac{p\xtworightarrow{P[\mathcal{K}]}P}{p+ p\xtworightarrow{P[\mathcal{K}]}P} (P\subseteq p) \quad \frac{p\xtworightarrow{P[\mathcal{K}]}P}{p\xtworightarrow{P[\mathcal{K}]}P}(P\subseteq p)$$

      $$\frac{p\xrightarrow{P}p'}{p+ p\xrightarrow{P}p'}(P\subseteq p) \quad \frac{p\xrightarrow{P}p'}{p\xrightarrow{P}p'}(P\subseteq p)$$

      $$\frac{p\xtworightarrow{P[\mathcal{K}]}p'}{p+ p\xtworightarrow{P[\mathcal{K}]}p'}(P\subseteq p) \quad \frac{p\xtworightarrow{P[\mathcal{K}]}p'}{p\xtworightarrow{P[\mathcal{K}]}p'}(P\subseteq p)$$

      So, $p+ p\sim_p^{fr} p$, as desired.
  \item \textbf{Axiom $A41$}. Let $p,q,s$ be BARTC processes, $Std(p),Std(q),Std(s)$, and $(p+ q)\cdot s=p\cdot s + q\cdot s$, it is sufficient to prove that $(p+ q)\cdot s \sim_p^{fr} p\cdot s + q\cdot s$. By the pomset transition rules for operators $+$ and $\cdot$ in Table \ref{PTRForBARTC}, we get

      $$\frac{p\xrightarrow{P}P[\mathcal{K}]}{(p+ q)\cdot s\xrightarrow{P}P[\mathcal{K}]\cdot s} (P\subseteq p) \quad \frac{p\xrightarrow{P}P[\mathcal{K}]}{p\cdot s + q\cdot s\xrightarrow{P}P[\mathcal{K}]\cdot s}(P\subseteq p)$$

      $$\frac{p\xrightarrow{P}p'}{(p+ q)\cdot s\xrightarrow{P}p'\cdot s}(P\subseteq p) \quad \frac{p\xrightarrow{P}p'}{p\cdot s+ q\cdot s\xrightarrow{P}p'\cdot s}(P\subseteq p)$$

      $$\frac{q\xrightarrow{Q}Q[\mathcal{L}]}{(p+ q)\cdot s\xrightarrow{Q}Q[\mathcal{K}]\cdot s}(Q\subseteq q) \quad \frac{q\xrightarrow{Q}Q[\mathcal{L}]}{p\cdot s+ q\cdot s\xrightarrow{Q}Q[\mathcal{K}]\cdot s}(Q\subseteq q)$$

      $$\frac{q\xrightarrow{Q}q'}{(p+ q)\cdot s\xrightarrow{Q}q'\cdot s}(Q\subseteq q) \quad \frac{q\xrightarrow{Q}q'}{p\cdot s+ q\cdot s\xrightarrow{Q}q'\cdot s}(Q\subseteq q)$$

      So, $(p+ q)\cdot s\sim_p^{fr} p\cdot s+ q\cdot s$, as desired.
    \item \textbf{Axiom $A42$}. Let $p,q,s$ be BARTC processes, $NStd(p),NStd(q),NStd(s)$, and $p\cdot(q+ s)=p\cdot q + p\cdot s$, it is sufficient to prove that $p\cdot (q+s) \sim_p^{fr} p\cdot q + p\cdot s$. By the pomset transition rules for operators $+$ and $\cdot$ in Table \ref{RPTRForBARTC}, we get

      $$\frac{q\xtworightarrow{Q[\mathcal{L}]}Q}{p\cdot (q+s)\xtworightarrow{Q[\mathcal{L}]}p\cdot Q} (Q\subseteq q) \quad \frac{q\xtworightarrow{Q[\mathcal{L}]}Q}{p\cdot q + p\cdot s\xtworightarrow{Q[\mathcal{L}]}p\cdot Q}(Q\subseteq q)$$

      $$\frac{q\xtworightarrow{Q[\mathcal{L}]}q'}{p\cdot (q+s)\xtworightarrow{Q[\mathcal{L}]}p\cdot q'} (Q\subseteq q) \quad \frac{q\xtworightarrow{Q[\mathcal{L}]}q'}{p\cdot q + p\cdot s\xtworightarrow{Q[\mathcal{L}]}p\cdot q'}(Q\subseteq q)$$

      $$\frac{s\xtworightarrow{S[\mathcal{M}]}S}{p\cdot (q+s)\xtworightarrow{S[\mathcal{M}]}p\cdot S} (S\subseteq s) \quad \frac{s\xtworightarrow{S[\mathcal{M}]}S}{p\cdot q + p\cdot s\xtworightarrow{S[\mathcal{M}]}p\cdot S}(S\subseteq s)$$

      $$\frac{s\xtworightarrow{S[\mathcal{M}]}s'}{p\cdot (q+s)\xtworightarrow{S[\mathcal{M}]}p\cdot s'} (S\subseteq s) \quad \frac{s\xtworightarrow{S[\mathcal{M}]}s'}{p\cdot q + p\cdot s\xtworightarrow{S[\mathcal{M}]}p\cdot s'}(S\subseteq s)$$

      So, $p\cdot (q+s)\sim_p^{fr} p\cdot q+ p\cdot s$, as desired.
  \item \textbf{Axiom $A5$}. Let $p,q,s$ be BARTC processes, and $(p\cdot q)\cdot s=p\cdot (q\cdot s)$, it is sufficient to prove that $(p\cdot q)\cdot s \sim_p^{fr} p\cdot (q\cdot s)$. By the pomset transition rules for operator $\cdot$ in Table \ref{PTRForBARTC} and Table \ref{RPTRForBARTC}, we get

      $$\frac{p\xrightarrow{P}P[\mathcal{K}]}{(p\cdot q)\cdot s\xrightarrow{P}(P[\mathcal{K}]\cdot q)\cdot s} (P\subseteq p) \quad \frac{p\xrightarrow{P}P[\mathcal{K}]}{p\cdot (q\cdot s)\xrightarrow{P}P[\mathcal{K}]\cdot(q\cdot s)}(P\subseteq p)$$

      $$\frac{p\xrightarrow{P}p'}{(p\cdot q)\cdot s\xrightarrow{P}(p'\cdot q)\cdot s}(P\subseteq p) \quad \frac{p\xrightarrow{P}p'}{p\cdot (q\cdot s)\xrightarrow{P}p'\cdot (q\cdot s)}(P\subseteq p)$$

      $$\frac{s\xtworightarrow{S[\mathcal{M}}S]}{(p\cdot q)\cdot s\xtworightarrow{S[\mathcal{M}]}(p\cdot q)\cdot S} (S\subseteq s) \quad \frac{s\xtworightarrow{S[\mathcal{M}]}S}{p\cdot (q\cdot s)\xtworightarrow{S[\mathcal{K}]}p\cdot(q\cdot S)}(S\subseteq s)$$

      $$\frac{s\xtworightarrow{S[\mathcal{M}]}s'}{(p\cdot q)\cdot s\xtworightarrow{S[\mathcal{M}]}(p\cdot q)\cdot s'}(S\subseteq s) \quad \frac{s\xtworightarrow{S[\mathcal{M}]}s'}{p\cdot (q\cdot s)\xtworightarrow{S[\mathcal{M}]}p\cdot (q\cdot s')}(S\subseteq s)$$

      With an assumptions $(p\cdot q)\cdot S=p\cdot(q\cdot S)$ and $(p\cdot q)\cdot s'=p\cdot(q\cdot s')$, so, $(p\cdot q)\cdot s\sim_p^{fr} p\cdot (q\cdot s)$, as desired.
\end{itemize}
\end{proof}

\begin{theorem}[Completeness of BARTC modulo FR pomset bisimulation equivalence]\label{CBARTCPBE}
Let $p$ and $q$ be closed BARTC terms, if $p\sim_p^{fr} q$ then $p=q$.
\end{theorem}

\begin{proof}
Firstly, by the elimination theorem of BARTC, we know that for each closed BARTC term $p$, there exists a closed basic $BARTC$ term $p'$, such that $BARTC\vdash p=p'$, so, we only need to consider closed basic $BARTC$ terms.

The basic terms (see Definition \ref{BTBARTC}) modulo associativity and commutativity (AC) of conflict $+$ (defined by axioms $A1$ and $A2$ in Table \ref{AxiomsForBARTC}), and this equivalence is denoted by $=_{AC}$. Then, each equivalence class $s$ modulo AC of $+$ has the following normal form

$$s_1+\cdots+ s_k$$

with each $s_i$ either an atomic event or of the form $t_1\cdot t_2$, and each $s_i$ is called the summand of $s$.

Now, we prove that for normal forms $n$ and $n'$, if $n\sim_p^{fr} n'$ then $n=_{AC}n'$. It is sufficient to induct on the sizes of $n$ and $n'$.

\begin{itemize}
  \item Consider a summand $e$ of $n$. Then $n\xrightarrow{e}e[m]$, so $n\sim_p^{fr} n'$ implies $n'\xrightarrow{e}e[m]$, meaning that $n'$ also contains the summand $e$.
  \item Consider a summand $e[m]$ of $n$. Then $n\xtworightarrow{e[m]}e$, so $n\sim_p^{fr} n'$ implies $n'\xtworightarrow{e[m]}e$, meaning that $n'$ also contains the summand $e[m]$.
  \item Consider a summand $t_1\cdot t_2$ of $n$. Then $n\xrightarrow{t_1}t_1[\mathcal{K}]\cdot t_2$, so $n\sim_p^{fr} n'$ implies $n'\xrightarrow{t_1}t_1[\mathcal{K}]\cdot t_2'$ with $t_1[\mathcal{K}]\cdot t_2\sim_p^{fr} t_1[\mathcal{K}]\cdot t_2'$, meaning that $n'$ contains a summand $t_1\cdot t_2'$. Since $t_2$ and $t_2'$ are normal forms and have sizes no greater than $n$ and $n'$, by the induction hypotheses $t_2\sim_p^{fr} t_2'$ implies $t_2=_{AC} t_2'$.
  \item Consider a summand $t_1\cdot t_2[\mathcal{L}]$ of $n$. Then $n\xtworightarrow{t_2[\mathcal{L}]}t_1\cdot t_2$, so $n\sim_p^{fr} n'$ implies $n'\xtworightarrow{t_2[\mathcal{L}]}t_1'\cdot t_2$ with $t_1\cdot t_2[\mathcal{L}]\sim_p^{fr} t_1'\cdot t_2[\mathcal{L}]$, meaning that $n'$ contains a summand $t_1'\cdot t_2[\mathcal{L}]$. Since $t_21$ and $t_1'$ are normal forms and have sizes no greater than $n$ and $n'$, by the induction hypotheses $t_1\sim_p^{fr} t_1'$ implies $t_1=_{AC} t_1'$.
\end{itemize}

So, we get $n=_{AC} n'$.

Finally, let $s$ and $t$ be basic terms, and $s\sim_p^{fr} t$, there are normal forms $n$ and $n'$, such that $s=n$ and $t=n'$. The soundness theorem of BARTC modulo FR pomset bisimulation equivalence (see Theorem \ref{SBARTCPBE}) yields $s\sim_p^{fr} n$ and $t\sim_p^{fr} n'$, so $n\sim_p^{fr} s\sim_p^{fr} t\sim_p^{fr} n'$. Since if $n\sim_p^{fr} n'$ then $n=_{AC}n'$, $s=n=_{AC}n'=t$, as desired.
\end{proof}

The step transition rules are almost the same as the transition rules in Table \ref{PTRForBARTC} and Table \ref{RPTRForBARTC}, the difference is that events in the transition pomset are pairwise concurrent for the step transition rules, and we omit them.

\begin{theorem}[Congruence of BARTC with respect to FR step bisimulation equivalence]
Step bisimulation equivalence $\sim_s^{fr}$ is a congruence with respect to BARTC.
\end{theorem}

\begin{proof}
It is easy to see that FR step bisimulation is an equivalent relation on BARTC terms, we only need to prove that $\sim_{s}^{fr}$ is preserved by the operators $\cdot$ and $+$. The proof is almost the same as proof of congruence of BARTC with respect to FR pomset bisimulation equivalence, the difference is that events in the transition pomset are pairwise concurrent for FR step bisimulation equivalence, and we omit it.
\end{proof}

\begin{theorem}[Soundness of BARTC modulo FR step bisimulation equivalence]\label{SBARTCSBE}
Let $x$ and $y$ be BARTC terms. If $BARTC\vdash x=y$, then $x\sim_{s}^{fr} y$.
\end{theorem}

\begin{proof}
Since FR step bisimulation $\sim_s^{fr}$ is both an equivalent and a congruent relation, we only need to check if each axiom in Table \ref{AxiomsForBARTC} is sound modulo FR step bisimulation equivalence. The soundness proof is almost the same as soundness proof of BARTC modulo FR pomset bisimulation equivalence, the difference is that events in the transition pomset are pairwise concurrent, and we omit it.
\end{proof}

\begin{theorem}[Completeness of BARTC modulo FR step bisimulation equivalence]\label{CBARTCSBE}
Let $p$ and $q$ be closed BARTC terms, if $p\sim_{s}^{fr} q$ then $p=q$.
\end{theorem}

\begin{proof}
The proof of completeness is almost the same as the proof of BARTC modulo FR pomset bisimulation equivalence, the only different is that events in the transition pomset are pairwise concurrent, and we omit it.
\end{proof}

The transition rules for (hereditary) FR hp-bisimulation of BARTC are the same as single event transition rules in Table \ref{SETRForBARTC} Table \ref{RSETRForBARTC}.

\begin{theorem}[Congruence of BARTC with respect to FR hp-bisimulation equivalence]
Hp-bisimulation equivalence $\sim_{hp}^{fr}$ is a congruence with respect to BARTC.
\end{theorem}

\begin{proof}
It is easy to see that history-preserving bisimulation is an equivalent relation on BARTC terms, we only need to prove that $\sim_{hp}^{fr}$ is preserved by the operators $\cdot$ and $+$.

The proof is similar to the proof of congruence of BARTC with respenct to FR pomset bisimulation equivalence, we omit it.
\end{proof}

\begin{theorem}[Soundness of BARTC modulo FR hp-bisimulation equivalence]\label{SBARTCHPBE}
Let $x$ and $y$ be BARTC terms. If $BARTC\vdash x=y$, then $x\sim_{hp}^{fr} y$.
\end{theorem}

\begin{proof}
Since FR hp-bisimulation $\sim_{hp}^{fr}$ is both an equivalent and a congruent relation, we only need to check if each axiom in Table \ref{AxiomsForBARTC} is sound modulo FR hp-bisimulation equivalence.

The proof is similar to the proof of soundness of BARTC modulo FR pomset and step bisimulation equivalences, we omit it.
\end{proof}

\begin{theorem}[Completeness of BARTC modulo FR hp-bisimulation equivalence]\label{CBARTCHPBE}
Let $p$ and $q$ be closed BARTC terms, if $p\sim_{hp}^{fr} q$ then $p=q$.
\end{theorem}

\begin{proof}
The proof is similar to the proof of completeness of BARTC modulo FR pomset and step bisimulation equivalences, we omit it.
\end{proof}

\begin{theorem}[Congruence of BARTC with respect to FR hhp-bisimulation equivalence]
Hhp-bisimulation equivalence $\sim_{hhp}^{fr}$ is a congruence with respect to BARTC.
\end{theorem}

\begin{proof}
It is easy to see that FR hhp-bisimulation is an equivalent relation on BARTC terms, we only need to prove that $\sim_{hhp}^{fr}$ is preserved by the operators $\cdot$ and $+$.

The proof is similar to the proof of congruence of BARTC with respect to FR hp-bisimulation equivalence, we omit it.
\end{proof}

\begin{theorem}[Soundness of BARTC modulo FR hhp-bisimulation equivalence]\label{SBARTCHHPBE}
Let $x$ and $y$ be BARTC terms. If $BARTC\vdash x=y$, then $x\sim_{hhp}^{fr} y$.
\end{theorem}

\begin{proof}
Since FR hhp-bisimulation $\sim_{hhp}^{fr}$ is both an equivalent and a congruent relation, we only need to check if each axiom in Table \ref{AxiomsForBARTC} is sound modulo FR hhp-bisimulation equivalence.

The proof is similar to the proof of soundness of BARTC modulo FR hp-bisimulation equivalence, we omit it.
\end{proof}

\begin{theorem}[Completeness of BARTC modulo FR hhp-bisimulation equivalence]\label{CBARTCHHPBE}
Let $p$ and $q$ be closed BARTC terms, if $p\sim_{hhp}^{fr} q$ then $p=q$.
\end{theorem}

\begin{proof}
The proof is similar to the proof of BARTC modulo FR hp-bisimulation equivalence, we omit it.
\end{proof}

\section{Algebra for Parallelism in Reversible True Concurrency}\label{aprtc}

In this section, we will discuss parallelism in reversible true concurrency. The resulted algebra is called Algebra for Parallelism in Reversible True Concurrency, abbreviated APRTC.

\subsection{Parallelism}

The forward transition rules for parallelism $\parallel$ are shown in Table \ref{TRForParallel}, and the reverse transition rules for $\parallel$ are shown in Table \ref{RTRForParallel}.

\begin{center}
    \begin{table}
        $$\frac{x\xrightarrow{e_1}e_1[m]\quad y\xrightarrow{e_2}e_2[m]}{x\parallel y\xrightarrow{\{e_1,e_2\}}e_1[m]\parallel e_2[m]} \quad\frac{x\xrightarrow{e_1}x'\quad y\xrightarrow{e_2}e_2[m]}{x\parallel y\xrightarrow{\{e_1,e_2\}}x'\parallel e_2[m]}$$
        $$\frac{x\xrightarrow{e_1}e_1[m]\quad y\xrightarrow{e_2}y'}{x\parallel y\xrightarrow{\{e_1,e_2\}}e_1[m]\parallel y'} \quad\frac{x\xrightarrow{e_1}x'\quad y\xrightarrow{e_2}y'}{x\parallel y\xrightarrow{\{e_1,e_2\}}x'\between y'}$$
        \caption{Forward transition rules of parallel operator $\parallel$}
        \label{TRForParallel}
    \end{table}
\end{center}

\begin{center}
    \begin{table}
        $$\frac{x\xtworightarrow{e_1[m]}e_1\quad y\xtworightarrow{e_2[m]}e_2}{x\parallel y\xtworightarrow{\{e_1[m],e_2[m]\}}e_1\parallel e_2} \quad\frac{x\xtworightarrow{e_1[m]}x'\quad y\xtworightarrow{e_2[m]}e_2}{x\parallel y\xtworightarrow{\{e_1[m],e_2[m]\}}x'\parallel e_2}$$
        $$\frac{x\xtworightarrow{e_1[m]}e_1\quad y\xtworightarrow{e_2[m]}y'}{x\parallel y\xtworightarrow{\{e_1[m],e_2[m]\}}e_1\parallel y'} \quad\frac{x\xtworightarrow{e_1[m]}x'\quad y\xtworightarrow{e_2[m]}y'}{x\parallel y\xtworightarrow{\{e_1[m],e_2[m]\}}x'\between y'}$$
        \caption{Reverse transition rules of parallel operator $\parallel$}
        \label{RTRForParallel}
    \end{table}
\end{center}

The forward and reverse transition rules of communication $\mid$ are shown in Table \ref{TRForCommunication} and Table \ref{RTRForCommunication}.

\begin{center}
    \begin{table}
        $$\frac{x\xrightarrow{e_1}e_1[m]\quad y\xrightarrow{e_2}e_2[m]}{x\mid y\xrightarrow{\gamma(e_1,e_2)}\gamma(e_1,e_2)[m]} \quad\frac{x\xrightarrow{e_1}x'\quad y\xrightarrow{e_2}e_2[m]}{x\mid y\xrightarrow{\gamma(e_1,e_2)}\gamma(e_1,e_2)[m]\cdot x'}$$
        $$\frac{x\xrightarrow{e_1}e_1[m]\quad y\xrightarrow{e_2}y'}{x\mid y\xrightarrow{\gamma(e_1,e_2)}\gamma(e_1,e_2)[m]\cdot y'} \quad\frac{x\xrightarrow{e_1}x'\quad y\xrightarrow{e_2}y'}{x\mid y\xrightarrow{\gamma(e_1,e_2)}\gamma(e_1,e_2)[m]\cdot x'\between y'}$$
        \caption{Forward transition rules of communication operator $\mid$}
        \label{TRForCommunication}
    \end{table}
\end{center}

\begin{center}
    \begin{table}
        $$\frac{x\xtworightarrow{e_1[m]}e_1\quad y\xtworightarrow{e_2[m]}e_2}{x\mid y\xtworightarrow{\gamma(e_1,e_2)[m]}\gamma(e_1,e_2)} \quad\frac{x\xtworightarrow{e_1[m]}x'\quad y\xtworightarrow{e_2[m]}e_2}{x\mid y\xtworightarrow{\gamma(e_1,e_2)[m]}\gamma(e_1,e_2)\cdot x'}$$
        $$\frac{x\xtworightarrow{e_1[m]}e_1\quad y\xtworightarrow{e_2[m]}y'}{x\mid y\xtworightarrow{\gamma(e_1,e_2)[m]}\gamma(e_1,e_2)\cdot y'} \quad\frac{x\xtworightarrow{e_1[m]}x'\quad y\xtworightarrow{e_2[m]}y'}{x\mid y\xtworightarrow{\gamma(e_1,e_2)[m]}\gamma(e_1,e_2)\cdot x'\between y'}$$
        \caption{Reverse transition rules of communication operator $\mid$}
        \label{RTRForCommunication}
    \end{table}
\end{center}

The conflict elimination is also captured by two auxiliary operators, the unary conflict elimination operator $\Theta$ and the binary unless operator $\triangleleft$. The forward and reverse transition rules for $\Theta$ and $\triangleleft$ are expressed by ten transition rules in Table \ref{TRForConflict} and Table \ref{RTRForConflict}.

\begin{center}
    \begin{table}
        $$\frac{x\xrightarrow{e_1}e_1[m]\quad (\sharp(e_1,e_2))}{\Theta(x)\xrightarrow{e_1}e_1[m]} \quad\frac{x\xrightarrow{e_2}e_2[n]\quad (\sharp(e_1,e_2))}{\Theta(x)\xrightarrow{e_2}e_2[n]}$$
        $$\frac{x\xrightarrow{e_1}x'\quad (\sharp(e_1,e_2))}{\Theta(x)\xrightarrow{e_1}\Theta(x')} \quad\frac{x\xrightarrow{e_2}x'\quad (\sharp(e_1,e_2))}{\Theta(x)\xrightarrow{e_2}\Theta(x')}$$
        $$\frac{x\xrightarrow{e_1}e_1[m] \quad y\nrightarrow^{e_2}\quad (\sharp(e_1,e_2))}{x\triangleleft y\xrightarrow{\tau}\surd}
        \quad\frac{x\xrightarrow{e_1}x' \quad y\nrightarrow^{e_2}\quad (\sharp(e_1,e_2))}{x\triangleleft y\xrightarrow{\tau}x'}$$
        $$\frac{x\xrightarrow{e_1}e_1[m] \quad y\nrightarrow^{e_3}\quad (\sharp(e_1,e_2),e_2\leq e_3)}{x\triangleleft y\xrightarrow{e_1}e_1[m]}
        \quad\frac{x\xrightarrow{e_1}x' \quad y\nrightarrow^{e_3}\quad (\sharp(e_1,e_2),e_2\leq e_3)}{x\triangleleft y\xrightarrow{e_1}x'}$$
        $$\frac{x\xrightarrow{e_3}e_3[l] \quad y\nrightarrow^{e_2}\quad (\sharp(e_1,e_2),e_1\leq e_3)}{x\triangleleft y\xrightarrow{\tau}\surd}
        \quad\frac{x\xrightarrow{e_3}x' \quad y\nrightarrow^{e_2}\quad (\sharp(e_1,e_2),e_1\leq e_3)}{x\triangleleft y\xrightarrow{\tau}x'}$$
        \caption{Forward transition rules of conflict elimination}
        \label{TRForConflict}
    \end{table}
\end{center}

\begin{center}
    \begin{table}
        $$\frac{x\xtworightarrow{e_1[m]}e_1\quad (\sharp(e_1,e_2))}{\Theta(x)\xtworightarrow{e_1[m]}e_1} \quad\frac{x\xtworightarrow{e_2[n]}e_2\quad (\sharp(e_1,e_2))}{\Theta(x)\xtworightarrow{e_2[n]}e_2}$$
        $$\frac{x\xtworightarrow{e_1[m]}x'\quad (\sharp(e_1,e_2))}{\Theta(x)\xtworightarrow{e_1[m]}\Theta(x')} \quad\frac{x\xtworightarrow{e_2[n]}x'\quad (\sharp(e_1,e_2))}{\Theta(x)\xtworightarrow{e_2[n]}\Theta(x')}$$
        $$\frac{x\xtworightarrow{e_1[m]}e_1 \quad y\xntworightarrow{e_2[n]}\quad (\sharp(e_1,e_2))}{x\triangleleft y\xtworightarrow{\tau}\surd}
        \quad\frac{x\xtworightarrow{e_1[m]}x' \quad y\xntworightarrow{e_2[n]}\quad (\sharp(e_1,e_2))}{x\triangleleft y\xtworightarrow{\tau}x'}$$
        $$\frac{x\xtworightarrow{e_1[m]}e_1 \quad y\xntworightarrow{e_3[l]}\quad (\sharp(e_1,e_2),e_2\geq e_3)}{x\triangleleft y\xtworightarrow{e_1[m]}e_1}
        \quad\frac{x\xtworightarrow{e_1[m]}x' \quad y\xntworightarrow{e_3[l]}\quad (\sharp(e_1,e_2),e_2\geq e_3)}{x\triangleleft y\xtworightarrow{e_1[m]}x'}$$
        $$\frac{x\xtworightarrow{e_3[l]}e_3 \quad y\xntworightarrow{e_2[n]}\quad (\sharp(e_1,e_2),e_1\geq e_3)}{x\triangleleft y\xtworightarrow{\tau}\surd}
        \quad\frac{x\xtworightarrow{e_3[l]}x' \quad y\xntworightarrow{e_2[n]}\quad (\sharp(e_1,e_2),e_1\geq e_3)}{x\triangleleft y\xtworightarrow{\tau}x'}$$
        \caption{Reverse transition rules of conflict elimination}
        \label{RTRForConflict}
    \end{table}
\end{center}

\begin{theorem}[Congruence theorem of APRTC]
FR truly concurrent bisimulation equivalences $\sim_{p}^{fr}$, $\sim_s^{fr}$, $\sim_{hp}^{fr}$ and $\sim_{hhp}^{fr}$ are all congruences with respect to APRTC.
\end{theorem}

\begin{proof}
(1) Case FR pomset bisimulation equivalence $\sim_p^{fr}$.

\begin{itemize}
  \item Case parallel operator $\parallel$. Let $x_1,x_2$ and $y_1,y_2$ be APRTC processes, and $x_1\sim_{p}^{fr} y_1$, $x_2\sim_{p}^{fr} y_2$, it is sufficient to prove that $x_1\parallel x_2\sim_{p}^{fr} y_1\parallel y_2$.

      By the definition of FR pomset bisimulation $\sim_p^{fr}$ (Definition \ref{FRPSB}), $x_1\sim_p^{fr} y_1$ means that

      $$x_1\xrightarrow{X_1} x_1' \quad y_1\xrightarrow{Y_1} y_1'$$

      with $X_1\subseteq x_1$, $Y_1\subseteq y_1$, $X_1\sim Y_1$ and $x_1'\sim_p^{fr} y_1'$. The meaning of $x_2\sim_p^{fr} y_2$ is similar.

      By the forward transition rules for parallel operator $\parallel$ in Table \ref{TRForParallel}, we can get

      $$x_1\parallel x_2\xrightarrow{\{X_1,X_2\}} X_1[\mathcal{K}]\parallel X_2[\mathcal{K}] \quad y_1\parallel y_2\xrightarrow{\{Y_1,Y_2\}} Y_1[\mathcal{J}]\parallel Y_2[\mathcal{J}]$$

      $$x_1\parallel x_2\xtworightarrow{\{X_1[\mathcal{K}],X_2[\mathcal{K}]\}} X_1\parallel X_2 \quad y_1\parallel y_2\xtworightarrow{\{Y_1[\mathcal{J}],Y_2[\mathcal{J}]\}} Y_1\parallel Y_2$$

      with $X_1\subseteq x_1$, $Y_1\subseteq y_1$, $X_2\subseteq x_2$, $Y_2\subseteq y_2$, $X_1\sim Y_1$ and $X_2\sim Y_2$, and the assumptions $X_1[\mathcal{K}\parallel X_2[\mathcal{K}]]\sim_p^{fr}Y_1[\mathcal{J}]\parallel Y_2[\mathcal{J}]$ and $X_1\parallel X_2\sim_p^{fr}Y_1\parallel Y_2$, so, we get $x_1\parallel x_2\sim_p^{fr} y_1\parallel y_2$, as desired.

      Or, we can get

      $$x_1\parallel x_2\xrightarrow{\{X_1,X_2\}} x_1'\parallel X_2[\mathcal{K}] \quad y_1\parallel y_2\xrightarrow{\{Y_1,Y_2\}} y_1'\parallel Y_2[\mathcal{J}]$$

      $$x_1\parallel x_2\xtworightarrow{\{X_1[\mathcal{K}],X_2[\mathcal{K}]\}} x_1'\parallel X_2 \quad y_1\parallel y_2\xtworightarrow{\{Y_1[\mathcal{J}],Y_2[\mathcal{J}]\}} y_1'\parallel Y_2$$

      with $X_1\subseteq x_1$, $Y_1\subseteq y_1$, $X_2\subseteq x_2$, $Y_2\subseteq y_2$, $X_1\sim Y_1$, $X_2\sim Y_2$, and the assumptions $x_1'\parallel X_2[\mathcal{K}]]\sim_p^{fr}y_1'\parallel Y_2[\mathcal{J}]$ and $x_1'\parallel X_2\sim_p^{fr}y_1'\parallel Y_2$ so, we get $x_1\parallel x_2\sim_p^{fr} y_1\parallel y_2$, as desired.

      Or, we can get

      $$x_1\parallel x_2\xrightarrow{\{X_1,X_2\}}X_1[\mathcal{K}]\parallel x_2' \quad y_1\parallel y_2\xrightarrow{\{Y_1,Y_2\}}Y_1[\mathcal{J}]\parallel y_2'$$

      $$x_1\parallel x_2\xtworightarrow{\{X_1[\mathcal{K}],X_2[\mathcal{K}]\}}X_1\parallel x_2' \quad y_1\parallel y_2\xtworightarrow{\{Y_1[\mathcal{J}],Y_2[\mathcal{J}]\}}Y_1\parallel y_2'$$

      with $X_1\subseteq x_1$, $Y_1\subseteq y_1$, $X_2\subseteq x_2$, $Y_2\subseteq y_2$, $X_1\sim Y_1$, $X_2\sim Y_2$, and the assumptions $X_1[\mathcal{K}\parallel x_2'\sim_p^{fr}Y_1[\mathcal{J}]\parallel y_2'$ and $X_1\parallel x_2'\sim_p^{fr}Y_1\parallel y_2'$, so, we get $x_1\parallel x_2\sim_p^{fr} y_1\parallel y_2$, as desired.

      Or, we can get

      $$x_1\parallel x_2\xrightarrow{\{X_1,X_2\}} x_1'\between x_2' \quad y_1\parallel y_2\xrightarrow{\{Y_1,Y_2\}} y_1'\between y_2'$$

      $$x_1\parallel x_2\xtworightarrow{\{X_1[\mathcal{K}],X_2[\mathcal{K}]\}} x_1'\between x_2' \quad y_1\parallel y_2\xtworightarrow{\{Y_1[\mathcal{J}],Y_2[\mathcal{J}]\}} y_1'\between y_2'$$

      with $X_1\subseteq x_1$, $Y_1\subseteq y_1$, $X_2\subseteq x_2$, $Y_2\subseteq y_2$, $X_1\sim Y_1$, $X_2\sim Y_2$, and the assumption $x_1'\between x_2'\sim_p^{fr}y_1'\between y_2'$, so, we get $x_1\parallel x_2\sim_p^{fr} y_1\parallel y_2$, as desired.

  \item Case communication operator $\mid$. It can be proved similarly to the case of parallel operator $\parallel$, we omit it. Note that, a communication is defined between two single communicating events.

  \item Case conflict elimination operator $\Theta$. It can be proved similarly to the above cases, we omit it. Note that the conflict elimination operator $\Theta$ is a unary operator.

  \item Case unless operator $\triangleleft$. It can be proved similarly to the case of parallel operator $\parallel$, we omit it. Note that, a conflict relation is defined between two single events.

\end{itemize}

(2) The cases of FR step bisimulation $\sim_s^{fr}$, FR hp-bisimulation $\sim_{hp}^{fr}$ and FR hhp-bisimulation $\sim_{hhp}^{fr}$ can be proven similarly, we omit them.
\end{proof}

\subsection{Axiom System of Parallelism}

\begin{definition}[Basic terms of APRTC]\label{BTAPTC}
The set of basic terms of APRTC, $\mathcal{B}(APRTC)$, is inductively defined as follows:
\begin{enumerate}
  \item $\mathbb{E}\subset\mathcal{B}(APRTC)$;
  \item if $e\in \mathbb{E}, t\in\mathcal{B}(APRTC)$ then $e\cdot t\in\mathcal{B}(APRTC)$;
  \item if $e[m]\in \mathbb{E}, t\in\mathcal{B}(APRTC)$ then $t\cdot e[m]\in\mathcal{B}(APRTC)$;
  \item if $t,s\in\mathcal{B}(APRTC)$ then $t+ s\in\mathcal{B}(APRTC)$;
  \item if $t,s\in\mathcal{B}(APRTC)$ then $t\parallel s\in\mathcal{B}(APRTC)$.
\end{enumerate}
\end{definition}

We design the axioms of parallelism in Table \ref{AxiomsForParallelism}, including algebraic laws for parallel operator $\parallel$, communication operator $\mid$, conflict elimination operator $\Theta$ and unless operator $\triangleleft$, and also the whole parallel operator $\between$. Since the communication between two communicating events in different parallel branches may cause deadlock (a state of inactivity), which is caused by mismatch of two communicating events or the imperfectness of the communication channel. We introduce a new constant $\delta$ to denote the deadlock, and let the atomic event $e\in \mathbb{E}\cup\{\delta\}$.

\begin{center}
    \begin{table}
        \begin{tabular}{@{}ll@{}}
            \hline No. &Axiom\\
            $A6$ & $x+ \delta = x$\\
            $A7$ & $\delta\cdot x =\delta$(\textrm{Std(x)})\\
            $RA7$ & $x\cdot \delta =\delta$(\textrm{NStd(x)})\\
            $P1$ & $x\between y = x\parallel y + x\mid y$\\
            $P2$ & $x\parallel y = y \parallel x$\\
            $P3$ & $(x\parallel y)\parallel z = x\parallel (y\parallel z)$\\
            $P4$ & $e_1\parallel (e_2\cdot y) = (e_1\parallel e_2)\cdot y$\\
            $RP4$ & $e_1[m]\parallel (y\cdot e_2[m]) = y\cdot(e_1[m]\parallel e_2[m])$\\
            $P5$ & $(e_1\cdot x)\parallel e_2 = (e_1\parallel e_2)\cdot x$\\
            $RP5$ & $(x\cdot e_1[m])\parallel e_2[m] = x\cdot(e_1[m]\parallel e_2[m])$\\
            $P6$ & $(e_1\cdot x)\parallel (e_2\cdot y) = (e_1\parallel e_2)\cdot (x\between y)$\\
            $RP6$ & $(x\cdot e_1[m])\parallel (y\cdot e_2[m]) = (x\between y)\cdot(e_1[m]\parallel e_2[m])$\\
            $P7$ & $(x+ y)\parallel z = (x\parallel z)+ (y\parallel z)$\\
            $P8$ & $x\parallel (y+ z) = (x\parallel y)+ (x\parallel z)$\\
            $P9$ & $\delta\parallel x = \delta$\\
            $P10$ & $x\parallel \delta = \delta$\\
            $C11$ & $e_1\mid e_2 = \gamma(e_1,e_2)$\\
            $RC11$ & $e_1[m]\mid e_2[m] = \gamma(e_1,e_2)[m]$\\
            $C12$ & $e_1\mid (e_2\cdot y) = \gamma(e_1,e_2)\cdot y$\\
            $RC12$ & $e_1[m]\mid (y \cdot e_2[m]) =y\cdot \gamma(e_1,e_2)[m]$\\
            $C13$ & $(e_1\cdot x)\mid e_2 = \gamma(e_1,e_2)\cdot x$\\
            $RC13$ & $(x \cdot e_1[m])\mid e_2[m] =x\cdot \gamma(e_1,e_2)[m]$\\
            $C14$ & $(e_1\cdot x)\mid (e_2\cdot y) = \gamma(e_1,e_2)\cdot (x\between y)$\\
            $RC14$ & $(x \cdot e_1[m])\mid (y \cdot e_2[m]) =(x\between y)\cdot \gamma(e_1,e_2)[m]$\\
            $C15$ & $(x+ y)\mid z = (x\mid z) + (y\mid z)$\\
            $C16$ & $x\mid (y+ z) = (x\mid y)+ (x\mid z)$\\
            $C17$ & $\delta\mid x = \delta$\\
            $C18$ & $x\mid\delta = \delta$\\
            $CE19$ & $\Theta(e) = e$\\
            $RCE19$ & $\Theta(e[m]) = e[m]$\\
            $CE20$ & $\Theta(\delta) = \delta$\\
            $CE21$ & $\Theta(x+ y) = \Theta(x)\triangleleft y + \Theta(y)\triangleleft x$\\
            $CE22$ & $\Theta(x\cdot y)=\Theta(x)\cdot\Theta(y)$\\
            $CE23$ & $\Theta(x\parallel y) = ((\Theta(x)\triangleleft y)\parallel y)+ ((\Theta(y)\triangleleft x)\parallel x)$\\
            $CE24$ & $\Theta(x\mid y) = ((\Theta(x)\triangleleft y)\mid y)+ ((\Theta(y)\triangleleft x)\mid x)$\\
            $U25$ & $(\sharp(e_1,e_2))\quad e_1\triangleleft e_2 = \tau$\\
            $RU25$ & $(\sharp(e_1[m],e_2[n]))\quad e_1[m]\triangleleft e_2[n] = \tau$\\
            $U26$ & $(\sharp(e_1,e_2),e_2\leq e_3)\quad e_1\triangleleft e_3 = e_1$\\
            $RU26$ & $(\sharp(e_1[m],e_2[n]),e_2[n]\geq e_3[l])\quad e_1[m]\triangleleft e_3[l] = e_1[m]$\\
            $U27$ & $(\sharp(e_1,e_2),e_2\leq e_3)\quad e3\triangleleft e_1 = \tau$\\
            $RU27$ & $(\sharp(e_1[m],e_2[n]),e_2[n]\geq e_3[l])\quad e3[l]\triangleleft e_1[m] = \tau$\\
            $U28$ & $e\triangleleft \delta = e$\\
            $U29$ & $\delta \triangleleft e = \delta$\\
            $U30$ & $(x+ y)\triangleleft z = (x\triangleleft z)+ (y\triangleleft z)$\\
            $U31$ & $(x\cdot y)\triangleleft z = (x\triangleleft z)\cdot (y\triangleleft z)$\\
            $U32$ & $(x\parallel y)\triangleleft z = (x\triangleleft z)\parallel (y\triangleleft z)$\\
            $U33$ & $(x\mid y)\triangleleft z = (x\triangleleft z)\mid (y\triangleleft z)$\\
            $U34$ & $x\triangleleft (y+ z) = (x\triangleleft y)\triangleleft z$\\
            $U35$ & $x\triangleleft (y\cdot z)=(x\triangleleft y)\triangleleft z$\\
            $U36$ & $x\triangleleft (y\parallel z) = (x\triangleleft y)\triangleleft z$\\
            $U37$ & $x\triangleleft (y\mid z) = (x\triangleleft y)\triangleleft z$\\
        \end{tabular}
        \caption{Axioms of parallelism}
        \label{AxiomsForParallelism}
    \end{table}
\end{center}

Based on the definition of basic terms for APRTC (see Definition \ref{BTAPTC}) and axioms of parallelism (see Table \ref{AxiomsForParallelism}), we can prove the elimination theorem of parallelism.

\begin{theorem}[Elimination theorem of FR parallelism]\label{ETParallelism}
Let $p$ be a closed APRTC term. Then there is a basic APRTC term $q$ such that $APRTC\vdash p=q$.
\end{theorem}

\begin{proof}
(1) Firstly, suppose that the following ordering on the signature of APRTC is defined: $\parallel > \cdot > +$ and the symbol $\parallel$ is given the lexicographical status for the first argument, then for each rewrite rule $p\rightarrow q$ in Table \ref{TRSForAPRTC} relation $p>_{lpo} q$ can easily be proved. We obtain that the term rewrite system shown in Table \ref{TRSForAPRTC} is strongly normalizing, for it has finitely many rewriting rules, and $>$ is a well-founded ordering on the signature of APRTC, and if $s>_{lpo} t$, for each rewriting rule $s\rightarrow t$ is in Table \ref{TRSForAPRTC} (see Theorem \ref{SN}).

\begin{center}
    \begin{table}
        \begin{tabular}{@{}ll@{}}
            \hline No. &Rewriting Rule\\
            $RA6$ & $x+ \delta \rightarrow x$\\
            $RA7$ & $\delta\cdot x \rightarrow\delta$\\
            $RRA7$ & $x \cdot \delta\rightarrow\delta$\\
            $RP1$ & $x\between y \rightarrow x\parallel y + x\mid y$\\
            $RP2$ & $x\parallel y \rightarrow y \parallel x$\\
            $RP3$ & $(x\parallel y)\parallel z \rightarrow x\parallel (y\parallel z)$\\
            $RP4$ & $e_1\parallel (e_2\cdot y) \rightarrow (e_1\parallel e_2)\cdot y$\\
            $RRP4$ & $e_1[m]\parallel (y\cdot e_2[m]) \rightarrow y\cdot(e_1[m]\parallel e_2[m])$\\
            $RP5$ & $(e_1\cdot x)\parallel e_2 \rightarrow (e_1\parallel e_2)\cdot x$\\
            $RRP5$ & $(x\cdot e_1[m])\parallel e_2[m] \rightarrow x\cdot(e_1[m]\parallel e_2[m])$\\
            $RP6$ & $(e_1\cdot x)\parallel (e_2\cdot y) \rightarrow (e_1\parallel e_2)\cdot (x\between y)$\\
            $RP6$ & $(x\cdot e_1[m])\parallel (y\cdot e_2[m]) \rightarrow (x\between y)\cdot(e_1[m]\parallel e_2[m])$\\
            $RP7$ & $(x+ y)\parallel z \rightarrow (x\parallel z)+ (y\parallel z)$\\
            $RP8$ & $x\parallel (y+ z) \rightarrow (x\parallel y)+ (x\parallel z)$\\
            $RP9$ & $\delta\parallel x \rightarrow \delta$\\
            $RP10$ & $x\parallel \delta \rightarrow \delta$\\
            $RC11$ & $e_1\mid e_2 \rightarrow \gamma(e_1,e_2)$\\
            $RRC11$ & $e_1[m]\mid e_2[m] \rightarrow \gamma(e_1,e_2)[m]$\\
            $RC12$ & $e_1\mid (e_2\cdot y) \rightarrow \gamma(e_1,e_2)\cdot y$\\
            $RRC12$ & $e_1[m]\mid (y \cdot e_2[m]) \rightarrow y\cdot \gamma(e_1,e_2)[m]$\\
            $RC13$ & $(e_1\cdot x)\mid e_2 \rightarrow \gamma(e_1,e_2)\cdot x$\\
            $RRC13$ & $(x \cdot e_1[m])\mid e_2[m] \rightarrow x\cdot \gamma(e_1,e_2)[m]$\\
            $RC14$ & $(e_1\cdot x)\mid (e_2\cdot y) \rightarrow \gamma(e_1,e_2)\cdot (x\between y)$\\
            $RRC14$ & $(x \cdot e_1[m])\mid (y \cdot e_2[m]) \rightarrow(x\between y)\cdot \gamma(e_1,e_2)[m]$\\
            $RC15$ & $(x+ y)\mid z \rightarrow (x\mid z) + (y\mid z)$\\
            $RC16$ & $x\mid (y+ z) \rightarrow (x\mid y)+ (x\mid z)$\\
            $RC17$ & $\delta\mid x \rightarrow \delta$\\
            $RC18$ & $x\mid\delta \rightarrow \delta$\\
            $RCE19$ & $\Theta(e) \rightarrow e$\\
            $RRCE19$ & $\Theta(e[m]) \rightarrow e[m]$\\
            $RCE20$ & $\Theta(\delta) \rightarrow \delta$\\
            $RCE21$ & $\Theta(x+ y) \rightarrow \Theta(x)\triangleleft y + \Theta(y)\triangleleft x$\\
            $RCE22$ & $\Theta(x\cdot y)\rightarrow\Theta(x)\cdot\Theta(y)$\\
            $RCE23$ & $\Theta(x\parallel y) \rightarrow ((\Theta(x)\triangleleft y)\parallel y)+ ((\Theta(y)\triangleleft x)\parallel x)$\\
            $RCE24$ & $\Theta(x\mid y) \rightarrow ((\Theta(x)\triangleleft y)\mid y)+ ((\Theta(y)\triangleleft x)\mid x)$\\
            $RU25$ & $(\sharp(e_1,e_2))\quad e_1\triangleleft e_2 \rightarrow \tau$\\
            $RRU25$ & $(\sharp(e_1[m],e_2[n]))\quad e_1[m]\triangleleft e_2[n] \rightarrow \tau$\\
            $RU26$ & $(\sharp(e_1,e_2),e_2\leq e_3)\quad e_1\triangleleft e_3 \rightarrow e_1$\\
            $RRU26$ & $(\sharp(e_1[m],e_2[n]),e_2[n]\geq e_3[l])\quad e_1[m]\triangleleft e_3[l] \rightarrow e_1[m]$\\
            $RU27$ & $(\sharp(e_1,e_2),e_2\leq e_3)\quad e3\triangleleft e_1 \rightarrow \tau$\\
            $RRU27$ & $(\sharp(e_1[m],e_2[n]),e_2[n]\geq e_3[l])\quad e3[l]\triangleleft e_1[m] \rightarrow \tau$\\
            $RU28$ & $e\triangleleft \delta \rightarrow e$\\
            $RU29$ & $\delta \triangleleft e \rightarrow \delta$\\
            $RU30$ & $(x+ y)\triangleleft z \rightarrow (x\triangleleft z)+ (y\triangleleft z)$\\
            $RU31$ & $(x\cdot y)\triangleleft z \rightarrow (x\triangleleft z)\cdot (y\triangleleft z)$\\
            $RU32$ & $(x\parallel y)\triangleleft z \rightarrow (x\triangleleft z)\parallel (y\triangleleft z)$\\
            $RU33$ & $(x\mid y)\triangleleft z \rightarrow (x\triangleleft z)\mid (y\triangleleft z)$\\
            $RU34$ & $x\triangleleft (y+ z) \rightarrow (x\triangleleft y)\triangleleft z$\\
            $RU35$ & $x\triangleleft (y\cdot z)\rightarrow(x\triangleleft y)\triangleleft z$\\
            $RU36$ & $x\triangleleft (y\parallel z) \rightarrow (x\triangleleft y)\triangleleft z$\\
            $RU37$ & $x\triangleleft (y\mid z) \rightarrow (x\triangleleft y)\triangleleft z$\\
        \end{tabular}
        \caption{Term rewrite system of APRTC}
        \label{TRSForAPRTC}
    \end{table}
\end{center}

(2) Then we prove that the normal forms of closed APRTC terms are basic APRTC terms.

Suppose that $p$ is a normal form of some closed APRTC term and suppose that $p$ is not a basic APRTC term. Let $p'$ denote the smallest sub-term of $p$ which is not a basic APRTC term. It implies that each sub-term of $p'$ is a basic APRTC term. Then we prove that $p$ is not a term in normal form. It is sufficient to induct on the structure of $p'$:

\begin{itemize}
  \item Case $p'\equiv e$ or $e[m], e\in \mathbb{E}$. $p'$ is a basic APRTC term, which contradicts the assumption that $p'$ is not a basic APRTC term, so this case should not occur.
  \item Case $p'\equiv p_1\cdot p_2$. By induction on the structure of the basic APRTC term $p_1$:
      \begin{itemize}
        \item Subcase $p_1\in \mathbb{E}$. $p'$ would be a basic APRTC term, which contradicts the assumption that $p'$ is not a basic APRTC term;
        \item Subcase $p_1\equiv e\cdot p_1'$. $RR5$ rewriting rule in Table \ref{TRSForBRATC} can be applied. So $p$ is not a normal form;
        \item Subcase $p_1\equiv p_1'\cdot e[m]$. $RA5$ rewriting rule in Table \ref{TRSForBRATC} can be applied. So $p$ is not a normal form;
        \item Subcase $p_1\equiv p_1'+ p_1''$. $RA4$ rewriting rule in Table \ref{TRSForBRATC} can be applied. So $p$ is not a normal form;
        \item Subcase $p_1\equiv p_1'\parallel p_1''$. $p'$ would be a basic APRTC term, which contradicts the assumption that $p'$ is not a basic APRTC term;
        \item Subcase $p_1\equiv p_1'\mid p_1''$. $RC11$ and $RRC11$ rewrite rule in Table \ref{TRSForAPRTC} can be applied. So $p$ is not a normal form;
        \item Subcase $p_1\equiv \Theta(p_1')$. $RCE19$, $RRCE19$ and $RCE20$ rewrite rules in Table \ref{TRSForAPRTC} can be applied. So $p$ is not a normal form.
      \end{itemize}
  \item Case $p'\equiv p_1+ p_2$. By induction on the structure of the basic APRTC terms both $p_1$ and $p_2$, all subcases will lead to that $p'$ would be a basic APRTC term, which contradicts the assumption that $p'$ is not a basic APRTC term.
  \item Case $p'\equiv p_1\parallel p_2$. By induction on the structure of the basic APRTC terms both $p_1$ and $p_2$, all subcases will lead to that $p'$ would be a basic APRTC term, which contradicts the assumption that $p'$ is not a basic APRTC term.
  \item Case $p'\equiv p_1\mid p_2$. By induction on the structure of the basic APRTC terms both $p_1$ and $p_2$, all subcases will lead to that $p'$ would be a basic APRTC term, which contradicts the assumption that $p'$ is not a basic APRTC term.
  \item Case $p'\equiv \Theta(p_1)$. By induction on the structure of the basic APRTC term $p_1$, $RCE19-RCE24$ rewrite rules in Table \ref{TRSForAPRTC} can be applied. So $p$ is not a normal form.
  \item Case $p'\equiv p_1\triangleleft p_2$. By induction on the structure of the basic APRTC terms both $p_1$ and $p_2$, all subcases will lead to that $p'$ would be a basic APRTC term, which contradicts the assumption that $p'$ is not a basic APRTC term.
\end{itemize}
\end{proof}

\subsection{Structured Operational Semantics of Parallelism}

\begin{theorem}[Generalization of the algebra for parallelism with respect to BARTC]
The algebra for parallelism is a generalization of BARTC.
\end{theorem}

\begin{proof}
It follows from the following three facts.

\begin{enumerate}
  \item The transition rules of BARTC in section \ref{batc} are all source-dependent;
  \item The sources of the transition rules for the algebra for parallelism contain an occurrence of $\between$, or $\parallel$, or $\mid$, or $\Theta$, or $\triangleleft$;
  \item The transition rules of APRTC are all source-dependent.
\end{enumerate}

So, the algebra for parallelism is a generalization of BARTC, that is, BARTC is an embedding of the algebra for parallelism, as desired.
\end{proof}

\begin{theorem}[Soundness of parallelism modulo FR step bisimulation equivalence]\label{SPSBE}
Let $x$ and $y$ be APRTC terms. If $APRTC\vdash x=y$, then $x\sim_{s}^{fr} y$.
\end{theorem}

\begin{proof}
Since FR step bisimulation $\sim_s^{fr}$ is both an equivalent and a congruent relation with respect to the operators $\between$, $\parallel$, $\mid$, $\Theta$ and $\triangleleft$, we only need to check if each axiom in Table \ref{AxiomsForParallelism} is sound modulo FR step bisimulation equivalence.

The proof is similar to the proof of soundness of BARTC modulo FR step bisimulation equivalence, we omit it.
\end{proof}

\begin{theorem}[Completeness of parallelism modulo FR step bisimulation equivalence]\label{CPSBE}
Let $p$ and $q$ be closed APRTC terms, if $p\sim_{s}^{fr} q$ then $p=q$.
\end{theorem}

\begin{proof}
Firstly, by the elimination theorem of APRTC (see Theorem \ref{ETParallelism}), we know that for each closed APRTC term $p$, there exists a closed basic APRTC term $p'$, such that $APRTC\vdash p=p'$, so, we only need to consider closed basic APRTC terms.

The basic terms (see Definition \ref{BTAPRTC}) modulo associativity and commutativity (AC) of conflict $+$ (defined by axioms $A1$ and $A2$ in Table \ref{AxiomsForBARTC}) and associativity and commutativity (AC) of parallel $\parallel$ (defined by axioms $P2$ and $P3$ in Table \ref{AxiomsForParallelism}), and these equivalences is denoted by $=_{AC}$. Then, each equivalence class $s$ modulo AC of $+$ and $\parallel$ has the following normal form

$$s_1+\cdots+ s_k$$

with each $s_i$ either an atomic event or of the form

$$t_1\cdot\cdots\cdot t_m$$

with each $t_j$ either an atomic event or of the form

$$u_1\parallel\cdots\parallel u_n$$

with each $u_l$ an atomic event, and each $s_i$ is called the summand of $s$.

Now, we prove that for normal forms $n$ and $n'$, if $n\sim_{s}^{fr} n'$ then $n=_{AC}n'$. It is sufficient to induct on the sizes of $n$ and $n'$.

\begin{itemize}
  \item Consider a summand $e$ of $n$. Then $n\xrightarrow{e}e[m]$, so $n\sim_s^{fr} n'$ implies $n'\xrightarrow{e}e[m]$, meaning that $n'$ also contains the summand $e$.
  \item Consider a summand $e[m]$ of $n$. Then $n\xtworightarrow{e[m]}e$, so $n\sim_s^{fr} n'$ implies $n'\xtworightarrow{e[m]}e$, meaning that $n'$ also contains the summand $e[m]$.
  \item Consider a summand $t_1\cdot t_2$ of $n$,
  \begin{itemize}
    \item if $t_1\equiv e'$, then $n\xrightarrow{e'}e'[m]\cdot t_2$, so $n\sim_s^{fr} n'$ implies $n'\xrightarrow{e'}e'[m]\cdot t_2'$ with $e'[m]\cdot t_2\sim_s^{fr} e'[m]\cdot t_2'$, meaning that $n'$ contains a summand $e'\cdot t_2'$. Since $t_2$ and $t_2'$ are normal forms and have sizes smaller than $n$ and $n'$, by the induction hypotheses if $t_2\sim_s^{fr} t_2'$ then $t_2=_{AC} t_2'$;
    \item if $t_2\equiv e'[m]$, then $n\xtworightarrow{e'[m]}t_1\cdot e'$, so $n\sim_s^{fr} n'$ implies $n'\xrightarrow{e'[m]} t_1'\cdot e'$ with $t_1\cdot e'\sim_s^{fr} t_1'\cdot e'$, meaning that $n'$ contains a summand $t_1'\cdot e'$. Since $t_1$ and $t_1'$ are normal forms and have sizes smaller than $n$ and $n'$, by the induction hypotheses if $t_1\sim_s^{fr} t_1'$ then $t_1=_{AC} t_1'$;
    \item if $t_1\equiv e_1\parallel\cdots\parallel e_n$, then $n\xrightarrow{\{e_1,\cdots,e_n\}}(e_1[m]\parallel\cdots\parallel e_n[m])\cdot t_2$, so $n\sim_s^{fr} n'$ implies $n'\xrightarrow{\{e_1,\cdots,e_n\}}(e_1[m]\parallel\cdots\parallel e_n[m])\cdot t_2'$ with $t_2\sim_s^{fr} t_2'$, meaning that $n'$ contains a summand $(e_1\parallel\cdots\parallel e_n)\cdot t_2'$. Since $t_2$ and $t_2'$ are normal forms and have sizes smaller than $n$ and $n'$, by the induction hypotheses if $t_2\sim_s^{fr} t_2'$ then $t_2=_{AC} t_2'$.
    \item if $t_2\equiv e_1[m]\parallel\cdots\parallel e_n[m]$, then $n\xtworightarrow{\{e_1[m],\cdots,e_n[m]\}}(t_1\cdot e_1\parallel\cdots\parallel e_n)$, so $n\sim_s^{fr} n'$ implies $n'\xtworightarrow{\{e_1[m],\cdots,e_n[m]\}}t_1'\cdot(e_1\parallel\cdots\parallel e_n)$ with $t_1\sim_s^{fr} t_1'$, meaning that $n'$ contains a summand $t_1'\cdot(e_1[m]\parallel\cdots\parallel e_n[m])$. Since $t_1$ and $t_1'$ are normal forms and have sizes smaller than $n$ and $n'$, by the induction hypotheses if $t_1\sim_s^{fr} t_1'$ then $t_1=_{AC} t_1'$.
  \end{itemize}
\end{itemize}

So, we get $n=_{AC} n'$.

Finally, let $s$ and $t$ be basic APRTC terms, and $s\sim_s^{fr} t$, there are normal forms $n$ and $n'$, such that $s=n$ and $t=n'$. The soundness theorem of parallelism modulo FR step bisimulation equivalence (see Theorem \ref{SPSBE}) yields $s\sim_s^{fr} n$ and $t\sim_s^{fr} n'$, so $n\sim_s^{fr} s\sim_s^{fr} t\sim_s^{fr} n'$. Since if $n\sim_s^{fr} n'$ then $n=_{AC}n'$, $s=n=_{AC}n'=t$, as desired.
\end{proof}

\begin{theorem}[Soundness of parallelism modulo FR pomset bisimulation equivalence]\label{SPPBE}
Let $x$ and $y$ be APRTC terms. If $APRTC\vdash x=y$, then $x\sim_p^{fr} y$.
\end{theorem}

\begin{proof}
Since FR pomset bisimulation $\sim_p^{fr}$ is both an equivalent and a congruent relation with respect to the operators $\between$, $\parallel$, $\mid$, $\Theta$ and $\triangleleft$, we only need to check if each axiom in Table \ref{AxiomsForParallelism} is sound modulo FR pomset bisimulation equivalence.

From the definition of FR pomset bisimulation (see Definition \ref{FRPSB}), we know that FR pomset bisimulation is defined by pomset transitions, which are labeled by pomsets. In a pomset transition, the events in the pomset are either within causality relations (defined by $\cdot$) or in concurrency (implicitly defined by $\cdot$ and $+$, and explicitly defined by $\between$), of course, they are pairwise consistent (without conflicts). In Theorem \ref{SPSBE}, we have already proven the case that all events are pairwise concurrent, so, we only need to prove the case of events in causality. Without loss of generality, we take a pomset of $P=\{e_1,e_2:e_1\cdot e_2\}$. Then the pomset transition labeled by the above $P$ is just composed of one single event transition labeled by $e_1$ succeeded by another single event transition labeled by $e_2$, that is, $\xrightarrow{P}=\xrightarrow{e_1}\xrightarrow{e_2}$ or $\xrightarrow{P}=\xtworightarrow{e_2[n]}\xtworightarrow{e_1[m]}$.

Similarly to the proof of soundness of parallelism modulo FR step bisimulation equivalence (see Theorem \ref{SPSBE}), we can prove that each axiom in Table \ref{AxiomsForParallelism} is sound modulo FR pomset bisimulation equivalence, we omit them.
\end{proof}

\begin{theorem}[Completeness of parallelism modulo FR pomset bisimulation equivalence]\label{CPPBE}
Let $p$ and $q$ be closed APRTC terms, if $p\sim_p^{fr} q$ then $p=q$.
\end{theorem}

\begin{proof}
The proof is similar to the proof of completeness of parallelism modulo FR step bisimulation equivalence, we omit it.
\end{proof}

\begin{theorem}[Soundness of parallelism modulo FR hp-bisimulation equivalence]\label{SPHPBE}
Let $x$ and $y$ be APRTC terms. If $APRTC\vdash x=y$, then $x\sim_{hp}^{fr} y$.
\end{theorem}

\begin{proof}
Since FR hp-bisimulation $\sim_{hp}^{fr}$ is both an equivalent and a congruent relation with respect to the operators $\between$, $\parallel$, $\mid$, $\Theta$ and $\triangleleft$, we only need to check if each axiom in Table \ref{AxiomsForParallelism} is sound modulo FR hp-bisimulation equivalence.

From the definition of FR hp-bisimulation (see Definition \ref{FRHHPB}), we know that FR hp-bisimulation is defined on the posetal product $(C_1,f,C_2),f:C_1\rightarrow C_2\textrm{ isomorphism}$. Two process terms $s$ related to $C_1$ and $t$ related to $C_2$, and $f:C_1\rightarrow C_2\textrm{ isomorphism}$. Initially, $(C_1,f,C_2)=(\emptyset,\emptyset,\emptyset)$, and $(\emptyset,\emptyset,\emptyset)\in\sim_{hp}^{fr}$. When $s\xrightarrow{e}s'$ ($C_1\xrightarrow{e}C_1'$), there will be $t\xrightarrow{e}t'$ ($C_2\xrightarrow{e}C_2'$), and we define $f'=f[e\mapsto e]$. And when $s\xtworightarrow{e[m]}s'$ ($C_1\xtworightarrow{e[m]}C_1'$), there will be $t\xtworightarrow{e[m]}t'$ ($C_2\xtworightarrow{e[m]}C_2'$), and we define $f'=f[e[m]\mapsto e[m]]$. Then, if $(C_1,f,C_2)\in\sim_{hp}^{fr}$, then $(C_1',f',C_2')\in\sim_{hp}^{fr}$.

Similarly to the proof of soundness of parallelism modulo FR pomset bisimulation equivalence (see Theorem \ref{SPPBE}), we can prove that each axiom in Table \ref{AxiomsForParallelism} is sound modulo FR hp-bisimulation equivalence, we just need additionally to check the above conditions on FR hp-bisimulation, we omit them.
\end{proof}

\begin{theorem}[Completeness of parallelism modulo FR hp-bisimulation equivalence]\label{CPHPBE}
Let $p$ and $q$ be closed APRTC terms, if $p\sim_{hp}^{fr} q$ then $p=q$.
\end{theorem}

\begin{proof}
The proof is similar to the proof of completeness of parallelism modulo FR pomset bisimulation equivalence, we omit it.
\end{proof}

\subsection{Encapsulation}

The mismatch of two communicating events in different parallel branches can cause deadlock, so the deadlocks in the concurrent processes should be eliminated. Like $APTC$ \cite{APTC}, we also introduce the unary encapsulation operator $\partial_H$ for set $H$ of atomic events, which renames all atomic events in $H$ into $\delta$. The whole algebra including parallelism for true concurrency in the above subsections, deadlock $\delta$ and encapsulation operator $\partial_H$, is called Reversible Algebra for Parallelism in True Concurrency, abbreviated APRTC.

The forward transition rules of encapsulation operator $\partial_H$ are shown in Table \ref{TRForEncapsulation}, and the reverse transition rules of encapsulation operator $\partial_H$ are shown in Table \ref{RTRForEncapsulation}.

\begin{center}
    \begin{table}
        $$\frac{x\xrightarrow{e}e[m]}{\partial_H(x)\xrightarrow{e}\partial_H(e[m])}\quad (e\notin H)\quad\frac{x\xrightarrow{e}x'}{\partial_H(x)\xrightarrow{e}\partial_H(x')}\quad(e\notin H)$$
        \caption{Forward transition rules of encapsulation operator $\partial_H$}
        \label{TRForEncapsulation}
    \end{table}
\end{center}

\begin{center}
    \begin{table}
        $$\frac{x\xtworightarrow{e[m]}e}{\partial_H(x)\xtworightarrow{e[m]}e}\quad (e\notin H)\quad\quad\frac{x\xtworightarrow{e}x'}{\partial_H(x)\xtworightarrow{e}\partial_H(x')}\quad(e\notin H)$$
        \caption{Reverse transition rules of encapsulation operator $\partial_H$}
        \label{RTRForEncapsulation}
    \end{table}
\end{center}

Based on the transition rules for encapsulation operator $\partial_H$ in Table \ref{TRForEncapsulation} and Table \ref{RTRForEncapsulation}, we design the axioms as Table \ref{AxiomsForEncapsulation} shows.

\begin{center}
    \begin{table}
        \begin{tabular}{@{}ll@{}}
            \hline No. &Axiom\\
            $D1$ & $e\notin H\quad\partial_H(e) = e$\\
            $RD1$ & $e\notin H\quad\partial_H(e[m]) = e[m]$\\
            $D2$ & $e\in H\quad \partial_H(e) = \delta$\\
            $RD2$ & $e\in H\quad \partial_H(e[m]) = \delta$\\
            $D3$ & $\partial_H(\delta) = \delta$\\
            $D4$ & $\partial_H(x+ y) = \partial_H(x)+\partial_H(y)$\\
            $D5$ & $\partial_H(x\cdot y) = \partial_H(x)\cdot\partial_H(y)$\\
            $D6$ & $\partial_H(x\parallel y) = \partial_H(x)\parallel\partial_H(y)$\\
        \end{tabular}
        \caption{Axioms of encapsulation operator}
        \label{AxiomsForEncapsulation}
    \end{table}
\end{center}

\begin{theorem}[Conservativity of APRTC with respect to the algebra for parallelism]
APRTC is a conservative extension of the algebra for parallelism.
\end{theorem}

\begin{proof}
It follows from the following two facts (see Theorem \ref{TCE}).

\begin{enumerate}
  \item The transition rules of the algebra for parallelism in the above subsections are all source-dependent;
  \item The sources of the transition rules for the encapsulation operator contain an occurrence of $\partial_H$.
\end{enumerate}

So, APRTC is a conservative extension of the algebra for parallelism, as desired.
\end{proof}

\begin{theorem}[Congruence theorem of encapsulation operator $\partial_H$]
Truly concurrent bisimulation equivalences $\sim_p^{fr}$, $\sim_s^{fr}$, $\sim_{hp}^{fr}$ and $\sim_{hhp}^{fr}$ are all congruences with respect to encapsulation operator $\partial_H$.
\end{theorem}

\begin{proof}
(1) Case FR pomset bisimulation equivalence $\sim_p^{fr}$.

Let $x$ and $y$ be APRTC processes, and $x\sim_{p}^{fr} y$, it is sufficient to prove that $\partial_H(x)\sim_{p}^{fr} \partial_H(y)$.

By the definition of FR pomset bisimulation $\sim_p^{fr}$ (Definition \ref{FRPSB}), $x\sim_p^{fr} y$ means that

$$x\xrightarrow{X} x' \quad y\xrightarrow{Y} y'$$

$$x\xtworightarrow{X[\mathcal{K}]} x' \quad y\xtworightarrow{Y[\mathcal{J}]} y'$$

with $X\subseteq x$, $Y\subseteq y$, $X\sim Y$ and $x'\sim_p^{fr} y'$.

By the FR pomset transition rules for encapsulation operator $\partial_H$ in Table \ref{TRForEncapsulation} and Table \ref{RTRForEncapsulation}, we can get

$$\partial_H(x)\xrightarrow{X} \partial_H(X[\mathcal{K}]) (X\nsubseteq H) \quad \partial_H(y)\xrightarrow{Y} \partial_H(Y[\mathcal{J}]) (Y\nsubseteq H)$$

$$\partial_H(x)\xtworightarrow{X[\mathcal{K}]} \partial_H(X) (X\nsubseteq H) \quad \partial_H(y)\xtworightarrow{Y[\mathcal{J}]} \partial_H(Y) (Y\nsubseteq H)$$

with $X\subseteq x$, $Y\subseteq y$, and $X\sim Y$, and the assumptions $\partial_H(X[\mathcal{K}])\sim_p^{fr} \partial_H(Y[\mathcal{J}])$, $\partial_H(X)\sim_p^{fr} \partial_H(Y)$ so, we get $\partial_H(x)\sim_p^{fr} \partial_H(y)$, as desired.

Or, we can get

$$\partial_H(x)\xrightarrow{X} \partial_H(x') (X\nsubseteq H) \quad \partial_H(y)\xrightarrow{Y} \partial_H(y') (Y\nsubseteq H)$$

$$\partial_H(x)\xtworightarrow{X} \partial_H(x') (X\nsubseteq H) \quad \partial_H(y)\xtworightarrow{Y} \partial_H(y') (Y\nsubseteq H)$$

with $X\subseteq x$, $Y\subseteq y$, $X\sim Y$, $x'\sim_p^{fr} y'$ and the assumption $\partial_H(x')\sim_p^{fr}\partial_H(y')$, so, we get $\partial_H(x)\sim_p^{fr} \partial_H(y)$, as desired.

(2) The cases of FR step bisimulation $\sim_s^{fr}$, FR hp-bisimulation $\sim_{hp}^{fr}$ and FR hhp-bisimulation $\sim_{hhp}^{fr}$ can be proven similarly, we omit them.
\end{proof}

\begin{theorem}[Elimination theorem of APRTC]\label{ETEncapsulation}
Let $p$ be a closed APRTC term including the encapsulation operator $\partial_H$. Then there is a basic APRTC term $q$ such that $APRTC\vdash p=q$.
\end{theorem}

\begin{proof}
(1) Firstly, suppose that the following ordering on the signature of APRTC is defined: $\parallel > \cdot > +$ and the symbol $\parallel$ is given the lexicographical status for the first argument, then for each rewrite rule $p\rightarrow q$ in Table \ref{TRSForEncapsulation} relation $p>_{lpo} q$ can easily be proved. We obtain that the term rewrite system shown in Table \ref{TRSForEncapsulation} is strongly normalizing, for it has finitely many rewriting rules, and $>$ is a well-founded ordering on the signature of APRTC, and if $s>_{lpo} t$, for each rewriting rule $s\rightarrow t$ is in Table \ref{TRSForEncapsulation} (see Theorem \ref{SN}).

\begin{center}
    \begin{table}
        \begin{tabular}{@{}ll@{}}
            \hline No. &Rewriting Rule\\
            $RD1$ & $e\notin H\quad\partial_H(e) \rightarrow e$\\
            $RRD1$ & $e\notin H\quad\partial_H(e[m]) \rightarrow e[m]$\\
            $RD2$ & $e\in H\quad \partial_H(e) \rightarrow \delta$\\
            $RRD2$ & $e\in H\quad \partial_H(e[m]) \rightarrow \delta$\\
            $RD3$ & $\partial_H(\delta) \rightarrow \delta$\\
            $RD4$ & $\partial_H(x+ y) \rightarrow \partial_H(x)+\partial_H(y)$\\
            $RD5$ & $\partial_H(x\cdot y) \rightarrow \partial_H(x)\cdot\partial_H(y)$\\
            $RD6$ & $\partial_H(x\parallel y) \rightarrow \partial_H(x)\parallel\partial_H(y)$\\
        \end{tabular}
        \caption{Term rewrite system of encapsulation operator $\partial_H$}
        \label{TRSForEncapsulation}
    \end{table}
\end{center}

(2) Then we prove that the normal forms of closed APRTC terms including encapsulation operator $\partial_H$ are basic APRTC terms.

Suppose that $p$ is a normal form of some closed APRTC term and suppose that $p$ is not a basic APRTC term. Let $p'$ denote the smallest sub-term of $p$ which is not a basic APRTC term. It implies that each sub-term of $p'$ is a basic APRTC term. Then we prove that $p$ is not a term in normal form. It is sufficient to induct on the structure of $p'$, following from Theorem \ref{ETParallelism}, we only prove the new case $p'\equiv \partial_H(p_1)$:

\begin{itemize}
  \item Case $p_1\equiv e$. The transition rules $RD1$ or $RD2$ can be applied, so $p$ is not a normal form;
  \item Case $p_1\equiv e[m]$. The transition rules $RRD1$ or $RRD2$ can be applied, so $p$ is not a normal form;
  \item Case $p_1\equiv \delta$. The transition rules $RD3$ can be applied, so $p$ is not a normal form;
  \item Case $p_1\equiv p_1'+ p_1''$. The transition rules $RD4$ can be applied, so $p$ is not a normal form;
  \item Case $p_1\equiv p_1'\cdot p_1''$. The transition rules $RD5$ can be applied, so $p$ is not a normal form;
  \item Case $p_1\equiv p_1'\parallel p_1''$. The transition rules $RD6$ can be applied, so $p$ is not a normal form.
\end{itemize}
\end{proof}

\begin{theorem}[Soundness of APRTC modulo FR step bisimulation equivalence]\label{SAPRTCSBE}
Let $x$ and $y$ be APRTC terms including encapsulation operator $\partial_H$. If $APRTC\vdash x=y$, then $x\sim_{s}^{fr} y$.
\end{theorem}

\begin{proof}
Since FR step bisimulation $\sim_s^{fr}$ is both an equivalent and a congruent relation with respect to the operator $\partial_H$, we only need to check if each axiom in Table \ref{AxiomsForEncapsulation} is sound modulo FR step bisimulation equivalence.

The proof is similar to the proof of soundness of the algebra of parallelism modulo FR step bisimulation equivalence, we omit it.
\end{proof}

\begin{theorem}[Completeness of APRTC modulo FR step bisimulation equivalence]\label{CAPRTCSBE}
Let $p$ and $q$ be closed APRTC terms including encapsulation operator $\partial_H$, if $p\sim_{s}^{fr} q$ then $p=q$.
\end{theorem}

\begin{proof}
Firstly, by the elimination theorem of APRTC (see Theorem \ref{ETEncapsulation}), we know that the normal form of APRTC does not contain $\partial_H$, and for each closed APRTC term $p$, there exists a closed basic APRTC term $p'$, such that $APRTC\vdash p=p'$, so, we only need to consider closed basic APRTC terms.

Similarly to Theorem \ref{CPSBE}, we can prove that for normal forms $n$ and $n'$, if $n\sim_{s}^{fr} n'$ then $n=_{AC}n'$.

Finally, let $s$ and $t$ be basic APRTC terms, and $s\sim_s^{fr} t$, there are normal forms $n$ and $n'$, such that $s=n$ and $t=n'$. The soundness theorem of APRTC modulo FR step bisimulation equivalence (see Theorem \ref{SAPRTCSBE}) yields $s\sim_s^{fr} n$ and $t\sim_s^{fr} n'$, so $n\sim_s^{fr} s\sim_s^{fr} t\sim_s^{fr} n'$. Since if $n\sim_s^{fr} n'$ then $n=_{AC}n'$, $s=n=_{AC}n'=t$, as desired.
\end{proof}

\begin{theorem}[Soundness of APRTC modulo FR pomset bisimulation equivalence]\label{SAPRTCPBE}
Let $x$ and $y$ be APRTC terms including encapsulation operator $\partial_H$. If $APRTC\vdash x=y$, then $x\sim_p^{fr} y$.
\end{theorem}

\begin{proof}
Since FR pomset bisimulation $\sim_p^{fr}$ is both an equivalent and a congruent relation with respect to the operator $\partial_H$, we only need to check if each axiom in Table \ref{AxiomsForEncapsulation} is sound modulo FR pomset bisimulation equivalence.

From the definition of FR pomset bisimulation (see Definition \ref{FRPSB}), we know that FR pomset bisimulation is defined by pomset transitions, which are labeled by pomsets. In a pomset transition, the events in the pomset are either within causality relations (defined by $\cdot$) or in concurrency (implicitly defined by $\cdot$ and $+$, and explicitly defined by $\between$), of course, they are pairwise consistent (without conflicts). In Theorem \ref{SAPTCSBE}, we have already proven the case that all events are pairwise concurrent, so, we only need to prove the case of events in causality. Without loss of generality, we take a pomset of $P=\{e_1,e_2:e_1\cdot e_2\}$. Then the pomset transition labeled by the above $P$ is just composed of one single event transition labeled by $e_1$ succeeded by another single event transition labeled by $e_2$, that is, $\xrightarrow{P}=\xrightarrow{e_1}\xrightarrow{e_2}$ or $\xrightarrow{P}=\xtworightarrow{e_2[n]}\xtworightarrow{e_1[m]}$.

Similarly to the proof of soundness of APRTC modulo FR step bisimulation equivalence (see Theorem \ref{SAPRTCSBE}), we can prove that each axiom in Table \ref{AxiomsForEncapsulation} is sound modulo FR pomset bisimulation equivalence, we omit them.
\end{proof}

\begin{theorem}[Completeness of APRTC modulo FR pomset bisimulation equivalence]\label{CAPRTCPBE}
Let $p$ and $q$ be closed APRTC terms including encapsulation operator $\partial_H$, if $p\sim_p^{fr} q$ then $p=q$.
\end{theorem}

\begin{proof}
The proof can be proven similarly to the proof of completeness of APRTC modulo FR step bisimulation equivalence, we omit it.
\end{proof}

\begin{theorem}[Soundness of APRTC modulo FR hp-bisimulation equivalence]\label{SAPRTCHPBE}
Let $x$ and $y$ be APRTC terms including encapsulation operator $\partial_H$. If $APRTC\vdash x=y$, then $x\sim_{hp}^{fr} y$.
\end{theorem}

\begin{proof}
Since FR hp-bisimulation $\sim_{hp}^{fr}$ is both an equivalent and a congruent relation with respect to the operator $\partial_H$, we only need to check if each axiom in Table \ref{AxiomsForEncapsulation} is sound modulo FR hp-bisimulation equivalence.

From the definition of FR hp-bisimulation (see Definition \ref{FRHHPB}), we know that FR hp-bisimulation is defined on the posetal product $(C_1,f,C_2),f:C_1\rightarrow C_2\textrm{ isomorphism}$. Two process terms $s$ related to $C_1$ and $t$ related to $C_2$, and $f:C_1\rightarrow C_2\textrm{ isomorphism}$. Initially, $(C_1,f,C_2)=(\emptyset,\emptyset,\emptyset)$, and $(\emptyset,\emptyset,\emptyset)\in\sim_{hp}^{fr}$. When $s\xrightarrow{e}s'$ ($C_1\xrightarrow{e}C_1'$), there will be $t\xrightarrow{e}t'$ ($C_2\xrightarrow{e}C_2'$), and we define $f'=f[e\mapsto e]$. And when $s\xtworightarrow{e[m]}s'$ ($C_1\xtworightarrow{e[m]}C_1'$), there will be $t\xtworightarrow{e[m]}t'$ ($C_2\xtworightarrow{e[m]}C_2'$), and we define $f'=f[e[m]\mapsto e[m]]$. Then, if $(C_1,f,C_2)\in\sim_{hp}^{fr}$, then $(C_1',f',C_2')\in\sim_{hp}^{fr}$.

Similarly to the proof of soundness of APRTC modulo FR pomset bisimulation equivalence (see Theorem \ref{SAPRTCPBE}), we can prove that each axiom in Table \ref{AxiomsForEncapsulation} is sound modulo FR hp-bisimulation equivalence, we just need additionally to check the above conditions on FR hp-bisimulation, we omit them.
\end{proof}

\begin{theorem}[Completeness of APRTC modulo FR hp-bisimulation equivalence]\label{CAPRTCHPBE}
Let $p$ and $q$ be closed APRTC terms including encapsulation operator $\partial_H$, if $p\sim_{hp}^{fr} q$ then $p=q$.
\end{theorem}

\begin{proof}
The proof is similar to the proof of completeness of APRTC modulo FR pomset bisimulation equivalence, we omit it.
\end{proof}

\section{Recursion}\label{rec}

In this section, we introduce recursion to capture infinite processes based on APRTC. In the following, $E,F,G$ are recursion specifications, $X,Y,Z$ are recursive variables.

The behavior of the solution $\langle X_i|E\rangle$ for the recursion variable $X_i$ in $E$, where $i\in\{1,\cdots,n\}$, is exactly the behavior of their right-hand sides $t_i(X_1,\cdots,X_n)$, which is captured by the two transition rules in Table \ref{TRForGR}.

\begin{center}
    \begin{table}
        $$\frac{t_i(\langle X_1|E\rangle,\cdots,\langle X_n|E\rangle)\xrightarrow{e}\surd}{\langle X_i|E\rangle\xrightarrow{e}\surd}$$
        $$\frac{t_i(\langle X_1|E\rangle,\cdots,\langle X_n|E\rangle)\xrightarrow{e} y}{\langle X_i|E\rangle\xrightarrow{e} y}$$
        \caption{Transition rules of guarded recursion}
        \label{TRForGR}
    \end{table}
\end{center}

\begin{theorem}[Conservitivity of APRTC with guarded recursion]
APRTC with guarded recursion is a conservative extension of APRTC.
\end{theorem}

\begin{proof}
Since the transition rules of APRTC are source-dependent, and the transition rules for guarded recursion in Table \ref{TRForGR} contain only a fresh constant in their source, so the transition rules of APRTC with guarded recursion are a conservative extension of those of APRTC.
\end{proof}

\begin{theorem}[Congruence theorem of APRTC with guarded recursion]
Truly concurrent bisimulation equivalences $\sim_p^{fr}$, $\sim_s^{fr}$ and $\sim_{hp}^{fr}$ are all congruences with respect to APRTC with guarded recursion.
\end{theorem}

\begin{proof}
It follows the following two facts:
\begin{enumerate}
  \item in a guarded recursive specification, right-hand sides of its recursive equations can be adapted to the form by applications of the axioms in APRTC and replacing recursion variables by the right-hand sides of their recursive equations;
  \item truly concurrent bisimulation equivalences $\sim_p^{fr}$, $\sim_s^{fr}$ and $\sim_{hp}^{fr}$ are all congruences with respect to all operators of APRTC.
\end{enumerate}
\end{proof}

\subsection{Recursive Definition and Specification Principles}

The RDP (Recursive Definition Principle) and the RSP (Recursive Specification Principle) are shown in Table \ref{RDPRSP}.

\begin{center}
\begin{table}
  \begin{tabular}{@{}ll@{}}
\hline No. &Axiom\\
  RDP & $\langle X_i|E\rangle = t_i(\langle X_1|E,\cdots,X_n|E\rangle)\quad (i\in\{1,\cdots,n\})$\\
  RSP & if $y_i=t_i(y_1,\cdots,y_n)$ for $i\in\{1,\cdots,n\}$, then $y_i=\langle X_i|E\rangle \quad(i\in\{1,\cdots,n\})$\\
\end{tabular}
\caption{Recursive definition and specification principle}
\label{RDPRSP}
\end{table}
\end{center}

\begin{theorem}[Elimination theorem of APRTC with linear recursion]\label{ETRecursion}
Each process term in APRTC with linear recursion is equal to a process term $\langle X_1|E\rangle$ with $E$ a linear recursive specification.
\end{theorem}

\begin{proof}
By applying structural induction with respect to term size, each process term $t_1$ in APRTC with linear recursion generates a process can be expressed in the form of equations

$$t_i=(a_{i11}\parallel\cdots\parallel a_{i1i_1})t_{i1}+\cdots+(a_{ik_i1}\parallel\cdots\parallel a_{ik_ii_k})t_{ik_i}+(b_{i11}\parallel\cdots\parallel b_{i1i_1})+\cdots+(b_{il_i1}\parallel\cdots\parallel b_{il_ii_l})$$

for $i\in\{1,\cdots,n\}$. Or,

$$t_i=t_{i1}(a_{i11}[m_{i1}]\parallel\cdots\parallel a_{i1i_1}[m_{i1}])+\cdots+t_{ik_i}(a_{ik_i1}[m_{ik}]\parallel\cdots\parallel a_{ik_ii_k}[m_{ik}])+(b_{i11}[n_{i1}]\parallel\cdots\parallel b_{i1i_1})[n_{i1}]+\cdots+(b_{il_i1}[n_{il}]\parallel\cdots\parallel b_{il_ii_l}[n_{il}])$$

Let the linear recursive specification $E$ consist of the recursive equations

$$X_i=(a_{i11}\parallel\cdots\parallel a_{i1i_1})X_{i1}+\cdots+(a_{ik_i1}\parallel\cdots\parallel a_{ik_ii_k})X_{ik_i}+(b_{i11}\parallel\cdots\parallel b_{i1i_1})+\cdots+(b_{il_i1}\parallel\cdots\parallel b_{il_ii_l})$$

or the equations,

$$X_i=X_{i1}(a_{i11}[m_{i1}]\parallel\cdots\parallel a_{i1i_1}[m_{i1}])+\cdots+X_{ik_i}(a_{ik_i1}[m_{ik}]\parallel\cdots\parallel a_{ik_ii_k}[m_{ik}])+(b_{i11}[n_{i1}]\parallel\cdots\parallel b_{i1i_1}[n_{i1}])+\cdots+(b_{il_i1}[n_{il}]\parallel\cdots\parallel b_{il_ii_l}[n_{il}])$$

for $i\in\{1,\cdots,n\}$. Replacing $X_i$ by $t_i$ for $i\in\{1,\cdots,n\}$ is a solution for $E$, RSP yields $t_1=\langle X_1|E\rangle$.
\end{proof}

\begin{theorem}[Soundness of APRTC with guarded recursion]\label{SAPRTCR}
Let $x$ and $y$ be APRTC with guarded recursion terms. If $APRTC\textrm{ with guarded recursion}\vdash x=y$, then
\begin{enumerate}
  \item $x\sim_{s}^{fr} y$;
  \item $x\sim_p^{fr} y$;
  \item $x\sim_{hp}^{fr} y$.
\end{enumerate}
\end{theorem}

\begin{proof}
(1) Soundness of APRTC with guarded recursion with respect to FR step bisimulation $\sim_s^{fr}$.

Since FR step bisimulation $\sim_s^{fr}$ is both an equivalent and a congruent relation with respect to APRTC with guarded recursion, we only need to check if each axiom in Table \ref{RDPRSP} is sound modulo FR step bisimulation equivalence.

This can be proven similarly to the proof of soundness of APRTC modulo FR step bisimulation equivalence, we omit them.

(2) Soundness of APRTC with guarded recursion with respect to FR pomset bisimulation $\sim_p^{fr}$.

Since FR pomset bisimulation $\sim_p^{fr}$ is both an equivalent and a congruent relation with respect to the guarded recursion, we only need to check if each axiom in Table \ref{RDPRSP} is sound modulo FR pomset bisimulation equivalence.

From the definition of FR pomset bisimulation (see Definition \ref{FRPSB}), we know that FR pomset bisimulation is defined by pomset transitions, which are labeled by pomsets. In a pomset transition, the events in the pomset are either within causality relations (defined by $\cdot$) or in concurrency (implicitly defined by $\cdot$ and $+$, and explicitly defined by $\between$), of course, they are pairwise consistent (without conflicts). In Theorem \ref{SAPTCSBE}, we have already proven the case that all events are pairwise concurrent, so, we only need to prove the case of events in causality. Without loss of generality, we take a pomset of $P=\{e_1,e_2:e_1\cdot e_2\}$. Then the pomset transition labeled by the above $P$ is just composed of one single event transition labeled by $e_1$ succeeded by another single event transition labeled by $e_2$, that is, $\xrightarrow{P}=\xrightarrow{e_1}\xrightarrow{e_2}$ or $\xrightarrow{P}=\xtworightarrow{e_2[n]}\xtworightarrow{e_1[m]}$.

Similarly to the proof of soundness of APRTC with guarded recursion modulo FR step bisimulation equivalence (1), we can prove that each axiom in Table \ref{RDPRSP} is sound modulo FR pomset bisimulation equivalence, we omit them.

(3) Soundness of APRTC with guarded recursion with respect to FR hp-bisimulation $\sim_{hp}^{fr}$.

Since FR hp-bisimulation $\sim_{hp}^{fr}$ is both an equivalent and a congruent relation with respect to guarded recursion, we only need to check if each axiom in Table \ref{RDPRSP} is sound modulo FR hp-bisimulation equivalence.

From the definition of FR hp-bisimulation (see Definition \ref{FRHHPB}), we know that FR hp-bisimulation is defined on the posetal product $(C_1,f,C_2),f:C_1\rightarrow C_2\textrm{ isomorphism}$. Two process terms $s$ related to $C_1$ and $t$ related to $C_2$, and $f:C_1\rightarrow C_2\textrm{ isomorphism}$. Initially, $(C_1,f,C_2)=(\emptyset,\emptyset,\emptyset)$, and $(\emptyset,\emptyset,\emptyset)\in\sim_{hp}^{fr}$. When $s\xrightarrow{e}s'$ ($C_1\xrightarrow{e}C_1'$), there will be $t\xrightarrow{e}t'$ ($C_2\xrightarrow{e}C_2'$), and we define $f'=f[e\mapsto e]$. And when $s\xtworightarrow{e[m]}s'$ ($C_1\xtworightarrow{e[m]}C_1'$), there will be $t\xtworightarrow{e[m]}t'$ ($C_2\xtworightarrow{e[m]}C_2'$), and we define $f'=f[e[m]\mapsto e[m]]$. Then, if $(C_1,f,C_2)\in\sim_{hp}^{fr}$, then $(C_1',f',C_2')\in\sim_{hp}^{fr}$.

Similarly to the proof of soundness of APRTC with guarded recursion modulo FR pomset bisimulation equivalence (2), we can prove that each axiom in Table \ref{RDPRSP} is sound modulo FR hp-bisimulation equivalence, we just need additionally to check the above conditions on FR hp-bisimulation, we omit them.
\end{proof}

\begin{theorem}[Completeness of APRTC with linear recursion]\label{CAPRTCR}
Let $p$ and $q$ be closed APRTC with linear recursion terms, then,
\begin{enumerate}
  \item if $p\sim_{s}^{fr} q$ then $p=q$;
  \item if $p\sim_p^{fr} q$ then $p=q$;
  \item if $p\sim_{hp}^{fr} q$ then $p=q$.
\end{enumerate}
\end{theorem}

\begin{proof}
Firstly, by the elimination theorem of APRTC with guarded recursion (see Theorem \ref{ETRecursion}), we know that each process term in APRTC with linear recursion is equal to a process term $\langle X_1|E\rangle$ with $E$ a linear recursive specification.

It remains to prove the following cases.

(1) If $\langle X_1|E_1\rangle \sim_s^{fr} \langle Y_1|E_2\rangle$ for linear recursive specification $E_1$ and $E_2$, then $\langle X_1|E_1\rangle = \langle Y_1|E_2\rangle$.

Let $E_1$ consist of recursive equations $X=t_X$ for $X\in \mathcal{X}$ and $E_2$
consists of recursion equations $Y=t_Y$ for $Y\in\mathcal{Y}$. Let the linear recursive specification $E$ consist of recursion equations $Z_{XY}=t_{XY}$, and $\langle X|E_1\rangle\sim_s^{fr}\langle Y|E_2\rangle$, and $t_{XY}$ consists of the following summands:

\begin{enumerate}
  \item $t_{XY}$ contains a summand $(a_1\parallel\cdots\parallel a_m)Z_{X'Y'}$ iff $t_X$ contains the summand $(a_1\parallel\cdots\parallel a_m)X'$ and $t_Y$ contains the summand $(a_1\parallel\cdots\parallel a_m)Y'$ such that $\langle X'|E_1\rangle\sim_s^{fr}\langle Y'|E_2\rangle$;
  \item $t_{XY}$ contains a summand $Z_{X'Y'}(a_1[m]\parallel\cdots\parallel a_m[m])$ iff $t_X$ contains the summand $X'(a_1[m]\parallel\cdots\parallel a_m[m])$ and $t_Y$ contains the summand $Y'(a_1[m]\parallel\cdots\parallel a_m[m])$ such that $\langle X'|E_1\rangle\sim_s^{fr}\langle Y'|E_2\rangle$;
  \item $t_{XY}$ contains a summand $b_1\parallel\cdots\parallel b_n$ iff $t_X$ contains the summand $b_1\parallel\cdots\parallel b_n$ and $t_Y$ contains the summand $b_1\parallel\cdots\parallel b_n$;
  \item $t_{XY}$ contains a summand $b_1[n]\parallel\cdots\parallel b_n[n]$ iff $t_X$ contains the summand $b_1[n]\parallel\cdots\parallel b_n[n]$ and $t_Y$ contains the summand $b_1[n]\parallel\cdots\parallel b_n[n]$.
\end{enumerate}

Let $\sigma$ map recursion variable $X$ in $E_1$ to $\langle X|E_1\rangle$, and let $\psi$ map recursion variable $Z_{XY}$ in $E$ to $\langle X|E_1\rangle$. So, $\sigma((a_1\parallel\cdots\parallel a_m)X')\equiv(a_1\parallel\cdots\parallel a_m)\langle X'|E_1\rangle\equiv\psi((a_1\parallel\cdots\parallel a_m)Z_{X'Y'})$, or $\sigma(X'(a_1[m]\parallel\cdots\parallel a_m[m]))\equiv\langle X'|E_1\rangle(a_1[m]\parallel\cdots\parallel a_m[m])\equiv\psi(Z_{X'Y'}(a_1[m]\parallel\cdots\parallel a_m[m]))$, so by RDP, we get $\langle X|E_1\rangle=\sigma(t_X)=\psi(t_{XY})$. Then by RSP, $\langle X|E_1\rangle=\langle Z_{XY}|E\rangle$, particularly, $\langle X_1|E_1\rangle=\langle Z_{X_1Y_1}|E\rangle$. Similarly, we can obtain $\langle Y_1|E_2\rangle=\langle Z_{X_1Y_1}|E\rangle$. Finally, $\langle X_1|E_1\rangle=\langle Z_{X_1Y_1}|E\rangle=\langle Y_1|E_2\rangle$, as desired.

(2) If $\langle X_1|E_1\rangle \sim_p^{fr} \langle Y_1|E_2\rangle$ for linear recursive specification $E_1$ and $E_2$, then $\langle X_1|E_1\rangle = \langle Y_1|E_2\rangle$.

It can be proven similarly to (1), we omit it.

(3) If $\langle X_1|E_1\rangle \sim_{hp}^{fr} \langle Y_1|E_2\rangle$ for linear recursive specification $E_1$ and $E_2$, then $\langle X_1|E_1\rangle = \langle Y_1|E_2\rangle$.

It can be proven similarly to (1), we omit it.
\end{proof}

\section{Abstraction}\label{abs}

To abstract away from the internal implementations of a program, and verify that the program exhibits the desired external behaviors, the silent step $\tau$ and abstraction operator $\tau_I$ are introduced, where $I\subseteq \mathbb{E}$ denotes the internal events. The transition rule of $\tau$ is shown in Table \ref{TRForTau}. In the following, let the atomic event $e$ range over $\mathbb{E}\cup\{\delta\}\cup\{\tau\}$, and let the communication function $\gamma:\mathbb{E}\cup\{\tau\}\times \mathbb{E}\cup\{\tau\}\rightarrow \mathbb{E}\cup\{\delta\}$, with each communication involved $\tau$ resulting in $\delta$.

\begin{center}
    \begin{table}
        $$\frac{}{\tau\xrightarrow{\tau}\surd}$$
        $$\frac{}{\tau\xtworightarrow{\tau}\surd}$$
        \caption{Transition rule of the silent step}
        \label{TRForTau}
    \end{table}
\end{center}

\begin{theorem}[Conservitivity of $RAPTC$ with silent step]
$RAPTC$ with silent step is a conservative extension of $RAPTC$.
\end{theorem}

\begin{proof}
Since the transition rules of $RAPTC$ are source-dependent, and the transition rules for silent step in Table \ref{TRForTau} contain only a fresh constant $\tau$ in their source, so the transition rules of $RAPTC$ with silent step is a conservative extension of those of $RAPTC$.
\end{proof}

\begin{theorem}[Congruence theorem of $RAPTC$ with silent step]
Rooted branching FR truly concurrent bisimulation equivalences $\approx_{rbp}^{fr}$, $\approx_{rbs}^{fr}$ and $\approx_{rbhp}^{fr}$ are all congruences with respect to $RAPTC$ with silent step.
\end{theorem}

\begin{proof}
It follows the following two facts:
\begin{enumerate}
  \item FR truly concurrent bisimulation equivalences $\sim_{p}^{fr}$, $\sim_s^{fr}$ and $\sim_{hp}^{fr}$ are all congruences with respect to all operators of $RAPTC$, while FR truly concurrent bisimulation equivalences $\sim_{p}^{fr}$, $\sim_s^{fr}$ and $\sim_{hp}^{fr}$ imply the corresponding rooted branching FR truly concurrent bisimulation $\approx{rbp}^{fr}$, $\approx_{rbs}^{fr}$ and $\approx_{rbhp}^{fr}$, so rooted branching FR truly concurrent bisimulation $\approx{rbp}^{fr}$, $\approx_{rbs}^{fr}$ and $\approx_{rbhp}^{fr}$ are all congruences with respect to all operators of $RAPTC$;
  \item While $\mathbb{E}$ is extended to $\mathbb{E}\cup\{\tau\}$, it can be proved that rooted branching FR truly concurrent bisimulation $\approx{rbp}^{fr}$, $\approx_{rbs}^{fr}$ and $\approx_{rbhp}^{fr}$ are all congruences with respect to all operators of $RAPTC$, we omit it.
\end{enumerate}
\end{proof}

\subsection{Algebraic Laws for the Silent Step}

We design the axioms for the silent step $\tau$ in Table \ref{AxiomsForTau}.

\begin{center}
\begin{table}
  \begin{tabular}{@{}ll@{}}
\hline No. &Axiom\\
  $B1$ & $e\cdot\tau=e$\\
  $RB1$ & $\tau\cdot e[m]=e[m]$\\
  $B2$ & $e\cdot(\tau\cdot(x+y)+x)=e\cdot(x+y)$\\
  $RB2$ & $((x+y)\cdot\tau+x)\cdot e[m]=(x+y)\cdot e[m]$\\
  $B3$ & $x\parallel\tau=x$\\
\end{tabular}
\caption{Axioms of silent step}
\label{AxiomsForTau}
\end{table}
\end{center}

\begin{theorem}[Elimination theorem of APRTC with silent step and guarded linear recursion]\label{ETTau}
Each process term in APRTC with silent step and guarded linear recursion is equal to a process term $\langle X_1|E\rangle$ with $E$ a guarded linear recursive specification.
\end{theorem}

\begin{proof}
By applying structural induction with respect to term size, each process term $t_1$ in APRTC with silent step and guarded linear recursion generates a process can be expressed in the form of equations

$$t_i=(a_{i11}\parallel\cdots\parallel a_{i1i_1})t_{i1}+\cdots+(a_{ik_i1}\parallel\cdots\parallel a_{ik_ii_k})t_{ik_i}+(b_{i11}\parallel\cdots\parallel b_{i1i_1})+\cdots+(b_{il_i1}\parallel\cdots\parallel b_{il_ii_l})$$

Or,

$$t_i=t_{i1}(a_{i11}[m_{i1}]\parallel\cdots\parallel a_{i1i_1}[m_{i1}])+\cdots+t_{ik_i}(a_{ik_i1}[m_{ik}]\parallel\cdots\parallel a_{ik_ii_k}[m_{ik}])+(b_{i11}[n_{i1}]\parallel\cdots\parallel b_{i1i_1}[n_{i1}])+\cdots+(b_{il_i1}[n_{il}]\parallel\cdots\parallel b_{il_ii_l}[n_{il}])$$

for $i\in\{1,\cdots,n\}$. Let the linear recursive specification $E$ consist of the recursive equations

$$X_i=(a_{i11}\parallel\cdots\parallel a_{i1i_1})X_{i1}+\cdots+(a_{ik_i1}\parallel\cdots\parallel a_{ik_ii_k})X_{ik_i}+(b_{i11}\parallel\cdots\parallel b_{i1i_1})+\cdots+(b_{il_i1}\parallel\cdots\parallel b_{il_ii_l})$$

Or,

$$X_i=X_{i1}(a_{i11}[m_{i1}]\parallel\cdots\parallel a_{i1i_1}[m_{i1}])+\cdots+X_{ik_i}(a_{ik_i1}[m_{ik}]\parallel\cdots\parallel a_{ik_ii_k}[m_{ik}])+(b_{i11}[n_{i1}]\parallel\cdots\parallel b_{i1i_1}[n_{i1}])+\cdots+(b_{il_i1}[n_{il}]\parallel\cdots\parallel b_{il_ii_l}[n_{il}])$$

for $i\in\{1,\cdots,n\}$. Replacing $X_i$ by $t_i$ for $i\in\{1,\cdots,n\}$ is a solution for $E$, RSP yields $t_1=\langle X_1|E\rangle$.
\end{proof}

\begin{theorem}[Soundness of APRTC with silent step and guarded linear recursion]\label{SAPRTCTAU}
Let $x$ and $y$ be APRTC with silent step and guarded linear recursion terms. If APRTC with silent step and guarded linear recursion $\vdash x=y$, then
\begin{enumerate}
  \item $x\approx_{rbs}^{fr} y$;
  \item $x\approx_{rbp}^{fr} y$;
  \item $x\approx_{rbhp}^{fr} y$.
\end{enumerate}
\end{theorem}

\begin{proof}
(1) Soundness of APRTC with silent step and guarded linear recursion with respect to rooted branching FR step bisimulation $\approx_{rbs}^{fr}$.

Since rooted branching FR step bisimulation $\approx_{rbs}^{fr}$ is both an equivalent and a congruent relation with respect to APRTC with silent step and guarded linear recursion, we only need to check if each axiom in Table \ref{AxiomsForTau} is sound modulo rooted branching FR step bisimulation equivalence.

Though transition rules in Table \ref{TRForTau} are defined in the flavor of single event, they can be modified into a step (a set of events within which each event is pairwise concurrent), we omit them. If we treat a single event as a step containing just one event, the proof of this soundness theorem does not exist any problem, so we use this way and still use the transition rules in Table \ref{TRForTau}.

\begin{itemize}
  \item \textbf{Axiom $B1$}. Assume that $e\cdot\tau=e$, it is sufficient to prove that $e\cdot\tau\approx_{rbs}^{fr}e$. By the forward transition rules for operator $\cdot$ in Table \ref{SETRForBARTC} and $\tau$ in Table \ref{TRForTau}, we get

      $$\frac{e\xrightarrow{e}e[m]}{e\cdot\tau\xrightarrow{e}\xrightarrow{\tau}e[m]}$$

      $$\frac{e\xrightarrow{e}e[m]}{e\xrightarrow{e}e[m]}$$

      By the reverse transition rules for operator $\cdot$ in Table \ref{RSETRForBARTC} and $\tau$ in Table \ref{TRForTau}, there are no transitions.

      So, $e\cdot\tau\approx_{rbs}^{fr}e$, as desired.

  \item \textbf{Axiom $RB1$}. Assume that $\tau \cdot e[m]=e[m]$, it is sufficient to prove that $\tau \cdot e[m]\approx_{rbs}^{fr}e[m]$. By the forward transition rules for operator $\cdot$ in Table \ref{SETRForBARTC} and $\tau$ in Table \ref{TRForTau}, there are no transitions.

      By the reverse transition rules for operator $\cdot$ in Table \ref{RSETRForBARTC} and $\tau$ in Table \ref{TRForTau}, we get

      $$\frac{e[m]\xtworightarrow{e[m]}e}{\tau\cdot e[m]\xtworightarrow{e[m]}\xtworightarrow{\tau}e}$$

      $$\frac{e[m]\xtworightarrow{e[m]}e}{e[m]\xtworightarrow{e[m]}e}$$

      So, $\tau \cdot e[m]\approx_{rbs}^{fr}e[m]$, as desired.

  \item \textbf{Axiom $B2$}. Let $p$ and $q$ be $RAPTC$ with silent step processes, and assume that $e\cdot(\tau\cdot(p+q)+p)=e\cdot(p+q)$, it is sufficient to prove that $e\cdot(\tau\cdot(p+q)+p)\approx_{rbs}^{fr}e\cdot(p+q)$. There are several cases, we will not enumerate all. By the forward transition rules for operators $\cdot$ and $+$ in Table \ref{SETRForBARTC} and $\tau$ in Table \ref{TRForTau}, we get

      $$\frac{e\xrightarrow{e}e[m]\quad p\xrightarrow{e_1}p'\quad q\xrightarrow{e_1}q'}{e\cdot(\tau\cdot(p+q)+p)\xrightarrow{e}\xrightarrow{\tau}\xrightarrow{e_1}e[m]\cdot((p'+q')+p')}$$

      $$\frac{e\xrightarrow{e}e[m]\quad p\xrightarrow{e_1}p'}{e\cdot(p+q)\xrightarrow{e}\xrightarrow{e_1}e[m]\cdot(p'+q')}$$

      By the reverse transition rules for operators $\cdot$ and $+$ in Table \ref{RSETRForBARTC} and $\tau$ in Table \ref{TRForTau}, there are no transitions.

      So, $e\cdot(\tau\cdot(p+q)+p)\approx_{rbs}^{fr}e\cdot(p+q)$, as desired.

  \item \textbf{Axiom $RB2$}. Let $p$ and $q$ be $RAPTC$ with silent step processes, and assume that $((x+y)\cdot\tau+x)\cdot e[m]=(x+y)\cdot e[m]$, it is sufficient to prove that $((x+y)\cdot\tau+x)\cdot e[m]\approx_{rbs}^{fr}(x+y)\cdot e[m]$. There are several cases, we will not enumerate all. By the forward transition rules for operators $\cdot$ and $+$ in Table \ref{STRForBARTC} and $\tau$ in Table \ref{TRForTau}, there are no transitions.

      By the reverse transition rules for operators $\cdot$ and $+$ in Table \ref{RSETRForBARTC} and $\tau$ in Table \ref{TRForTau}, we get

      $$\frac{e[m]\xtworightarrow{e[m]}e\quad p\xtworightarrow{e_1[n]}p'\quad q\xtworightarrow{e_1[n]}q'}{((p+q)\cdot\tau+p)\cdot e[m]\xtworightarrow{e[m]}\xtworightarrow{\tau}\xtworightarrow{e_1[n]}((p'+q')+p')\cdot e}$$

      $$\frac{e[m]\xtworightarrow{e[m]}e\quad p\xtworightarrow{e_1[n]}p'}{(p+q)\cdot e[m]\xtworightarrow{e[m]}\xtworightarrow{e_1[n]}(p'+q'\cdot e)}$$

      So, $((p+q)\cdot\tau+p)\cdot e[m]\approx_{rbs}^{fr}(p+q)\cdot e[m]$, as desired.

  \item \textbf{Axiom $B3$}. Let $p$ be an $RAPTC$ with silent step, and assume that $p\parallel\tau=p$, it is sufficient to prove that $p\parallel\tau\approx_{rbs}^{fr}p$. By the forward transition rules for operator $\parallel$ in Table \ref{TRForParallel} and $\tau$ in Table \ref{TRForTau}, we get

      $$\frac{p\xrightarrow{e}e[m]}{p\parallel\tau\xRightarrow{e}e[m]}$$

      $$\frac{p\xrightarrow{e}p'}{p\parallel\tau\xRightarrow{e}p'}$$

      By the reverse transition rules for operator $\parallel$ in Table \ref{RTRForParallel} and $\tau$ in Table \ref{TRForTau}, we get

      $$\frac{p\xtworightarrow{e[m]}e}{p\parallel\tau\xTworightarrow{e[m]}e}$$

      $$\frac{p\xtworightarrow{e[m]}p'}{p\parallel\tau\xTworightarrow{e[m]}p'}$$

      So, $p\parallel\tau\approx_{rbs}^{fr}p$, as desired.
\end{itemize}
(2) Soundness of APRTC with silent step and guarded linear recursion with respect to rooted branching FR pomset bisimulation $\approx_{rbp}^{fr}$.

Since rooted branching FR pomset bisimulation $\approx_{rbp}^{fr}$ is both an equivalent and a congruent relation with respect to APRTC with silent step and guarded linear recursion, we only need to check if each axiom in Table \ref{AxiomsForTau} is sound modulo rooted branching FR pomset bisimulation $\approx_{rbp}^{fr}$.

From the definition of rooted branching FR pomset bisimulation $\approx_{rbp}^{fr}$ (see Definition \ref{FRRBPSB}), we know that rooted branching FR pomset bisimulation $\approx_{rbp}^{fr}$ is defined by weak pomset transitions, which are labeled by pomsets with $\tau$. In a weak pomset transition, the events in the pomset are either within causality relations (defined by $\cdot$) or in concurrency (implicitly defined by $\cdot$ and $+$, and explicitly defined by $\between$), of course, they are pairwise consistent (without conflicts). In (1), we have already proven the case that all events are pairwise concurrent, so, we only need to prove the case of events in causality. Without loss of generality, we take a pomset of $P=\{e_1,e_2:e_1\cdot e_2\}$. Then the weak pomset transition labeled by the above $P$ is just composed of one single event transition labeled by $e_1$ succeeded by another single event transition labeled by $e_2$, that is, $\xRightarrow{P}=\xRightarrow{e_1}\xRightarrow{e_2}$ or $\xTworightarrow{P}=\xTworightarrow{e_2}\xTworightarrow{e_1}$.

Similarly to the proof of soundness of APRTC with silent step and guarded linear recursion modulo rooted branching FR step bisimulation $\approx_{rbs}^{fr}$ (1), we can prove that each axiom in Table \ref{AxiomsForTau} is sound modulo rooted branching FR pomset bisimulation $\approx_{rbp}^{fr}$, we omit them.

(3) Soundness of APRTC with silent step and guarded linear recursion with respect to rooted branching FR hp-bisimulation $\approx_{rbhp}^{fr}$.

Since rooted branching FR hp-bisimulation $\approx_{rbhp}^{fr}$ is both an equivalent and a congruent relation with respect to APRTC with silent step and guarded linear recursion, we only need to check if each axiom in Table \ref{AxiomsForTau} is sound modulo rooted branching FR hp-bisimulation $\approx_{rbhp}^{fr}$.

From the definition of rooted branching FR hp-bisimulation $\approx_{rbhp}^{fr}$ (see Definition \ref{FRRBHHPB}), we know that rooted branching FR hp-bisimulation $\approx_{rbhp}^{fr}$ is defined on the weakly posetal product $(C_1,f,C_2),f:\hat{C_1}\rightarrow \hat{C_2}\textrm{ isomorphism}$. Two process terms $s$ related to $C_1$ and $t$ related to $C_2$, and $f:\hat{C_1}\rightarrow \hat{C_2}\textrm{ isomorphism}$. Initially, $(C_1,f,C_2)=(\emptyset,\emptyset,\emptyset)$, and $(\emptyset,\emptyset,\emptyset)\in\approx_{rbhp}^{fr}$. When $s\xrightarrow{e}s'$ ($C_1\xrightarrow{e}C_1'$), there will be $t\xRightarrow{e}t'$ ($C_2\xRightarrow{e}C_2'$), and we define $f'=f[e\mapsto e]$. And when $s\xTworightarrow{e[m]}s'$ ($C_1\xTworightarrow{e[m]}C_1'$), there will be $t\xTworightarrow{e[m]}t'$ ($C_2\xTworightarrow{e[m]}C_2'$), and we define $f'=f[e[m]\mapsto e[m]$. Then, if $(C_1,f,C_2)\in\approx_{rbhp}^{fr}$, then $(C_1',f',C_2')\in\approx_{rbhp}^{fr}$.

Similarly to the proof of soundness of APRTC with silent step and guarded linear recursion modulo rooted branching FR pomset bisimulation equivalence (2), we can prove that each axiom in Table \ref{AxiomsForTau} is sound modulo rooted branching FR hp-bisimulation equivalence, we just need additionally to check the above conditions on rooted branching FR hp-bisimulation, we omit them.
\end{proof}

\begin{theorem}[Completeness of APRTC with silent step and guarded linear recursion]\label{CAPRTCTAU}
Let $p$ and $q$ be closed APRTC with silent step and guarded linear recursion terms, then,
\begin{enumerate}
  \item if $p\approx_{rbs}^{fr} q$ then $p=q$;
  \item if $p\approx_{rbp}^{fr} q$ then $p=q$;
  \item if $p\approx_{rbhp}^{fr} q$ then $p=q$.
\end{enumerate}
\end{theorem}

\begin{proof}
Firstly, by the elimination theorem of APRTC with silent step and guarded linear recursion (see Theorem \ref{ETTau}), we know that each process term in APRTC with silent step and guarded linear recursion is equal to a process term $\langle X_1|E\rangle$ with $E$ a guarded linear recursive specification.

It remains to prove the following cases.

(1) If $\langle X_1|E_1\rangle \approx_{rbs}^{fr} \langle Y_1|E_2\rangle$ for guarded linear recursive specification $E_1$ and $E_2$, then $\langle X_1|E_1\rangle = \langle Y_1|E_2\rangle$.

Firstly, the recursive equation $W=\tau+\cdots+\tau$ with $W\nequiv X_1$ in $E_1$ and $E_2$, can be removed, and the corresponding summands $aW$ are replaced by $a$, to get $E_1'$ and $E_2'$, by use of the axioms RDP, $A3$ and $B1$, $RB1$, and $\langle X|E_1\rangle = \langle X|E_1'\rangle$, $\langle Y|E_2\rangle = \langle Y|E_2'\rangle$.

Let $E_1$ consists of recursive equations $X=t_X$ for $X\in \mathcal{X}$ and $E_2$
consists of recursion equations $Y=t_Y$ for $Y\in\mathcal{Y}$, and are not the form $\tau+\cdots+\tau$. Let the guarded linear recursive specification $E$ consists of recursion equations $Z_{XY}=t_{XY}$, and $\langle X|E_1\rangle\approx_{rbs}^{fr}\langle Y|E_2\rangle$, and $t_{XY}$ consists of the following summands:

\begin{enumerate}
  \item $t_{XY}$ contains a summand $(a_1\parallel\cdots\parallel a_m)Z_{X'Y'}$ iff $t_X$ contains the summand $(a_1\parallel\cdots\parallel a_m)X'$ and $t_Y$ contains the summand $(a_1\parallel\cdots\parallel a_m)Y'$ such that $\langle X'|E_1\rangle\approx_{rbs}^{fr}\langle Y'|E_2\rangle$;
  \item $t_{XY}$ contains a summand $Z_{X'Y'}(a_1[m]\parallel\cdots\parallel a_m[m])$ iff $t_X$ contains the summand $X'(a_1[m]\parallel\cdots\parallel a_m[m])$ and $t_Y$ contains the summand $Y'(a_1[m]\parallel\cdots\parallel a_m[m])$ such that $\langle X'|E_1\rangle\approx_{rbs}^{fr}\langle Y'|E_2\rangle$;
  \item $t_{XY}$ contains a summand $b_1\parallel\cdots\parallel b_n$ iff $t_X$ contains the summand $b_1\parallel\cdots\parallel b_n$ and $t_Y$ contains the summand $b_1\parallel\cdots\parallel b_n$;
  \item $t_{XY}$ contains a summand $b_1[n]\parallel\cdots\parallel b_n[n]$ iff $t_X$ contains the summand $b_1[n]\parallel\cdots\parallel b_n[n]$ and $t_Y$ contains the summand $b_1[n]\parallel\cdots\parallel b_n[n]$;
  \item $t_{XY}$ contains a summand $\tau Z_{X'Y}$ iff $XY\nequiv X_1Y_1$, $t_X$ contains the summand $\tau X'$, and $\langle X'|E_1\rangle\approx_{rbs}^{fr}\langle Y|E_2\rangle$;
  \item $t_{XY}$ contains a summand $Z_{X'Y}\tau$ iff $XY\nequiv X_1Y_1$, $t_X$ contains the summand $X'\tau$, and $\langle X'|E_1\rangle\approx_{rbs}^{fr}\langle Y|E_2\rangle$;
  \item $t_{XY}$ contains a summand $\tau Z_{XY'}$ iff $XY\nequiv X_1Y_1$, $t_Y$ contains the summand $\tau Y'$, and $\langle X|E_1\rangle\approx_{rbs}^{fr}\langle Y'|E_2\rangle$;
  \item $t_{XY}$ contains a summand $Z_{XY'}\tau$ iff $XY\nequiv X_1Y_1$, $t_Y$ contains the summand $Y'\tau$, and $\langle X|E_1\rangle\approx_{rbs}^{fr}\langle Y'|E_2\rangle$.
\end{enumerate}

Since $E_1$ and $E_2$ are guarded, $E$ is guarded. Constructing the process term $u_{XY}$ consist of the following summands:

\begin{enumerate}
  \item $u_{XY}$ contains a summand $(a_1\parallel\cdots\parallel a_m)\langle X'|E_1\rangle$ iff $t_X$ contains the summand $(a_1\parallel\cdots\parallel a_m)X'$ and $t_Y$ contains the summand $(a_1\parallel\cdots\parallel a_m)Y'$ such that $\langle X'|E_1\rangle\approx_{rbs}^{fr}\langle Y'|E_2\rangle$;
  \item $u_{XY}$ contains a summand $\langle X'|E_1\rangle(a_1[m]\parallel\cdots\parallel a_m[m])$ iff $t_X$ contains the summand $X'(a_1[m]\parallel\cdots\parallel a_m[m])$ and $t_Y$ contains the summand $Y'(a_1[m]\parallel\cdots\parallel a_m[m])$ such that $\langle X'|E_1\rangle\approx_{rbs}^{fr}\langle Y'|E_2\rangle$;
  \item $u_{XY}$ contains a summand $b_1\parallel\cdots\parallel b_n$ iff $t_X$ contains the summand $b_1\parallel\cdots\parallel b_n$ and $t_Y$ contains the summand $b_1\parallel\cdots\parallel b_n$;
  \item $u_{XY}$ contains a summand $b_1[n]\parallel\cdots\parallel b_n[n]$ iff $t_X$ contains the summand $b_1[n]\parallel\cdots\parallel b_n[n]$ and $t_Y$ contains the summand $b_1[n]\parallel\cdots\parallel b_n[n]$;
  \item $u_{XY}$ contains a summand $\tau \langle X'|E_1\rangle$ iff $XY\nequiv X_1Y_1$, $t_X$ contains the summand $\tau X'$, and $\langle X'|E_1\rangle\approx_{rbs}^{fr}\langle Y|E_2\rangle$;
  \item $u_{XY}$ contains a summand $\langle X'|E_1\rangle\tau$ iff $XY\nequiv X_1Y_1$, $t_X$ contains the summand $X'\tau$, and $\langle X'|E_1\rangle\approx_{rbs}^{fr}\langle Y|E_2\rangle$.
\end{enumerate}

Let the process term $s_{XY}$ be defined as follows:

\begin{enumerate}
  \item $s_{XY}\triangleq\tau\langle X|E_1\rangle + u_{XY}$ iff $XY\nequiv X_1Y_1$, $t_Y$ contains the summand $\tau Y'$, and $\langle X|E_1\rangle\approx_{rbs}^{fr}\langle Y'|E_2\rangle$;
  \item $s_{XY}\triangleq\langle X|E_1\rangle\tau + u_{XY}$ iff $XY\nequiv X_1Y_1$, $t_Y$ contains the summand $Y'\tau$, and $\langle X|E_1\rangle\approx_{rbs}^{fr}\langle Y'|E_2\rangle$;
  \item $s_{XY}\triangleq\langle X|E_1\rangle$, otherwise.
\end{enumerate}

So, $\langle X|E_1\rangle=\langle X|E_1\rangle+u_{XY}$, and $(a_1\parallel\cdots\parallel a_m)(\tau\langle X|E_1\rangle+u_{XY})=(a_1\parallel\cdots\parallel a_m)((\tau\langle X|E_1\rangle+u_{XY})+u_{XY})=(a_1\parallel\cdots\parallel a_m)(\langle X|E_1\rangle+u_{XY})=(a_1\parallel\cdots\parallel a_m)\langle X|E_1\rangle$, or $(\langle X|E_1\rangle\tau+u_{XY})(a_1[m]\parallel\cdots\parallel a_m[m])=((\langle X|E_1\rangle\tau+u_{XY})+u_{XY})(a_1[m]\parallel\cdots\parallel a_m[m])=(\langle X|E_1\rangle+u_{XY})(a_1[m]\parallel\cdots\parallel a_m[m])=\langle X|E_1\rangle(a_1[m]\parallel\cdots\parallel a_m[m])$, hence, $s_{XY}(a_1\parallel\cdots\parallel a_m)=(a_1[m]\parallel\cdots\parallel a_m[m])\langle X|E_1\rangle$.

Let $\sigma$ map recursion variable $X$ in $E_1$ to $\langle X|E_1\rangle$, and let $\psi$ map recursion variable $Z_{XY}$ in $E$ to $s_{XY}$. It is sufficient to prove $s_{XY}=\psi(t_{XY})$ for recursion variables $Z_{XY}$ in $E$. Either $XY\equiv X_1Y_1$ or $XY\nequiv X_1Y_1$, we all can get $s_{XY}=\psi(t_{XY})$. So, $s_{XY}=\langle Z_{XY}|E\rangle$ for recursive variables $Z_{XY}$ in $E$ is a solution for $E$. Then by RSP, particularly, $\langle X_1|E_1\rangle=\langle Z_{X_1Y_1}|E\rangle$. Similarly, we can obtain $\langle Y_1|E_2\rangle=\langle Z_{X_1Y_1}|E\rangle$. Finally, $\langle X_1|E_1\rangle=\langle Z_{X_1Y_1}|E\rangle=\langle Y_1|E_2\rangle$, as desired.

(2) If $\langle X_1|E_1\rangle \approx_{rbp}^{fr} \langle Y_1|E_2\rangle$ for guarded linear recursive specification $E_1$ and $E_2$, then $\langle X_1|E_1\rangle = \langle Y_1|E_2\rangle$.

It can be proven similarly to (1), we omit it.

(3) If $\langle X_1|E_1\rangle \approx_{rbhb} \langle Y_1|E_2\rangle$ for guarded linear recursive specification $E_1$ and $E_2$, then $\langle X_1|E_1\rangle = \langle Y_1|E_2\rangle$.

It can be proven similarly to (1), we omit it.
\end{proof}

\subsection{Abstraction}

The unary abstraction operator $\tau_I$ ($I\subseteq \mathbb{E}$) renames all atomic events in $I$ into $\tau$. APRTC with silent step and abstraction operator is called $APRTC_{\tau}$. The transition rules of operator $\tau_I$ are shown in Table \ref{TRForAbstraction}.

\begin{center}
    \begin{table}
        $$\frac{x\xrightarrow{e}\surd}{\tau_I(x)\xrightarrow{e}\surd}\quad e\notin I
        \quad\quad\frac{x\xrightarrow{e}x'}{\tau_I(x)\xrightarrow{e}\tau_I(x')}\quad e\notin I$$

        $$\frac{x\xrightarrow{e}\surd}{\tau_I(x)\xrightarrow{\tau}\surd}\quad e\in I
        \quad\quad\frac{x\xrightarrow{e}x'}{\tau_I(x)\xrightarrow{\tau}\tau_I(x')}\quad e\in I$$

        $$\frac{x\xtworightarrow{e[m]}e}{\tau_I(x)\xtworightarrow{e[m]}e}\quad e[m]\notin I
        \quad\quad\frac{x\xtworightarrow{e[m]}x'}{\tau_I(x)\xtworightarrow{e[m]}\tau_I(x')}\quad e[m]\notin I$$

        $$\frac{x\xtworightarrow{e[m]}\surd}{\tau_I(x)\xtworightarrow{\tau}\surd}\quad e[m]\in I
        \quad\quad\frac{x\xtworightarrow{e[m]}x'}{\tau_I(x)\xtworightarrow{\tau}\tau_I(x')}\quad e[m]\in I$$
        \caption{Transition rule of the abstraction operator}
        \label{TRForAbstraction}
    \end{table}
\end{center}

\begin{theorem}[Conservitivity of $APRTC_{\tau}$]
$APRTC_{\tau}$ is a conservative extension of APRTC with silent step.
\end{theorem}

\begin{proof}
Since the transition rules of APRTC with silent step are source-dependent, and the transition rules for abstraction operator in Table \ref{TRForAbstraction}contain only a fresh operator $\tau_I$ in their source, so the transition rules of $APRTC_{\tau}$ is a conservative extension of those of $RAPTC$ with silent step.
\end{proof}

\begin{theorem}[Congruence theorem of $APRTC_{\tau}$]
Rooted branching FR truly concurrent bisimulation equivalences $\approx_{rbp}^{fr}$, $\approx_{rbs}^{fr}$ and $\approx_{rbhp}^{fr}$ are all congruences with respect to $APRTC_{\tau}$.
\end{theorem}

\begin{proof}

(1) Case rooted branching FR pomset bisimulation equivalence $\approx_{rbp}^{fr}$.

Let $x$ and $y$ be $APRTC_{\tau}$ processes, and $x\approx_{rbp}^{fr} y$, it is sufficient to prove that $\tau_I(x)\approx_{rbp}^{fr} \tau_I(y)$.

By the transition rules for operator $\tau_I$ in Table \ref{TRForAbstraction}, we can get

$$\tau_I(x)\xrightarrow{X} X[\mathcal{K}] (X\nsubseteq I) \quad \tau_I(y)\xrightarrow{Y} Y[\mathcal{J}] (Y\nsubseteq I)$$

$$\tau_I(x)\xtworightarrow{X[\mathcal{K}]} X (X\nsubseteq I) \quad \tau_I(y)\xtworightarrow{Y[\mathcal{J}]} Y (Y\nsubseteq I)$$

with $X\subseteq x$, $Y\subseteq y$, and $X\sim Y$.

Or, we can get

$$\tau_I(x)\xrightarrow{X} \tau_I(x') (X\nsubseteq I) \quad \tau_I(y)\xrightarrow{Y} \tau_I(y') (Y\nsubseteq I)$$

$$\tau_I(x)\xtworightarrow{X[\mathcal{K}]} \tau_I(x') (X\nsubseteq I) \quad \tau_I(y)\xtworightarrow{Y[\mathcal{J}]} \tau_I(y') (Y\nsubseteq I)$$

with $X\subseteq x$, $Y\subseteq y$, and $X\sim Y$ and the hypothesis $\tau_I(x')\approx_{rbp}^{fr}\tau_I(y')$.

Or, we can get

$$\tau_I(x)\xrightarrow{\tau^*} \surd (X\subseteq I) \quad \tau_I(y)\xrightarrow{\tau^*} \surd (Y\subseteq I)$$

$$\tau_I(x)\xtworightarrow{\tau^*} \surd (X\subseteq I) \quad \tau_I(y)\xtworightarrow{\tau^*} \surd (Y\subseteq I)$$

with $X\subseteq x$, $Y\subseteq y$, and $X\sim Y$.

Or, we can get

$$\tau_I(x)\xrightarrow{\tau^*} \tau_I(x') (X\subseteq I) \quad \tau_I(y)\xrightarrow{\tau^*} \tau_I(y') (Y\subseteq I)$$

$$\tau_I(x)\xtworightarrow{\tau^*} \tau_I(x') (X\subseteq I) \quad \tau_I(y)\xtworightarrow{\tau^*} \tau_I(y') (Y\subseteq I)$$

with $X\subseteq x$, $Y\subseteq y$, and $X\sim Y$ and the hypothesis $\tau_I(x')\approx_{rbp}^{fr}\tau_I(y')$.

So, we get $\tau_I(x)\approx_{rbp}^{fr} \tau_I(y)$, as desired

(2) The cases of rooted branching FR step bisimulation $\approx_{rbs}^{fr}$, rooted branching FR hp-bisimulation $\approx_{rbhp}^{fr}$ can be proven similarly, we omit them.
\end{proof}

We design the axioms for the abstraction operator $\tau_I$ in Table \ref{AxiomsForAbstraction}.

\begin{center}
\begin{table}
  \begin{tabular}{@{}ll@{}}
\hline No. &Axiom\\
  $TI1$ & $e\notin I\quad \tau_I(e)=e$\\
  $RTI1$ & $e[m]\notin I\quad \tau_I(e[m])=e[m]$\\
  $TI2$ & $e\in I\quad \tau_I(e)=\tau$\\
  $RTI2$ & $e[m]\in I\quad \tau_I(e[m])=\tau$\\
  $TI3$ & $\tau_I(\delta)=\delta$\\
  $TI4$ & $\tau_I(x+y)=\tau_I(x)+\tau_I(y)$\\
  $TI5$ & $\tau_I(x\cdot y)=\tau_I(x)\cdot\tau_I(y)$\\
  $TI6$ & $\tau_I(x\parallel y)=\tau_I(x)\parallel\tau_I(y)$\\
\end{tabular}
\caption{Axioms of abstraction operator}
\label{AxiomsForAbstraction}
\end{table}
\end{center}

\begin{theorem}[Soundness of $APRTC_{\tau}$ with guarded linear recursion]\label{SAPRTCABS}
Let $x$ and $y$ be $APRTC_{\tau}$ with guarded linear recursion terms. If $APRTC_{\tau}$ with guarded linear recursion $\vdash x=y$, then
\begin{enumerate}
  \item $x\approx_{rbs}^{fr} y$;
  \item $x\approx_{rbp}^{fr} y$;
  \item $x\approx_{rbhp}^{fr} y$.
\end{enumerate}
\end{theorem}

\begin{proof}
(1) Soundness of $APRTC_{\tau}$ with guarded linear recursion with respect to rooted branching FR step bisimulation $\approx_{rbs}^{fr}$.

Since rooted branching FR step bisimulation $\approx_{rbs}^{fr}$ is both an equivalent and a congruent relation with respect to $APRTC_{\tau}$ with guarded linear recursion, we only need to check if each axiom in Table \ref{AxiomsForAbstraction} is sound modulo rooted branching FR step bisimulation equivalence.

The proof is similar to the proof of soundness of APRTC with silent step and guarded linear recursion, we omit them.

(2) Soundness of $APRTC_{\tau}$ with guarded linear recursion with respect to rooted branching FR pomset bisimulation $\approx_{rbp}^{fr}$.

Since rooted branching FR pomset bisimulation $\approx_{rbp}^{fr}$ is both an equivalent and a congruent relation with respect to $APRTC_{\tau}$ with guarded linear recursion, we only need to check if each axiom in Table \ref{AxiomsForAbstraction} is sound modulo rooted branching FR pomset bisimulation $\approx_{rbp}^{fr}$.

From the definition of rooted branching FR pomset bisimulation $\approx_{rbp}^{fr}$ (see Definition \ref{FRRBPSB}), we know that rooted branching FR pomset bisimulation $\approx_{rbp}^{fr}$ is defined by weak pomset transitions, which are labeled by pomsets with $\tau$. In a weak pomset transition, the events in the pomset are either within causality relations (defined by $\cdot$) or in concurrency (implicitly defined by $\cdot$ and $+$, and explicitly defined by $\between$), of course, they are pairwise consistent (without conflicts). In (1), we have already proven the case that all events are pairwise concurrent, so, we only need to prove the case of events in causality. Without loss of generality, we take a pomset of $P=\{e_1,e_2:e_1\cdot e_2\}$. Then the weak pomset transition labeled by the above $P$ is just composed of one single event transition labeled by $e_1$ succeeded by another single event transition labeled by $e_2$, that is, $\xRightarrow{P}=\xRightarrow{e_1}\xRightarrow{e_2}$  or $\xTworightarrow{P}=\xTworightarrow{e_2}\xTworightarrow{e_1}$.

Similarly to the proof of soundness of $APRTC_{\tau}$ with guarded linear recursion modulo rooted branching FR step bisimulation $\approx_{rbs}^{fr}$ (1), we can prove that each axiom in Table \ref{AxiomsForAbstraction} is sound modulo rooted branching FR pomset bisimulation $\approx_{rbp}^{fr}$, we omit them.

(3) Soundness of $APRTC_{\tau}$ with guarded linear recursion with respect to rooted branching FR hp-bisimulation $\approx_{rbhp}^{fr}$.

Since rooted branching FR hp-bisimulation $\approx_{rbhp}^{fr}$ is both an equivalent and a congruent relation with respect to $APRTC_{\tau}$ with guarded linear recursion, we only need to check if each axiom in Table \ref{AxiomsForAbstraction} is sound modulo rooted branching FR hp-bisimulation $\approx_{rbhp}^{fr}$.

From the definition of rooted branching FR hp-bisimulation $\approx_{rbhp}^{fr}$ (see Definition \ref{FRRBHHPB}), we know that rooted branching FR hp-bisimulation $\approx_{rbhp}^{fr}$ is defined on the weakly posetal product $(C_1,f,C_2),f:\hat{C_1}\rightarrow \hat{C_2}\textrm{ isomorphism}$. Two process terms $s$ related to $C_1$ and $t$ related to $C_2$, and $f:\hat{C_1}\rightarrow \hat{C_2}\textrm{ isomorphism}$. Initially, $(C_1,f,C_2)=(\emptyset,\emptyset,\emptyset)$, and $(\emptyset,\emptyset,\emptyset)\in\approx_{rbhp}^{fr}$. When $s\xrightarrow{e}s'$ ($C_1\xrightarrow{e}C_1'$), there will be $t\xRightarrow{e}t'$ ($C_2\xRightarrow{e}C_2'$), and we define $f'=f[e\mapsto e]$. And when $s\xTworightarrow{e[m]}s'$ ($C_1\xTworightarrow{e[m]}C_1'$), there will be $t\xTworightarrow{e[m]}t'$ ($C_2\xTworightarrow{e[m]}C_2'$), and we define $f'=f[e[m]\mapsto e[m]$. Then, if $(C_1,f,C_2)\in\approx_{rbhp}^{fr}$, then $(C_1',f',C_2')\in\approx_{rbhp}^{fr}$.

Similarly to the proof of soundness of $APRTC_{\tau}$ with guarded linear recursion modulo rooted branching FR pomset bisimulation equivalence (2), we can prove that each axiom in Table \ref{AxiomsForAbstraction} is sound modulo rooted branching FR hp-bisimulation equivalence, we just need additionally to check the above conditions on rooted branching FR hp-bisimulation, we omit them.
\end{proof}

Though $\tau$-loops are prohibited in guarded linear recursive specifications (see Definition \ref{GLRS}) in specifiable way, they can be constructed using the abstraction operator, for example, there exist $\tau$-loops in the process term $\tau_{\{a\}}(\langle X|X=aX\rangle)$. To avoid $\tau$-loops caused by $\tau_I$ and ensure fairness, the concept of cluster and CFAR (Cluster Fair Abstraction Rule) are still valid in true concurrency, we introduce them below.

\begin{definition}[Cluster]\label{CLUSTER}
Let $E$ be a guarded linear recursive specification, and $I\subseteq \mathbb{E}$. Two recursion variable $X$ and $Y$ in $E$ are in the same cluster for $I$ iff there exist sequences of transitions $\langle X|E\rangle\xrightarrow{\{b_{11},\cdots, b_{1i}\}}\cdots\xrightarrow{\{b_{m1},\cdots, b_{mi}\}}\langle Y|E\rangle$ and $\langle Y|E\rangle\xrightarrow{\{c_{11},\cdots, c_{1j}\}}\cdots\xrightarrow{\{c_{n1},\cdots, c_{nj}\}}\langle X|E\rangle$, or $\langle X|E\rangle\xtworightarrow{\{b_{11}[m],\cdots, b_{1i}[m]\}}\cdots\xtworightarrow{\{b_{m1}[m],\cdots, b_{mi}[m]\}}\langle Y|E\rangle$ and $\langle Y|E\rangle\xtworightarrow{\{c_{11}[n],\cdots, c_{1j}[n]\}}\cdots\xtworightarrow{\{c_{n1}[n],\cdots, c_{nj}[n]\}}\langle X|E\rangle$, where $b_{11},\cdots,b_{mi},c_{11},\cdots,c_{nj}, b_{11}[m],\cdots,b_{mi}[m],c_{11}[n],\cdots,c_{nj}[n]\in I\cup\{\tau\}$.

$a_1\parallel\cdots\parallel a_k$, or $(a_1\parallel\cdots\parallel a_k) X$, or $a_1[m]\parallel\cdots\parallel a_k[m]$, or $X (a_1[m]\parallel\cdots\parallel a_k[m])$ is an exit for the cluster $C$ iff: (1) $a_1\parallel\cdots\parallel a_k$, or $(a_1\parallel\cdots\parallel a_k) X$, or $a_1[m]\parallel\cdots\parallel a_k[m]$, or $X (a_1[m]\parallel\cdots\parallel a_k[m])$ is a summand at the right-hand side of the recursive equation for a recursion variable in $C$, and (2) in the case of $(a_1\parallel\cdots\parallel a_k) X$, and $X (a_1[m]\parallel\cdots\parallel a_k[m])$ either $a_l, a_l[m]\notin I\cup\{\tau\}(l\in\{1,2,\cdots,k\})$ or $X\notin C$.
\end{definition}

\begin{center}
\begin{table}
  \begin{tabular}{@{}ll@{}}
\hline No. &Axiom\\
  CFAR & If $X$ is in a cluster for $I$ with exits \\
           & $\{(a_{11}\parallel\cdots\parallel a_{1i})Y_1,\cdots,(a_{m1}\parallel\cdots\parallel a_{mi})Y_m, b_{11}\parallel\cdots\parallel b_{1j},\cdots,b_{n1}\parallel\cdots\parallel b_{nj}\}$, \\
           & then $\tau\cdot\tau_I(\langle X|E\rangle)=$\\
           & $\tau\cdot\tau_I((a_{11}\parallel\cdots\parallel a_{1i})\langle Y_1|E\rangle+\cdots+(a_{m1}\parallel\cdots\parallel a_{mi})\langle Y_m|E\rangle+b_{11}\parallel\cdots\parallel b_{1j}+\cdots+b_{n1}\parallel\cdots\parallel b_{nj})$\\
           & Or exists,\\
           & $\{Y_1(a_{11}[m]\parallel\cdots\parallel a_{1i}[m1]),\cdots,Y_m(a_{m1}[mm]\parallel\cdots\parallel a_{mi}[mm]), b_{11}[n1]parallel\cdots\parallel b_{1j}[n1],\cdots,b_{n1}[nn]\parallel\cdots\parallel b_{nj}[nn]\}$, \\
           & then $\tau_I(\langle X|E\rangle)\cdot\tau=$\\
           & $\tau_I(\langle Y_1|E\rangle(a_{11}[m1]\parallel\cdots\parallel a_{1i}[m1])+\cdots+\langle Y_m|E\rangle(a_{m1}[mm]\parallel\cdots\parallel a_{mi}[mm])+b_{11}[n1]\parallel\cdots\parallel b_{1j}[n1]+\cdots+b_{n1}[nn]\parallel\cdots\parallel b_{nj}[nn])\cdot\tau$\\
\end{tabular}
\caption{Cluster fair abstraction rule}
\label{CFAR}
\end{table}
\end{center}

\begin{theorem}[Soundness of CFAR]\label{SCFAR}
CFAR is sound modulo rooted branching FR truly concurrent bisimulation equivalences $\approx_{rbs}^{fr}$, $\approx_{rbp}^{fr}$ and $\approx_{rbhp}^{fr}$.
\end{theorem}

\begin{proof}
(1) Soundness of CFAR with respect to rooted branching FR step bisimulation $\approx_{rbs}^{fr}$.

Let $X$ be in a cluster for $I$ with exits $\{(a_{11}\parallel\cdots\parallel a_{1i})Y_1,\cdots,(a_{m1}\parallel\cdots\parallel a_{mi})Y_m,b_{11}\parallel\cdots\parallel b_{1j},\cdots,b_{n1}\parallel\cdots\parallel b_{nj}\}$ and $\{Y_1(a_{11}[m1]\parallel\cdots\parallel a_{1i}[m1]),\cdots,Y_m(a_{m1}[mm]\parallel\cdots\parallel a_{mi}[mm]),b_{11}[n1]\parallel\cdots\parallel b_{1j}[n1],\cdots,b_{n1}[nn]\parallel\cdots\parallel b_{nj}[nn]\}$. Then $\langle X|E\rangle$ can execute a string of atomic events from $I\cup\{\tau\}$ inside the cluster of $X$, followed by an exit $(a_{i'1}\parallel\cdots\parallel a_{i'i})Y_{i'}$ for $i'\in\{1,\cdots,m\}$ or $b_{j'1}\parallel\cdots\parallel b_{j'j}$ for $j'\in\{1,\cdots,n\}$, or $Y_{i'}(a_{i'1}[m1]\parallel\cdots\parallel a_{i'i}[mi'])$ for $i'\in\{1,\cdots,m\}$ or $b_{j'1}[n1]\parallel\cdots\parallel b_{j'j}[nj']$ for $j'\in\{1,\cdots,n\}$. Hence, $\tau_I(\langle X|E\rangle)$ can execute a string of $\tau^*$ inside the cluster of $X$, followed by an exit $\tau_I((a_{i'1}\parallel\cdots\parallel a_{i'i})\langle Y_{i'}|E\rangle)$ for $i'\in\{1,\cdots,m\}$ or $\tau_I(b_{j'1}\parallel\cdots\parallel b_{j'j})$ for $j'\in\{1,\cdots,n\}$, or $\tau_I(\langle Y_{i'}|E\rangle(a_{i'1}[m1]\parallel\cdots\parallel a_{i'i}[mi']))$ for $i'\in\{1,\cdots,m\}$ or $\tau_I(b_{j'1}[n1]\parallel\cdots\parallel b_{j'j}[nj'])$ for $j'\in\{1,\cdots,n\}$. And these $\tau^*$ are non-initial in $\tau\tau_I(\langle X|E\rangle)$ and $\tau_I(\langle X|E\rangle)\tau$, so they are truly silent by the axiom $B1$ and $RB1$, we obtain $\tau\tau_I(\langle X|E\rangle)\approx_{rbs}^{fr}\tau\cdot\tau_I((a_{11}\parallel\cdots\parallel a_{1i})\langle Y_1|E\rangle+\cdots+(a_{m1}\parallel\cdots\parallel a_{mi})\langle Y_m|E\rangle+b_{11}\parallel\cdots\parallel b_{1j}+\cdots+b_{n1}\parallel\cdots\parallel b_{nj})$, and $\tau_I(\langle X|E\rangle)\tau\approx_{rbs}^{fr}\tau_I(\langle Y_1|E\rangle(a_{11}[m1]\parallel\cdots\parallel a_{1i}[m1])+\cdots+\langle Y_m|E\rangle(a_{m1}[mm]\parallel\cdots\parallel a_{mi}[mm])+b_{11}[n1]\parallel\cdots\parallel b_{1j}[n1]+\cdots+b_{n1}[nn]\parallel\cdots\parallel b_{nj}[nn])\cdot\tau$ as desired.

(2) Soundness of CFAR with respect to rooted branching FR pomset bisimulation $\approx_{rbp}^{fr}$.

From the definition of rooted branching FR pomset bisimulation $\approx_{rbp}^{fr}$ (see Definition \ref{FRRBPSB}), we know that rooted branching FR pomset bisimulation $\approx_{rbp}^{fr}$ is defined by weak pomset transitions, which are labeled by pomsets with $\tau$. In a weak pomset transition, the events in the pomset are either within causality relations (defined by $\cdot$) or in concurrency (implicitly defined by $\cdot$ and $+$, and explicitly defined by $\between$), of course, they are pairwise consistent (without conflicts). In (1), we have already proven the case that all events are pairwise concurrent, so, we only need to prove the case of events in causality. Without loss of generality, we take a pomset of $P=\{e_1,e_2:e_1\cdot e_2\}$. Then the weak pomset transition labeled by the above $P$ is just composed of one single event transition labeled by $e_1$ succeeded by another single event transition labeled by $e_2$, that is, $\xRightarrow{P}=\xRightarrow{e_1}\xRightarrow{e_2}$ or $\xTworightarrow{P}=\xTworightarrow{e_2}\xTworightarrow{e_1}$.

Similarly to the proof of soundness of CFAR modulo rooted branching FR step bisimulation $\approx_{rbs}^{fr}$ (1), we can prove that CFAR in Table \ref{CFAR} is sound modulo rooted branching FR pomset bisimulation $\approx_{rbp}^{fr}$, we omit them.

(3) Soundness of CFAR with respect to rooted branching FR hp-bisimulation $\approx_{rbhp}^{fr}$.

From the definition of rooted branching FR hp-bisimulation $\approx_{rbhp}^{fr}$ (see Definition \ref{FRRBHHPB}), we know that rooted branching FR hp-bisimulation $\approx_{rbhp}^{fr}$ is defined on the weakly posetal product $(C_1,f,C_2),f:\hat{C_1}\rightarrow \hat{C_2}\textrm{ isomorphism}$. Two process terms $s$ related to $C_1$ and $t$ related to $C_2$, and $f:\hat{C_1}\rightarrow \hat{C_2}\textrm{ isomorphism}$. Initially, $(C_1,f,C_2)=(\emptyset,\emptyset,\emptyset)$, and $(\emptyset,\emptyset,\emptyset)\in\approx_{rbhp}^{fr}$. When $s\xrightarrow{e}s'$ ($C_1\xrightarrow{e}C_1'$), there will be $t\xRightarrow{e}t'$ ($C_2\xRightarrow{e}C_2'$), and we define $f'=f[e\mapsto e]$. And when $s\xTworightarrow{e[m]}s'$ ($C_1\xTworightarrow{e[m]}C_1'$), there will be $t\xTworightarrow{e[m]}t'$ ($C_2\xTworightarrow{e[m]}C_2'$), and we define $f'=f[e[m]\mapsto e[m]$. Then, if $(C_1,f,C_2)\in\approx_{rbhp}^{fr}$, then $(C_1',f',C_2')\in\approx_{rbhp}^{fr}$.

Similarly to the proof of soundness of CFAR modulo rooted branching FR pomset bisimulation equivalence (2), we can prove that CFAR in Table \ref{CFAR} is sound modulo rooted branching FR hp-bisimulation equivalence, we just need additionally to check the above conditions on rooted branching FR hp-bisimulation, we omit them.
\end{proof}

\begin{theorem}[Completeness of $APRTC_{\tau}$ with guarded linear recursion and CFAR]\label{CCFAR}
Let $p$ and $q$ be closed $APRTC_{\tau}$ with guarded linear recursion and CFAR terms, then,
\begin{enumerate}
  \item if $p\approx_{rbs}^{fr} q$ then $p=q$;
  \item if $p\approx_{rbp}^{fr} q$ then $p=q$;
  \item if $p\approx_{rbhp}^{fr} q$ then $p=q$.
\end{enumerate}
\end{theorem}

\begin{proof}
(1) For the case of rooted branching FR step bisimulation, the proof is following.

Firstly, in the proof the Theorem \ref{CAPRTCTAU}, we know that each process term $p$ in APRTC with silent step and guarded linear recursion is equal to a process term $\langle X_1|E\rangle$ with $E$ a guarded linear recursive specification. And we prove if $\langle X_1|E_1\rangle\approx_{rbs}^{fr}\langle Y_1|E_2\rangle$, then $\langle X_1|E_1\rangle=\langle Y_1|E_2\rangle$

The only new case is $p\equiv\tau_I(q)$. Let $q=\langle X|E\rangle$ with $E$ a guarded linear recursive specification, so $p=\tau_I(\langle X|E\rangle)$. Then the collection of recursive variables in $E$ can be divided into its clusters $C_1,\cdots,C_N$ for $I$. Let

$$(a_{1i1}\parallel\cdots\parallel a_{k_{i1}i1}) Y_{i1}+\cdots+(a_{1im_i}\parallel\cdots\parallel a_{k_{im_i}im_i}) Y_{im_i}+b_{1i1}\parallel\cdots\parallel b_{l_{i1}i1}+\cdots+b_{1im_i}\parallel\cdots\parallel b_{l_{im_i}im_i}$$

or,

$$Y_{i1}(a_{1i1}[m1]\parallel\cdots\parallel a_{k_{i1}i1}[m1]) +\cdots+ Y_{im_i}(a_{1im_i}[mm]\parallel\cdots\parallel a_{k_{im_i}im_i}[mm])+b_{1i1}[n1]\parallel\cdots\parallel b_{l_{i1}i1}[n1]+\cdots+b_{1im_i}[nn]\parallel\cdots\parallel b_{l_{im_i}im_i[nn]}$$

be the conflict composition of exits for the cluster $C_i$, with $i\in\{1,\cdots,N\}$.

For $Z,Z'\in C_i$ with $i\in\{1,\cdots,N\}$, we define

$$s_Z\triangleq (\hat{a_{1i1}}\parallel\cdots\parallel \hat{a_{k_{i1}i1}}) \tau_I(\langle Y_{i1}|E\rangle)+\cdots+(\hat{a_{1im_i}}\parallel\cdots\parallel \hat{a_{k_{im_i}im_i}}) \tau_I(\langle Y_{im_i}|E\rangle)+\hat{b_{1i1}}\parallel\cdots\parallel \hat{b_{l_{i1}i1}}+\cdots+\hat{b_{1im_i}}\parallel\cdots\parallel \hat{b_{l_{im_i}im_i}}$$

and

$$s_Z'\triangleq \tau_I(\langle Y_{i1}|E\rangle)(\hat{a_{1i1}}[m1]\parallel\cdots\parallel \hat{a_{k_{i1}i1}}[m1])+\cdots+(\hat{a_{1im_i}}[mm]\parallel\cdots\parallel \hat{a_{k_{im_i}im_i}}[mm]) \tau_I(\langle Y_{im_i}|E\rangle)+\hat{b_{1i1}}[n1]\parallel\cdots\parallel \hat{b_{l_{i1}i1}}[n1]+\cdots+\hat{b_{1im_i}}[nn]\parallel\cdots\parallel \hat{b_{l_{im_i}im_i}}[nn]$$

For $Z,Z'\in C_i$ and $a_1,\cdots,a_j\in \mathbb{E}\cup\{\tau\}$ with $j\in\mathbb{N}$, we have

$(a_1\parallel\cdots\parallel a_j)\tau_I(\langle Z|E\rangle)$

$=(a_1\parallel\cdots\parallel a_j)\tau_I((a_{1i1}\parallel\cdots\parallel a_{k_{i1}i1}) \langle Y_{i1}|E\rangle+\cdots+(a_{1im_i}\parallel\cdots\parallel a_{k_{im_i}im_i}) \langle Y_{im_i}|E\rangle+b_{1i1}\parallel\cdots\parallel b_{l_{i1}i1}+\cdots+b_{1im_i}\parallel\cdots\parallel b_{l_{im_i}im_i})$

$=(a_1\parallel\cdots\parallel a_j)s_Z$

$\tau_I(\langle Z|E\rangle)(a_1[m]\parallel\cdots\parallel a_j[m])$

$=\tau_I(\langle Y_{i1}|E\rangle(a_{1i1}[m1]\parallel\cdots\parallel a_{k_{i1}i1}[m1]) +\cdots+ \langle Y_{im_i}|E\rangle(a_{1im_i}[mm]\parallel\cdots\parallel a_{k_{im_i}im_i}[mm])+b_{1i1}[n1]\parallel\cdots\parallel b_{l_{i1}i1}[n1]+\cdots+b_{1im_i}[nn]\parallel\cdots\parallel b_{l_{im_i}im_i}[nn])(a_1[m]\parallel\cdots\parallel a_j[m])$

$=(a_1\parallel\cdots\parallel a_j)s_Z'$

Let the linear recursive specification $F$ contain the same recursive variables as $E$, for $Z,Z'\in C_i$, $F$ contains the following recursive equation

$$Z=(\hat{a_{1i1}}\parallel\cdots\parallel \hat{a_{k_{i1}i1}}) Y_{i1}+\cdots+(\hat{a_{1im_i}}\parallel\cdots\parallel \hat{a_{k_{im_i}im_i}})  Y_{im_i}+\hat{b_{1i1}}\parallel\cdots\parallel \hat{b_{l_{i1}i1}}+\cdots+\hat{b_{1im_i}}\parallel\cdots\parallel \hat{b_{l_{im_i}im_i}}$$

Let the linear recursive specification $F'$ contain the same recursive variables as $E$, for $Z,Z'\in C_i$, $F$ contains the following recursive equation

$$Z'=Y_{i1}(\hat{a_{1i1}}[m1]\parallel\cdots\parallel \hat{a_{k_{i1}i1}}[m1]) +\cdots+  Y_{im_i}(\hat{a_{1im_i}}[mm]\parallel\cdots\parallel \hat{a_{k_{im_i}im_i}}[mm])+\hat{b_{1i1}}[n1]\parallel\cdots\parallel \hat{b_{l_{i1}i1}}[n1]+\cdots+\hat{b_{1im_i}}[nn]\parallel\cdots\parallel \hat{b_{l_{im_i}im_i}}[nn]$$

It is easy to see that there is no sequence of one or more $\tau$-transitions from $\langle Z|F\rangle$ and $\langle Z'|F'\rangle$ to itself, so $F$ and $F'$ is guarded.

For

$$s_Z=(\hat{a_{1i1}}\parallel\cdots\parallel \hat{a_{k_{i1}i1}}) Y_{i1}+\cdots+(\hat{a_{1im_i}}\parallel\cdots\parallel \hat{a_{k_{im_i}im_i}}) Y_{im_i}+\hat{b_{1i1}}\parallel\cdots\parallel \hat{b_{l_{i1}i1}}+\cdots+\hat{b_{1im_i}}\parallel\cdots\parallel \hat{b_{l_{im_i}im_i}}$$

is a solution for $F$. So, $(a_1\parallel\cdots\parallel a_j)\tau_I(\langle Z|E\rangle)=(a_1\parallel\cdots\parallel a_j)s_Z=(a_1\parallel\cdots\parallel a_j)\langle Z|F\rangle$.

So,

$$\langle Z|F\rangle=(\hat{a_{1i1}}\parallel\cdots\parallel \hat{a_{k_{i1}i1}}) \langle Y_{i1}|F\rangle+\cdots+(\hat{a_{1im_i}}\parallel\cdots\parallel \hat{a_{k_{im_i}im_i}}) \langle Y_{im_i}|F\rangle+\hat{b_{1i1}}\parallel\cdots\parallel \hat{b_{l_{i1}i1}}+\cdots+\hat{b_{1im_i}}\parallel\cdots\parallel \hat{b_{l_{im_i}im_i}}$$

For

$$s_Z'=Y_{i1}(\hat{a_{1i1}}[m1]\parallel\cdots\parallel \hat{a_{k_{i1}i1}}[m1]) +\cdots+  Y_{im_i}(\hat{a_{1im_i}}[mm]\parallel\cdots\parallel \hat{a_{k_{im_i}im_i}}[mm])+\hat{b_{1i1}}[n1]\parallel\cdots\parallel \hat{b_{l_{i1}i1}}[n1]+\cdots+\hat{b_{1im_i}}[nn]\parallel\cdots\parallel \hat{b_{l_{im_i}im_i}}[nn]$$

is a solution for $F'$. So, $\tau_I(\langle Z'|E\rangle)(a_1[m]\parallel\cdots\parallel a_j[m])=s_Z'(a_1[m]\parallel\cdots\parallel a_j[m])=\langle Z'|F'\rangle(a_1[m]\parallel\cdots\parallel a_j[m])$.

So,

$$\langle Z'|F'\rangle=\langle Y_{i1}|F\rangle(\hat{a_{1i1}}[m1]\parallel\cdots\parallel \hat{a_{k_{i1}i1}}[m1])+\cdots+\langle Y_{im_i}|F\rangle(\hat{a_{1im_i}}[mm]\parallel\cdots\parallel \hat{a_{k_{im_i}im_i}}[mm]) +\hat{b_{1i1}}[n1]\parallel\cdots\parallel \hat{b_{l_{i1}i1}}[n1]+\cdots+\hat{b_{1im_i}}[nn]\parallel\cdots\parallel \hat{b_{l_{im_i}im_i}}[nn]$$

Hence, $\tau_I(\langle X|E\rangle=\langle Z|F\rangle)$, as desired.

(2) For the case of rooted branching FR pomset bisimulation, it can be proven similarly to (1), we omit it.

(3) For the case of rooted branching FR hp-bisimulation, it can be proven similarly to (1), we omit it.
\end{proof}

\section{Conclusions}{\label{con}}

Based on our previous process algebra for concurrency APTC, we prove that it is reversible with a little modifications. The reversible algebra has four parts: Basic Algebra for Reversible True Concurrency (BARTC), Algebra for Parallelism in Reversible True Concurrency (APRTC), recursion and abstraction.

This work can be used to verify the behavior of computational systems in a reversible flavor.

\newpage

%

\label{lastpage}

\end{document}